\PassOptionsToPackage{unicode}{hyperref}
\PassOptionsToPackage{hyphens}{url}
\PassOptionsToPackage{dvipsnames,svgnames,x11names}{xcolor}

\documentclass[
12pt]{article}

\usepackage{amsmath,amsfonts,amssymb,amsthm}     
\usepackage[flushleft]{threeparttable} 
\usepackage{algorithm}
\usepackage{algpseudocode} 
\usepackage{siunitx}
\usepackage{multirow, booktabs}
\usepackage{authblk}
\usepackage{enumerate}
\usepackage{iftex}
\ifPDFTeX
\usepackage[T1]{fontenc}
\usepackage[utf8]{inputenc}
\usepackage{textcomp} 
\else 
\usepackage{unicode-math}
\defaultfontfeatures{Scale=MatchLowercase}
\defaultfontfeatures[\rmfamily]{Ligatures=TeX,Scale=1}
\fi
\usepackage{lmodern}
\ifPDFTeX\else  

\fi

\IfFileExists{upquote.sty}{\usepackage{upquote}}{}
\IfFileExists{microtype.sty}{
	\usepackage[]{microtype}
	\UseMicrotypeSet[protrusion]{basicmath} 
}{}

\usepackage{xcolor}
\makeatletter
\ifx\paragraph\undefined\else
\let\oldparagraph\paragraph
\renewcommand{\paragraph}{
	\@ifstar
	\xxxParagraphStar
	\xxxParagraphNoStar
}
\newcommand{\xxxParagraphStar}[1]{\oldparagraph*{#1}\mbox{}}
\newcommand{\xxxParagraphNoStar}[1]{\oldparagraph{#1}\mbox{}}
\fi
\ifx\subparagraph\undefined\else
\let\oldsubparagraph\subparagraph
\renewcommand{\subparagraph}{
	\@ifstar
	\xxxSubParagraphStar
	\xxxSubParagraphNoStar
}
\newcommand{\xxxSubParagraphStar}[1]{\oldsubparagraph*{#1}\mbox{}}
\newcommand{\xxxSubParagraphNoStar}[1]{\oldsubparagraph{#1}\mbox{}}
\fi
\makeatother

\usepackage{longtable,booktabs,array}
\usepackage{calc} 

\usepackage{etoolbox}
\makeatletter
\patchcmd\longtable{\par}{\if@noskipsec\mbox{}\fi\par}{}{}
\makeatother

\IfFileExists{footnotehyper.sty}{\usepackage{footnotehyper}}{\usepackage{footnote}}
\makesavenoteenv{longtable}
\usepackage{graphicx}
\makeatletter
\def\maxwidth{\ifdim\Gin@nat@width>\linewidth\linewidth\else\Gin@nat@width\fi}
\def\maxheight{\ifdim\Gin@nat@height>\textheight\textheight\else\Gin@nat@height\fi}
\makeatother

\setkeys{Gin}{width=\maxwidth,height=\maxheight,keepaspectratio}

\makeatletter
\def\fps@figure{htbp}
\makeatother

\makeatletter
\@ifpackageloaded{caption}{}{\usepackage{caption}}
\AtBeginDocument{
	\ifdefined\contentsname
	\renewcommand*\contentsname{Table of contents}
	\else
	\newcommand\contentsname{Table of contents}
	\fi
	\ifdefined\listfigurename
	\renewcommand*\listfigurename{List of Figures}
	\else
	\newcommand\listfigurename{List of Figures}
	\fi
	\ifdefined\listtablename
	\renewcommand*\listtablename{List of Tables}
	\else
	\newcommand\listtablename{List of Tables}
	\fi
	\ifdefined\figurename
	\renewcommand*\figurename{Figure}
	\else
	\newcommand\figurename{Figure}
	\fi
	\ifdefined\tablename
	\renewcommand*\tablename{Table}
	\else
	\newcommand\tablename{Table}
	\fi
}
\@ifpackageloaded{float}{}{\usepackage{float}}
\floatstyle{ruled}
\@ifundefined{c@chapter}{\newfloat{codelisting}{h}{lop}}{\newfloat{codelisting}{h}{lop}[chapter]}
\floatname{codelisting}{Listing}

\makeatother
\makeatletter
\@ifpackageloaded{caption}{}{\usepackage{caption}}
\@ifpackageloaded{subcaption}{}{\usepackage{subcaption}}
\makeatother

\ifLuaTeX
\usepackage{selnolig}  
\fi
\usepackage[]{natbib}
\usepackage{bookmark}  
\IfFileExists{xurl.sty}{\usepackage{xurl}}{} 
\hypersetup{
	pdftitle={Title},
	pdfauthor={Author 1; Author 2},
	pdfkeywords={3 to 6 keywords, that do not appear in the title},
	colorlinks=true,
	linkcolor={blue},
	filecolor={Maroon},
	citecolor={Blue},
	urlcolor={Blue},
	pdfcreator={LaTeX via pandoc}}

\usepackage{unicode-math}

\newcommand{\vecsym}[1]{\symbfit{#1}}

\newcommand{\matsym}[1]{\symbf{#1}}

\newcommand{\grsym}[1]{\symbfit{#1}}

\def\ba{\vecsym{a}} \def\bb{\vecsym{b}}  
\def\be{\vecsym{e}} \def\bfv{\vecsym{f}}  
   
\def\bmm{\vecsym{m}}   
 \def\br{\vecsym{r}} \def\bs{\vecsym{s}} 
\def\bu{\vecsym{u}} \def\bv{\vecsym{v}} \def\bw{\vecsym{w}} \def\bx{\vecsym{x}}
\def\by{\vecsym{y}}  \def\bxi{\vecsym{\xi}}

\def\bA{\matsym{A}} \def\bB{\matsym{B}} \def\bC{\matsym{C}} 
 \def\bF{\matsym{F}}  \def\bH{\matsym{H}}
\def\bI{\matsym{I}}   
\def\bM{\matsym{M}}   \def\bP{\matsym{P}}
 \def\bR{\matsym{R}}  \def\bT{\matsym{T}}
\def\bU{\matsym{U}} \def\bV{\matsym{V}} \def\bW{\matsym{W}} 
 \def\bZ{\matsym{Z}}

\def\balpha{\grsym{\alpha}}
\def\bbeta{\grsym{\beta}}
\def\bgamma{\grsym{\gamma}}   
\def\bdelta{\grsym{\delta}}   \def\bDelta{\grsym{\Delta}}

\def\btheta{\grsym{\theta}}   \def\bTheta{\grsym{\Theta}}

   \def\bSigma{\grsym{\Sigma}}

\def\bomega{\grsym{\omega}}   \def\bOmega{\grsym{\Omega}}

\def\PP{\mathbb{P}}

\def\E{\mathbb{E}}
\def\Var{\mathrm{Var}}
\def\Cov{\mathrm{Cov}}
\def\RR{\mathbb{R}}

\def\EE{\mathbb{E}}

\def\t{^{\top}}
\def\wt{\widetilde}
\def\wh{\widehat}

\theoremstyle{plain}
\newtheorem{theorem}{Theorem}

\newtheorem{proposition}{Proposition}

\newtheorem{lemma}[theorem]{Lemma}
\newtheorem{corollary}{Corollary}

\newtheorem{assumption}{Assumption}
\newtheorem{rem}{Remark}

\newcommand{\anon}{1}

\usepackage{hyperref}
\usepackage{xr-hyper}
\begin{document}

	\def\spacingset#1{\renewcommand{\baselinestretch}
		{#1}\small\normalsize} \spacingset{1}

	\if1\anon
	{
		\title{\bf Factor-Adjusted Multiple Testing for High-Dimensional Individual Mediation Effects}
		\author[1]{Chen Shi }
		\author[1]{Zhao Chen \thanks{Corresponding author: zchen\_fdu@fudan.edu.cn. }}
		\author[2]{Christina Dan Wang \thanks{Corresponding author: christina.wang@nyu.edu.}}
		\affil[1]{\small  School of Data Science, Fudan University}
		\affil[2]{\small  Business Division, New York University Shanghai}
		\maketitle
	} \fi
	
	\if0\anon
	{
		\bigskip
		\bigskip
		\bigskip
		\begin{center}
			{\LARGE\bf Factor-Adjusted Multiple Testing for High-Dimensional Individual Mediation Effects}
		\end{center}
		\medskip
	} \fi
	
	\bigskip
	\begin{abstract}
		Identifying individual mediators is a central goal of high-dimensional mediation analysis, yet pervasive dependence among mediators can invalidate standard debiased inference and lead to substantial false discovery rate (FDR) inflation. We propose a Factor-Adjusted
		Debiased Mediation Testing (FADMT) framework that enables large-scale inference for individual mediation effects with FDR control under complex dependence structures. Our approach posits an approximate factor structure on the unobserved errors of the mediator model, extracts common latent factors, and constructs decorrelated pseudo-mediators  for the subsequent inferential procedure. We establish the asymptotic normality of the debiased estimator and develop a multiple testing procedure with theoretical FDR control under mild
		high-dimensional conditions. By adjusting for latent factor induced dependence,
		FADMT also improves robustness to spurious associations driven by shared latent
		variation in observational studies. Extensive simulations demonstrate 
		the superior finite-sample performance across a wide range of correlation structures.
	Applications to TCGA-BRCA multi-omics data and to China’s stock connect study further illustrate the practical utility of the proposed method.
	\end{abstract}
	
	\noindent
	{\it Keywords:}  High-dimensional mediation analysis; High dimensional inference; Factor model; False discovery rate
	\vfill
	
	\newpage
	\spacingset{1.8}

	\section{Introduction}\label{sec-intro}
	Understanding the mechanisms through which an exposure affects an outcome is a central problem across many scientific disciplines. 
	Mediation analysis provides a principled framework for decomposing the total effect into a direct effect and an indirect effect transmitted through intermediate variables \citep{baron1986moderator,mackinnon2004confidence}.
	In modern genomics and multi-omics studies and increasingly in finance, hundreds or thousands of candidate mediators are routinely measured, making individual mediator discovery both scientifically important and statistically challenging.
	Two difficulties are particularly acute in high dimensions: mediators exhibit strong dependence driven by shared latent variation, and identifying active mediators necessitates rigorous simultaneous inference to maintain false discovery rate (FDR) control.
	
	A growing literature studies high-dimensional mediation, with much of it focusing
	on inference for the overall indirect effect, which aggregates contributions from all mediators \citep{huang2016hypothesis,zhou2020estimation,guo2022high,guo2023statistical,lin2023testing}. Although informative, overall indirect effects can mask important mechanistic
	signals when individual mediation effects cancel out due to opposing directions.
	This motivates methods for individual mediation effect discovery.
	
	Existing methods for testing individual mediation effects in high dimensions generally follow two paradigms: marginal modeling, which relies on simplifying independence assumptions and apply
	multiple testing to marginal regression coefficients
	\citep{dai2022multiple,liu2022large,du2023methods} and joint modeling, which employs high-dimensional inference or variable selection techniques such as screening, debiased Lasso, and adaptive Lasso \citep{zhang2016estimating, zhang2021mediation, derkach2019high, shuai2023mediation}. While these approaches offer modeling flexibility, their validity is severely compromised by pervasive dependence among mediators.
	
	Strong dependence among mediators presents two fundamental challenges. First, high-dimensional inference or variable selection procedure rely on structural conditions such as irrepresentable condition for variable selection, or the compatibility/restricted eigenvalue (RE) conditions for valid inference.  Strong correlation can disrupt these regular conditions, leading to procedure failure \citep{fan2020factor}. Second, controlling the FDR in large-scale multiple testing
	becomes difficult when test statistics are strongly dependent, violating the
	assumptions underlying many classical FDR procedures
	\citep{benjamini1995controlling,benjamini2001control,storey2004strong}.
	Empirical evidence shows that ignoring such dependence can result in severe FDR
	inflation \citep{wu2008false,blanchard2009adaptive,fan2012estimating,fan2019farmtest}. 
	
	To address these challenges, we propose a Factor-Adjusted Debiased Mediation Testing (FADMT) framework for high-dimensional individual mediation analysis with FDR control. Our motivation is that dependence among mediators in modern omics studies is
	often driven by a few common factors, so that an approximate factor structure
	can provide a useful and parsimonious representation \citep{bai2003inferential,fan2013large}. By separating pervasive factor-driven dependence from idiosyncratic variation, factor-adjusted methods have been shown to substantially improve inference and multiple testing accuracy \citep{fan2019farmtest,fan2024latent}. A fundamental difference between our setting and existing factor-adjusted frameworks is that traditional approximate factor models are applied to observable data, whereas we
	innovatively apply factor analysis to the unobserved errors of the mediator model. This requires a two-step construction:  we first estimate the latent factor component and obtain estimated idiosyncratic
	components, which serve as decorrelated pseudo-mediators. We then use these pseudo-mediators for downstream debiased inference and multiple testing. This shift introduces new technical challenges as  the first-step estimation error propagates into downstream procedure.

	We establish the asymptotic normality of the debiased estimator under mild regular conditions and
	develop a theoretically valid FDR control rule for individual mediation effects.
	Extensive simulations further demonstrate strong finite-sample performance:
	FADMT controls FDR across a wide range of dependence structures while maintaining
	competitive power relative to existing methods. We further apply our method to a multi-omics dataset from the TCGA-BRCA cohort, investigating whether DNA methylation mediates the effect of age at diagnosis on MKI67 gene expression, and to a financial stock connect setting, examining whether market liberalization affects firms' idiosyncratic risk through changes in corporate fundamentals. These applications demonstrate the method’s ability to uncover interpretable mediation effects in real-world high-dimensional data across domains.

	This paper makes several key contributions. 
	First, in the theory of high-dimensional inference and FDR control under strong dependence, we provide a rigorous foundation for large-scale multiple testing in settings where classical FDR procedures can fail due to pervasive correlations. Rather than relying on independence or weak-dependence assumptions, we establish a factor-adjusted framework which enables accurate FDR control even in strongly correlated, high-dimensional regimes.  This offers a general theoretical resolution to a long-standing challenge in high-dimensional inference.
	
	Second, at the methodological level of factor-adjusted inference, we expand the existing paradigm by moving factor adjustment from  observable quantities to a latent error structure. Unlike FARM \citep{fan2024latent} and FarmTest \citep{fan2019farmtest}, which impose factor models on observed variables, our framework
	posits and exploits an approximate factor structure in the unobserved errors of the mediator model. 
	
	Third, we instantiate these theoretical and methodological developments in
	high-dimensional individual mediation testing under a composite null. For the product-form mediation effect, we propose and analyze a MaxP-based,
	factor-adjusted testing procedure that explicitly accounts for the union structure of the null hypothesis.  Moreover, our framework improves robustness to latent common-factor dependence that manifests as pervasive shared components driving strong correlations among mediators. By estimating and adjusting for these latent factors through the mediator error structure, FADMT separates shared variation from mediator-specific signals, stabilizes downstream inference, and enhances the interpretability of discovered mediation findings.
	
	We adopt the following notations throughout the article.  For a vector $\ba = (a_1, \dots, a_d) \in \RR^d$, we denote its $\ell_1, \ell_2,$ and $\ell_\infty$ norms by $\|\ba\|_1, \|\ba\|_2,$ and $\|\ba\|_\infty$, respectively. The sub-Gaussian norm of a random variable $Z$ is defined as $\| Z \|_{\psi_2} = \inf \{t>0: \EE \exp (Z^2/t^2) \leq 2\}$, and for a random vector $\bx$, $\| \bx \|_{\psi_2} = \sup_{\| \bv\|_2 = 1} \| \bv\t \bx \|_{\psi_2}$. For a matrix $\bA = (a_{ij})$, we let $\|\bA\|_{\max} = \max_{i,j} |a_{ij}|$, $\|\bA\|_F$ be its Frobenius norm, and $\|\bA\|_1, \|\bA\|_\infty, \|\bA\|$ (or $\|\bA\|_2$) be its $\ell_1, \ell_\infty,$ and spectral norms. $\lambda_{\min}(\bA)$ and $\lambda_{\max}(\bA)$ denote the minimum and maximum eigenvalues of a square matrix $\bA$. For any set $\mathcal{S}$, $|\mathcal{S}|$ denotes its cardinality, and $[p] = \{1, \dots, p\}$. For two positive sequences $a_n$ and $b_n$, we write $a_n = O(b_n)$ if there exists a positive constant $C$ such that $a_n \leq C b_n$ for all sufficiently large $n$, and $a_n = o(b_n)$ if $a_n / b_n \to 0$ as $n \to \infty$. Similarly, $a_n = O_P(b_n)$ and $a_n = o_P(b_n)$ indicate that the corresponding relationships hold in probability.

	The rest of the article is organized as follows.
	Section~2 introduces the model and hypotheses and presents the proposed
	factor-adjusted debiased inference and multiple testing procedure.
	Section~3 establishes theoretical guarantees, including asymptotic validity and FDR control.
	Section~4 reports simulation studies. Section~5 presents data applications, and Section~6 concludes the paper.
	
	\section{Methodology}
	\subsection{Problem Setup}
	We consider a high-dimensional mediation framework with $n$ independent and identically distributed (i.i.d.)  observations $\{(y_i, s_i, \mathbf{m}_i)\}_{i=1}^n$, where $y_i \in \RR$  is the outcome, $s_i \in \RR$ is the exposure (or treatment), and $\bmm_i = (m_{1i}, \dots, m_{pi})\t \in \RR^p$ collects $p$ candidate mediators. We allow the number of mediators $p$ to exceed the sample size $n$, accommodating high-dimensional settings. The underlying relationships among these variables are modeled via the following linear structural equations:
	\begin{align}
		&y_i = s_i \alpha_s + \bmm_i^{\top}\balpha_m+ \xi_i,  \label{eq:1}\\
		& \bmm_i = s_i \bgamma_s + \br_i, \label{eq:2}
	\end{align}
	where $\balpha_m=(\alpha_{m1},\ldots,\alpha_{mp})^\top\in\mathbb R^p$ captures the mediator-outcome effects, and
	$\bgamma_s=(\gamma_{s1},\ldots,\gamma_{sp})^\top\in\mathbb R^p$ captures the exposure-mediator effects.
	The errors $\{\xi_i\}_{i=1}^n$ are i.i.d.\ with $\E(\xi_i)=0$ and $\Var(\xi_i)=\sigma^2$.
	The residual vectors $\{\br_i\}_{i=1}^n$ are i.i.d.\ with $\E(\br_i)=\mathbf 0$ and
	$\Cov(\br_i)=\bSigma_R$.
	We assume $\xi_i$ is independent of $(s_i,\bmm_i,\br_i)$ and $\br_i$ is independent of $s_i$.
	
	In matrix form, let $\by = (y_1, \ldots, y_n)\t \in \RR^n$ denote the outcome vector, $\bs = (s_1, \ldots, s_n)\t \in \RR^n$ the exposure vector, and $\bM = (\bmm_1, \ldots, \bmm_n)\t \in \RR^{n \times p}$ the mediator design matrix. We also write $\bxi=(\xi_1,\ldots,\xi_n)\t \in \RR^n$ and
	$\bR=(\br_1,\ldots,\br_n)^\top\in\RR^{n\times p}$. Substituting \eqref{eq:2} into \eqref{eq:1} yields the reduced-form model for the outcome:
	\begin{align}
		y_i = s_i \alpha_s + s_i \bgamma_s^{\top} \balpha_m +  \br_i^{\top}\balpha_m + \xi_i. \label{tradition}
	\end{align}
	In this expression, $\alpha_s$ captures the direct effect of the exposure $s_i$ on the outcome $y_i$, and
	$\bgamma_s^\top\balpha_m=\sum_{j=1}^p \gamma_{sj}\alpha_{mj}$
	is the total (aggregate) mediation effect through all mediators. Prior work has primarily focused on testing the overall mediation effect \citep{zhou2020estimation, guo2023statistical,lin2023testing}:
	\begin{align*}
		\text{H}_0: \bgamma_s^{\top}\balpha_m = 0  \quad versus \quad \text{H}_1: \bgamma_s^{\top}\balpha_m \neq 0. 
	\end{align*}
	While testing the overall mediation effect provides a global assessment of whether mediators collectively transmit the effect of the exposure on the outcome, it does not reveal which specific mediators are responsible for the observed indirect effect. This limitation is especially pertinent in high-dimensional settings, where the effects of individual mediators may vary in direction, potentially canceling each other out in the overall effect. Therefore, testing individual mediation effects becomes essential for identifying specific active mediators and understanding the underlying causal mechanisms in greater detail.
	
	Our goal is to test individual mediation effects for each mediator $j \in [p]$, defined as the product $\gamma_{sj}\alpha_{mj}$. The corresponding hypotheses  are stated as:
	\begin{align*}
		\text{H}_{0j}: \gamma_{sj}\alpha_{mj} = 0  \quad versus \quad \text{H}_{1j}: \gamma_{sj}\alpha_{mj} \neq 0, \qquad j=1,\dots,p.
	\end{align*}
	Testing these individual hypotheses in high-dimensional settings presents several challenges. First, the regime where $p \gg n$ necessitates regularized estimation. Classical estimators tend to be biased in high dimensions due to regularization effects, which calls for debiased estimation methods. Second, mediators often exhibit complex dependencies due to shared biological pathways or latent factors. These dependencies pose significant challenges for both variable selection and statistical inference. Third, the simultaneous testing of $p$ hypotheses requires rigorous control of the False Discovery Rate (FDR). Standard FDR procedures may suffer from inflated error rates under the strong dependence, emphasizing the importance of dependence-adjusted multiple testing frameworks. 
	
	\subsection{Latent Factor Adjustment and Model Reformulation} \label{sec:factor_adj}
	To mitigate the strong dependence among mediators, we adopt a factor-adjusted strategy inspired by recent work on high-dimensional inference under latent factor structures \citep{fan2019farmtest, fan2024latent}. The key idea is to remove a low-rank common component that drives most cross-mediator correlations, and to use the estimated idiosyncratic components for downstream inference.
	
	We work with the mediator-equation errors $\br_i$ in \eqref{eq:2}, which capture the variation in $\bmm_i$
	not explained by the exposure $s_i$. We assume an approximate factor model:
	\[
	\br_i = \bB \bfv_i + \bu_i,
	\]
	where $\bfv_i\in\mathbb R^K$ are latent factors, $\bB\in\mathbb R^{p\times K}$ is the loading matrix, and $\bu_i\in\mathbb R^p$ are idiosyncratic components with weak dependence. A crucial distinction between our framework and traditional factor analysis where models are typically applied to observable data is that $\br_i$ is unobserved. Consequently, we must first estimate it by $\wh{\br}_i=\bmm_i-s_i\wh{\bgamma}_s$ using the OLS estimator $\wh{\bgamma}_s$. This additional estimation step introduces a first-stage error, which propagates into the subsequent factor extraction and high-dimensional inference.

	Let $\bF=(\bfv_1,\ldots,\bfv_n)^\top\in\mathbb R^{n\times K}$, and $\bU=(\bu_1,\ldots,\bu_n)^\top\in\mathbb R^{n\times p}$. We apply principal component analysis \citep{bai2003inferential, fan2013large} to the estimated residual matrix $\wh{\bR} = (\wh{\br}_1,\dots,\wh{\br}_n)\t \in\RR^{n\times p}$ to obtain the latent factors and loadings. Under standard identifiability conditions:
	\begin{align*}
		\Cov(\bfv_i) = \bI_K, \ \text{and}  \ \bB\t \bB \ \text{is diagonal}.
	\end{align*}
	The estimators are derived as follows: the columns of $\widehat{\mathbf{F}}/\sqrt{n}$ are the eigenvectors of $\wh{\bR} \wh{\bR}^{\top}$ corresponding to the top $K$ eigenvalues, $\widehat{\mathbf{B}}=n^{-1}\wh{\bR}^{\top} \widehat{\mathbf{F}}$, and $\wh{\bU} =  \wh{\bR} - \wh{\bF} \wh{\bB}\t = (\bI_n - n^{-1}\wh{\bF} \wh{\bF}\t ) \wh{\bR}$.
	
	\begin{rem}
		A practical consideration is the choice of the number of factors $K$. There have been various method to estimate the number of factors \citep{bai2002determining,lam2012factor,ahn2013eigenvalue,fan2022estimating}.We adopt the eigenvalue ratio method in \cite{lam2012factor,ahn2013eigenvalue}, which is widely used in the factor modeling literature and yields a consistent estimator for the number of factors $K$. Let $\lambda_k(\wh{\bR}\wh{\bR}\t)$ be the $k$-th largest eigenvalue of $\wh{\bR}\wh{\bR}\t$ and $K_{\max}$ be a prescribed upper bound. Then, the number of factors is given by
		\begin{align*}
			\wh{K} = \arg \max_{k \leq K_{\max} } \frac{\lambda_k(\wh{\bR}\wh{\bR}\t)}{\lambda_{k+1}(\wh{\bR}\wh{\bR}\t)}.
		\end{align*}
	\end{rem}

		With the factor structure identified, we can decompose the mediator matrix as $\bM = \bs \wh{\bgamma}_s\t  + \wh{\bF} \wh{\bB}\t + \wh{\bU}$.  Substituting this decomposition into the outcome model \eqref{eq:1} and rearranging terms, we arrive at the factor-adjusted regression framework:
	\begin{align}
		\by &= \bs (\alpha_s + \wh{\bgamma}_s \t \balpha_m) + \wh{\bF}  \wh{\bB}\t \balpha_m +  \wh{\bU} \balpha_m +   \bxi \nonumber \\
		&= \bs \alpha_1  + \wh{\bF} \balpha_2 +  \wh{\bU} \balpha_m +  \bxi, \label{eq:factor_adj}
	\end{align}
	where $\alpha_1 = \alpha_s + \wh{\bgamma}_s\t \balpha_m$ and $ \balpha_2 = \wh{\bB}\t \balpha_m$ are treated as  nuisance parameters.
	This reformulation allows us to explicitly adjust for shared latent dependencies among mediators, while leveraging the decorrelated idiosyncratic residuals $\wh{\bU}$ as pseudo-predictors 
	for inference. This transformation enables valid inference even in the presence of strong dependence and high dimensionality, as will be demonstrated in our theoretical and numerical analyses.

	\subsection{Test Statistic}
	Building on the factor-adjusted model, we now develop the inferential procedure for individual mediation effects. A widely adopted strategy in mediation literature is the joint significance test (also known as the MaxP test) in \cite{mackinnon2002comparison}, which has been shown to outperform the Sobel test in both theoretical and empirical studies \citep{liu2022large, du2023methods}. The MaxP test rejects the null hypothesis $H_{0j}$ only when both components are statistically significant. Specifically, let $p_{\gamma_{sj}}$ and $p_{\alpha_{mj}}$ be the $p$-values for testing $\gamma_{sj} = 0$ and $\alpha_{mj} = 0$, respectively. The test statistic is defined as:   
	\begin{align*}P_{\max,j} = \max\{p_{\gamma_{sj}}, p_{\alpha_{mj}}\}, \quad j=1,\dots,p.\end{align*} 
	While can be readily obtained via standard OLS regression, constructing a valid $p$-value for $\alpha_{mj}$ in \eqref{eq:factor_adj} is challenging due to the high-dimensionality and the latent dependence structure. To this end, we employ a factor-adjusted debiased Lasso approach to recover asymptotic normality for the estimated mediator-outcome effects.
	
	Recalling the augmented Equation  \eqref{eq:factor_adj}, we obtain an initial penalized estimator of $\balpha_m^{n}$ by fitting a high-dimensional regression of $\by$ on $(\bs,\wh{\bF},\wh{\bU})$, treating the
	coefficients on $(\bs,\wh{\bF})$ as nuisance parameters. Specifically, we solve
\begin{align}
	\wh{\balpha}_m^{n} \ \in \arg \min_{\alpha_1,\balpha_2, \balpha_m} \left\{\frac{1}{2n}\|\by- \bs \alpha_1 - \wh{\bF} \balpha_2 - \wh{\bU} \balpha_m \|_2^2+\lambda\|\balpha_m\|_1\right\}. \label{lasso}
\end{align}
	Due to regularization, $\wh{\balpha}_m^{n}$ is biased. Following the debiasing
	framework for high-dimensional M-estimators
	\citep{javanmard2014confidence,van2014asymptotically}, we define the debiased
	estimator as:
\begin{align*}
	\wh{\balpha}_m^d &= \wh{\balpha}_m^{n} + \frac{1}{n} \wt{\bOmega} \wh{\bU}\t (\by - \wh{\bU}  \wh{\balpha}_m^{n})  \\
	&= \wh{\balpha}_m^{n}  + \frac{1}{n} \wt{\bOmega} \wh{\bU}\t (\bs \alpha_1  + \wh{\bF} \balpha_2 +  \wh{\bU} \balpha_m +  \bxi - \wh{\bU}  \wh{\balpha}_m^{n})
\end{align*}
where $\wt{\bOmega} \in \RR^{p \times p}$ serves as a decorrelating matrix. Using the orthogonality properties $\wh{\bU}^\top \bs = 0$ and $\wh{\bU}^\top \wh{\bF} = 0$,  the error of the debiased estimator can be decomposed as:
\begin{align*}
	\sqrt{n}(\wh{\balpha}_m^d - \balpha_m) = \frac{1}{\sqrt{n}} \wt{\bOmega}  \wh{\bU}^{\top}\bxi + \sqrt{n}(\wt{\bOmega} \wh{\bSigma}_U - \bI)(\balpha_m - \wh{\balpha}_m^n),
\end{align*}
where $\wh{\bSigma}_U = n^{-1} \wh{\bU}\t \wh{\bU}$. The first term represents the leading stochastic component with asymptotic variance $\sigma^2 \wt{\bOmega} \widehat{\bSigma}_U \wt{\bOmega}\t$, while the second term is the bias that becomes negligible under appropriate regular conditions. The decorrelating matrix $\wt{\bOmega}$ can be constructed via node-wise Lasso regression \citep{van2014asymptotically,javanmard2018debiasing} or constrained convex optimization \citep{javanmard2014confidence,battey2018distributed}. We provide details for these approaches in the Appendix and compare their finite-sample performance in Section 4.

\begin{rem}[Orthogonality of the pseudo-mediators]
	By construction, the estimated idiosyncratic component matrix satisfies
	$\wh{\bU}^\top \wh{\bF}= \mathbf 0$ and $\wh{\bU}^\top \bs=\mathbf 0$.
	Since $\wh{\bR}$ is the OLS residual matrix from regressing $\bM$ on $\bs$,
	we have $\wh{\bR}^\top \bs=0$. Given that $\wh{\bF}$ is derived from the principal components of $\wh{\bR}$, its columns lie in the column space of $\wh{\bR}$ , implying $\wh{\bF}^\top\bs=0$. Consequently, $\wh{\bU}=(\bI_n - n^{-1}\wh{\bF} \wh{\bF}\t ) \wh{\bR}$ also satisfies $\wh{\bU}^\top \bs = 0$.
\end{rem}

	Based on the asymptotic normality of the debiased estimator $\wh{\balpha}_m^d$ established in Theorem \ref{the1}, the $p$-value  for testing $H_0: \alpha_{mj} = 0$  is given by:
	\begin{align*}
		p_{\alpha_{mj}}=2\bigg(1-\Phi\bigg(\frac{\sqrt{n}|\widehat{\alpha}^d_{mj}|}{\widehat{\sigma}[\wt{\bOmega} \widehat{\bSigma}_U \wt{\bOmega}\t]_{j,j}^{1/2}}\bigg)\bigg),
	\end{align*}
	where $\Phi(\cdot)$  denotes the cumulative distribution function (CDF) of the standard normal distribution, $\wh{\sigma}$ is a consistent estimator for $\sigma$. By combining $p_{\alpha_{mj}}$  with the $p$-value $p_{\gamma_{sj}}$ obtained from the first-stage OLS regression, we arrive at the joint significance test statistic $P_{\max,j} = \max\{p_{\gamma_{sj}}, p_{\alpha_{mj}}\}$ for each $j \in [p]$. The complete procedure is summarized in Algorithm~\ref{alg:estimation}.
	
	\begin{algorithm}[t]
		\caption{Factor-adjusted debiased inference for individual mediation effects}
		\label{alg:estimation}
		\begin{algorithmic}[1]
			\Require Data $\{(y_i,s_i,\bmm_i)\}_{i=1}^n$ with $\bmm_i\in\mathbb R^p$.
			\Ensure MaxP $p$-values $P_{\max,j}$ for $j=1,\ldots,p$.
			
			\State \textbf{Path A (exposure$\to$mediator).}
			For each $j$, regress $m_{ji}$ on $s_i$ by OLS to obtain $\wh{\gamma}_{sj}$ and the $p$-value $p_{\gamma_{sj}}$.
			Let $\wh{\bR}=\bM-\bs\wh{\bgamma}_s^\top$.
			
			\State \textbf{Residual factor extraction.}
			Apply PCA to $\wh{\bR}$ to obtain $(\wh{\bF},\wh{\bB})$ and the estimated idiosyncratic component matrix
			$\wh{\bU}=\wh{\bR}-\wh{\bF}\wh{\bB}^\top$.
			
			\State \textbf{Factor-adjusted Lasso.}
			Fit the penalized regression of $\by$ on $(\bs,\wh{\bF},\wh{\bU})$ with an $\ell_1$ penalty on $\balpha_m$
			(as in \eqref{lasso}), yielding $\wh{\balpha}_m^{n}$.
			
			\State \textbf{Debiasing.}
			Construct a decorrelating matrix $\wt{\bOmega}$ (nodewise Lasso or convex optimization).
			Form the debiased estimator
			\[
			\wh{\balpha}_m^d = \wh{\balpha}_m^{n} + \frac{1}{n} \wt{\bOmega} \wh{\bU}\t (\by - \wh{\bU}  \wh{\balpha}_m^{n}).
			\]
			
			\State \textbf{Path B (mediator$\to$outcome) $p$-values.}
			Compute $\wh{\bSigma}_U=n^{-1}\wh{\bU}^\top\wh{\bU}$ and a consistent $\wh{\sigma}$.
			For each $j$, compute
			\[
			p_{\alpha_{mj}}
			=
			2\left\{1-\Phi\!\left(
			\frac{\sqrt{n}\,|\wh{\alpha}_{mj}^d|}
			{\wh{\sigma}\,[\wt{\bOmega}\wh{\bSigma}_U\wt{\bOmega}^\top]_{jj}^{1/2}}
			\right)\right\}.
			\]
			
			\State \textbf{MaxP combination.} Output $P_{\max,j}=\max\{p_{\gamma_{sj}},p_{\alpha_{mj}}\}$ for $j=1,\ldots,p$.
		\end{algorithmic}
	\end{algorithm}

	\subsection{Multiple Testing and FDR Control}
	We aim to simultaneously test the $p$ mediation hypotheses
	$H_{0j}:\gamma_{sj}\alpha_{mj}=0$ versus $H_{1j}:\gamma_{sj}\alpha_{mj}\neq 0$ for $j=1,\ldots,p$.
	Let $P_{\max,j}=\max\{p_{\gamma_{sj}},p_{\alpha_{mj}}\}$ be the MaxP $p$-value. For a threshold $t\in[0,1]$, define the rejection set
	$\mathcal R(t)=\{j: P_{\max,j}\le t\}$ with $R(t)=|\mathcal R(t)|$.
	Let $\mathcal S_0=\{j:\gamma_{sj}\alpha_{mj}=0\}$ be the set of true mediation nulls and
	$V(t)=|\mathcal R(t)\cap \mathcal S_0|$ be the number of false discoveries. False Discovery Proportion (FDP) and FDR are defined as:
	\begin{align*}
		\mathrm{FDP}(t)=\frac{V(t)}{R(t)\vee 1},
		\qquad
		\mathrm{FDR}(t)=\E\{\mathrm{FDP}(t)\}.
	\end{align*}
	A unique challenge in mediation analysis is that the null hypothesis $\text{H}_{0j}$ is a composite null, representable as the union of three disjoint cases:
	\begin{align*}
		&\text{H}_{01,j}: \alpha_{mj} = 0, \quad \gamma_{sj} \neq 0 \\
		&\text{H}_{10,j}: \alpha_{mj} \neq 0, \quad \gamma_{sj} = 0 \\
		&\text{H}_{00,j}: \alpha_{mj} = 0, \quad \gamma_{sj} = 0.
	\end{align*}
	Write $\pi_{01},\pi_{10},\pi_{00}$ for their proportions among all $p$ tests. Standard FDR procedures, such as the Benjamini-Hochberg (BH) method, assume a uniform $U(0,1)$ distribution for null $p$-values \citep{benjamini1995controlling, benjamini2001control}. However, under the double-null $\text{H}_{00,j}$, $P_{max,j}$ follows Beta(2,1) distribution, as shown in \cite{liu2022large,dai2022multiple}. This deviates from the standard uniform reference and makes BH-type procedures overly conservative. Motivated by recent developments in mediation testing, we construct a mixture null distribution to estimate the FDP more accurately \citep{liu2022large,dai2022multiple}.

	Under the high-dimensional sparse modeling framework, we assume that most mediators have no effect on the outcome. Consequently, the proportion of null cases where only the mediator-outcome effect is present ($\pi_{10}$) is negligible, and we focus on the mixture of $\text{H}_{01}$ and $\text{H}_{00}$. We estimate the proportion of null exposure-mediator effects, $\pi_0^\gamma$, using Storey’s method \citep{storey2002direct}: 
	\begin{align}
		\wh{\pi}_0^\gamma(\eta) = \frac{1}{(1-\eta)p}\sum_{j=1}^p \mathbf{1}(p_{\gamma_{sj}} \geq \eta), \label{eq:pi}
	\end{align}
	where $\eta$ is a tuning parameter. This estimator assumes that most large $p$ values  come from true null hypotheses and are uniformly distributed. For well-chosen $\eta$, about $\pi_0 (1-\eta)$ of the $p$ values lie in the interval $(\eta,1]$. Therefore, the proportion of $p$ values that exceed $\eta$ should be close to $\pi_0 (1-\eta)$. Under sparsity, $\hat{\pi}_0^\gamma(\eta)$ serves as an estimate for $\pi_{00}$, while $1-\hat{\pi}_0^\gamma(\eta)$ estimates $\pi_{01}$.

	\begin{rem}
		A value of  $\eta = 1/2$ is used in the SAM software in \cite{storey2003positive}. \cite{blanchard2009adaptive} suggests to use $\eta$ equal to the significance level                                                                                                                                                                                         for dependent $p$ values. We follow \cite{storey2004strong} to adopt a bootstrap-based automatic selection for $\eta$.
	\end{rem}
	
	Using the estimated null proportions, we define the adjusted FDP estimator as:
	\begin{align}
		\widehat{\text{FDP}}_\eta(t) = \frac{p \left[\wh{\pi}_0^\gamma(\eta) t^2 +(1-\wh{\pi}_0^\gamma(\eta) ) t\right]}{\sum_{j=1}^p \mathbf{1}(P_{max,j} \leq t)}. \label{eq:fdp}
	\end{align}
	For a target FDR level $q \in (0,1)$, the optimal significance threshold is determined by:
	\begin{align}
		\hat{t}_{q,\eta} = \sup\{t: \widehat{\text{FDP}}_\eta(t) \leq q\}. \label{eq:thres}
	\end{align}
	Theoretical results show this procedure controls the FDR asymptotically under appropriate regularity conditions. The full procedure is detailed in Algorithm~\ref{alg:fdr}.
	
	\begin{algorithm}[htbp]
		\caption{FDR control for MaxP $p$-values}
		\label{alg:fdr}
		\begin{algorithmic}[1]
			\Require $\{p_{\gamma_{sj}},p_{\alpha_{mj}}\}_{j=1}^p$; target level $q$; tuning $\eta$.
			\State Compute $P_{\max,j}=\max\{p_{\gamma_{sj}},p_{\alpha_{mj}}\}$ for all $j$.
			\State Estimate $\wh{\pi}_0^\gamma(\eta)$ by \eqref{eq:pi} (with $\eta$ selected as in \citealp{storey2004strong}).
			\State For $t$ in the set of observed $\{P_{\max,j}\}$, compute $\wh{\mathrm{FDP}}_\eta(t)$ in \eqref{eq:fdp}.
			\State Set $\wh t_{q,\eta}$ by \eqref{eq:thres} and reject $\{j:P_{\max,j}\le \wh t_{q,\eta}\}$.
		\end{algorithmic}
	\end{algorithm}

	\section{Theoretical Results}
	This section establishes the theoretical foundations of the proposed factor-adjusted inference framework. We begin by outlining the regularity conditions, then derive the convergence rates for the estimated latent structures. Finally, we establish the asymptotic normality of the debiased estimator and the validity of the FDR control procedure.
	\subsection{Regularity Conditions and Error Propagation}
	To accommodate the high-dimensional setting, we impose the following assumptions on the data-generating process and the latent factor structure.
	\begin{assumption}[Sub-Gaussianity]
		The sequence $\{(s_i, \bfv_i\t,\bu_i\t)\t \}_{i=1}^n$  are i.i.d. random vectors. The exposure  $s_i$ is sub-Gaussian with $\EE[s_i^2] >0$. The latent factors $\bfv_i$ and idiosyncratic components $\bu_i$ are zero mean sub-Gaussian vectors such that $\|\bfv_i\|_{\psi_2} \leq c_0$ and $\|\bu_i\|_{\psi_2} \leq c_0$ for some positive constant $c_0 \leq \infty$ .
	\end{assumption}
	\begin{assumption}[Pervasive Condition]
		All the eigenvalues of $\mathbf{B}^{\top} \mathbf{B}/p$ are bounded away from 0 and $\infty$ as $p \to \infty$. That is,  $0 < c  \leq \lambda_{\min}(\mathbf{B}^{\top} \mathbf{B}/p)  \leq \lambda_{\max}(\mathbf{B}^{\top} \mathbf{B}/p) \leq C < \infty$.
	\end{assumption}
	\begin{assumption}[Loading matrix and Idiosyncratic component]
		Let $\bSigma_u=\Cov(\bu_i)$. Assume $\lambda_{\min}(\bSigma_u)\ge c_1$, $\|\bSigma_u\|_1\le c_2$, $\min_{s,t \in [p]} \mathrm{var}(u_{si}u_{ti}) > c_1$ for some
		constants $c_1,c_2>0$, and $\|\bB\|_{\max}\le C$ for some constant $C>0$.
		In addition, for all $i,j \in [n]$ there exists $C_4<\infty$ such that $\E[p^{-1/2}\{\bu_i^{\top}\bu_j - \E(\bu_i^{\top}\bu_j)\}^4] < C_4$ and $\E \|p^{-1/2} \bB^{\top} \bu_i \|_2^4 < C_4$.
	\end{assumption}

	\begin{rem}
		Assumption 1 is standard for high-dimensional inference, ensuring that tail behaviors are well-controlled via concentration inequalities. Assumptions 2 and 3 is  common in factor models \citep{fan2013large,fan2020factor, fan2024latent}. Together, Assumptions 2 and 3 ensure that $\bfv_i$ and $ \bu_i$ can be consistently estimated by the PCA method . 
	\end{rem}
	A distinctive feature of our framework is that the factor model is fitted to estimated residuals. The following proposition quantifies the error introduced by the first-stage OLS estimation.
	
	\begin{proposition} \label{prop1}
		Under Assumption 1, we have
		\begin{align*}
			\|  \wh{\bR} - \bR \|_{\max} =O_P \left(\sqrt{\frac{\log n \log p }{n}}\right). 
		\end{align*}
	\end{proposition}
	Proposition \ref{prop1} shows  the OLS estimation error which propagates into the factor extraction process. This rate determines the precision of the estimated idiosyncratic component $\wh{\bU}$, which serve as our pseudo-predictors. We define $r_{n,p}:=\sqrt{\frac{\log n\,\log p}{n}}$ for notational convenience.
	
	\begin{proposition} \label{lemma:factor}
		Suppose Assumptions 1-3 hold. Let $\bH = n^{-1} \bV^{-1} \wh{\bF}\t \bF \bB\t \bB$, where $\bV \in \RR^{K \times K}$ is a diagonal matrix containing the first $K$ largest eigenvalues of $n^{-1}\wh{\bR}\wh{\bR}\t$. Then:
		\begin{enumerate}
			\item $ \max_{k \in [K]} \frac{1}{n} \sum_{i=1}^{n} \left| (\wh{\bfv}_i - \bH \bfv_i)_k\right|^2 =O_P\left(1/n + 1/p +r_{n,p}^2  \right)  ; $
			\item $ \frac{1}{n} \sum_{i=1}^{n} \left\| \wh{\bfv}_i - \bH \bfv_i\right\|^2 =O_P\left(1/n + 1/p +r_{n,p}^2  \right)  $
			\item $ \max_{j \in [p]} n^{-1} \sum_{i=1}^n |\widehat{u}_{ji} - u_{ji}|^2 = O_P\left(\frac{\log p}{n} + \frac{1}{p} +	r_{n,p}^2  \right); $
		\end{enumerate}
	\end{proposition}
	
	\begin{rem}
		The results extend classical PCA error bounds (e.g., \citealp{fan2013large}) to the present
		two-stage setting, where PCA is applied to $\wh{\bR}$ rather than the unobserved $\bR$.
		The additional $r_{n,p}^2$ term quantifies the impact of the first-stage residual-proxy error.
	\end{rem}
	
	\subsection{Asymptotic Normality of the Debiased Estimator}
	To establish the validity of the individual mediation tests, we require the debiased estimator to be asymptotically normal. 
	In the theory below, we focus on constructing the decorrelating matrix $\wt{\bOmega}$ via nodewise Lasso. This necessitates a sparsity condition on the mediator-outcome effects and the precision matrix of the idiosyncratic errors. 
	Accordingly, the following sparsity conditions are imposed on the precision matrix $\bTheta_u=\bSigma_u^{-1}$,
	which are standard for nodewise-Lasso-based debiasing.

	\begin{assumption}[Sparsity]\label{ass:sparsity}
		Let $s_\alpha := |\{j\in[p]:\alpha_{mj}\neq 0\}|$ and let
		$s_j := |\{\ell\neq j: (\bTheta_u)_{j\ell}\neq 0\}|$ where $\bTheta_u=\bSigma_u^{-1}$. Let $d:=K+1$ denote the dimension of the unpenalized nuisance parameter $(\alpha_1,\balpha_2)$.
		Assume
		\[
		s_\alpha + d = o\!\left(\frac{\sqrt n}{\log p\,\log n}\right),
		\qquad
		\max_{1\le j\le p}s_j = o\!\left(\frac{n}{\log p\,\log n}\right).
		\]
	\end{assumption}
	
	\begin{assumption}[Consistent Estimation of $\sigma$]
		Assume $\wh{\sigma}^2 = \sigma^2 + o_p(1)$.
	\end{assumption}
	\begin{rem}
		Assumption 4 imposes a sparsity condition to ensure the asymptotic normality of the debiased Lasso estimator. Previous studies have established that $s^*= o(n /\log p)$ is the sparisity condition for consistent estimation \citep{candes2007dantzig,bickel2009simultaneous}, while $s^*= o(\sqrt{n} /\log p)$ is the condition for asympotic normality \citep{van2014asymptotically,javanmard2014confidence,zhang2014confidence}. We adopt a similar sparsity condition up to the logarithmic factor to account for the additional complexity introduced by the factor estimation. The sparsity assumption for $s_j$ is standard in high-dimensional inference using nodewise Lasso \citep{van2014asymptotically,javanmard2018debiasing}.  Under the usual sub-Gaussian assumption, the sparsity requirement on $s_j$ is typically $\max_j s_j = o(n/\log p)$. We also adopt a similar sparsity condition, with an additional $\log n $ factor in the denominator, to account for the extra complexity introduced by the factor estimation step. Assumption 5 holds when $\wh{\sigma}$ is computed by refitted cross-validation in \cite{fan2012variance}  or scaled lasso in \cite{sun2013sparse}.
	\end{rem}

	\begin{theorem} \label{the1}
		Under Assumptions 1-4, and assuming that the error term $\xi_i \sim N(0,\sigma^2)$ .Let $\lambda\asymp \sqrt{\log p/n} + p^{-1/2} + r_{n,p}$ in the factor-adjusted Lasso (\ref{lasso}),
		and let the nodewise-Lasso tuning parameters satisfy the same order uniformly in $j$. Then
		\[
		\sqrt n\bigl(\wh{\balpha}_m^d-\balpha_m\bigr)
		=
		\bZ+\bDelta,
		\qquad
		\bZ\,\big|\,\wh{\bU}\sim N\!\bigl(\mathbf 0,\ \sigma^2\wt{\bOmega}\wh{\bSigma}_U\wt{\bOmega}^\top\bigr),
		\qquad
		\|\bDelta\|_\infty=o_p(1),
		\]
		where $\wh{\bSigma}_U=n^{-1}\wh{\bU}^\top\wh{\bU}$.
	\end{theorem}
	
	\begin{corollary}  \label{cor1}
		Under the conditions of Theorem \ref{the1}, for any $j \in \mathcal{H}_{0j}^{\alpha}$, the debiased $p$-value $p_{\alpha_{mj}}$ is asymptotically uniform conditional on $\mathcal{F} := \sigma(\bs, \bM)$. Specifically, for any $t \in [0,1]$,
		\begin{align*}
			\sup_{t \in [0,1]} \Big| \Pr\bigl(p_{\alpha_{mj}} \le t \,\big|\, \mathcal{F}\bigr) - t \Big| \xrightarrow{P} 0 \quad \text{as } n, p \to \infty.
		\end{align*}
		Consequently, for any $\mathcal F$-measurable statistic $T$,
		$p_{\alpha_{mj}}$ is asymptotically independent of $T$.
	\end{corollary}
	
	\begin{rem}
		In finite samples, the debiased inference for $\alpha_{mj}$ may exhibit mild conservativeness, manifesting as $p$-values that are stochastically larger than the uniform distribution (i.e., super-uniform) under the null hypothesis. This behavior is primarily attributed to factor-estimation error and the regularization involved in constructing $\wt{\bOmega}$.
		Such conservativeness does not compromise the validity of the FDR control procedure. Since the FDR control framework relies on the FDP estimator being a conservative upper bound, super-uniform null $p$-values merely result in a more protective threshold.
	\end{rem}

	\subsection{Validity of FDR Control}
	Finally, we show that the proposed FDR control procedure is valid. We make the following assumptions for our asymptotic results.
	
	\begin{assumption}[Empirical convergence of $p_{\gamma_{sj}}$ and sparsity]\label{ass:A6}
		Let $\mathcal H_{00},\mathcal H_{01},\mathcal H_{10},\mathcal H_{11}$ denote the four component sets, and
		$|\mathcal H_{ab}|/p \to \pi_{ab}$ for $ab\in\{00,01,10,11\}$ as $p\to\infty$.
		\smallskip
		\noindent (i) (Empirical tail convergence of $p_{\gamma_{sj}}$)  
		For $ab\in\{01,11\}$, there exist continuous functions $G_{ab}(\cdot)$ such that for all $t\in(0,1]$,
		\[
		\frac{1}{|\mathcal H_{ab}|}\sum_{j\in\mathcal H_{ab}}\mathbf 1\!\left(p_{\gamma_{sj}}>t\right)\ \to\ G_{ab}(t)
		\quad\text{almost surely}.
		\]
		Moreover, for $j\in\mathcal H_{00}\cup\mathcal H_{10}$ (i.e., $\gamma_{sj}=0$),
		$p_{\gamma_{sj}}$ are asymptotically $\mathrm{Unif}(0,1)$ and satisfy the corresponding empirical convergence.
		
		\smallskip
		\noindent (ii) (Sparsity)
		$|\mathcal H_{10}|/p \to 0$.
	\end{assumption}

	\begin{rem}
		Assumption 6 requires the almost sure pointwise convergence of the empirical $p$-value processes, a condition widely used in establishing FDR control \citep{storey2004strong,dai2022multiple}.  Assumption 6  does not preclude strong cross-sectional correlation in $\br_i$ which is explicitly modeled
		through a low-dimensional latent factor structure. In particular,  strong dependence induced by a fixed number of pervasive factors is compatible with empirical convergence, provided that the remaining idiosyncratic components exhibit only weak dependence.
	\end{rem}

	\begin{theorem}\label{the2}
		Assume Theorem~\ref{the1} and Corollary~\ref{cor1} hold, and Assumption~\ref{ass:A6} holds.  As $n,p \to \infty$, $\wh{\mathrm{FDP}}_\eta(t) $ is a conservative estimate of $\mathrm{FDR}(t)$ for all $t \in (0,1)$ that satisfies $G_{01}(t) \geq \frac{G_{01}(\eta)}{1-\eta}$. Moreover,
		the significance threshold $\wh{t}_{q,\eta}$ expressed in (\ref{eq:thres}) controls the FDR at level $q$:
		\begin{align*}
			\mathrm{FDP}(\wh{t}_{q,\eta})  \leq q + o_p(1) \ \text{and} \ \limsup_{p\to\infty} \mathrm{FDR}(\wh{t}_{q,\eta}) \leq q.
		\end{align*}
	\end{theorem}

\begin{rem}
		The conservativeness of $\widehat{\mathrm{FDP}}_{\eta}(t)$  is established pointwise over an
		admissible set of $t$'s, rather than uniformly for all $t\in(0,1)$.
		This phenomenon is standard for mixture-based FDP estimators under composite
		(union) nulls; see \cite{dai2022multiple} for an analogous
		condition.
		Indeed, $\eta$ is typically chosen close to $1$, and the distribution of $p$-values under alternative hypothesis is stochastically less than the uniform distribution.
		The condition should hold for the small significance cutoffs typically used in
		multiple testing.
\end{rem}
	
	\section{Simulation Studies}
	In this section, we conduct Monte Carlo simulations to investigate the finite-sample performance of our proposed Factor-Adjusted Debiased Mediation Testing (FADMT) and compare it with existing methodologies.
	
	We consider a sample size $n=300$ and $p=500$ mediators. The exposure is generated as $s_i \sim N(0,1)$. The mediators are generated as $\bmm_i = s_i \bgamma_s + \br_i$, where $\br_i \sim N(0,\bSigma^*)$. We set the exposure-mediator effects $\gamma_{sj} = \delta$ for $j=1,\dots,250$, and $\gamma_{sj}=0$ for $j > 250$. The response $y_i$ is generated from  $y_i = s_i \alpha_s + \bmm_i^{\top}\balpha_m+ \xi_i$, where $\alpha_s=0.5$, $\xi_i \sim N(0,0.5^2)$, and the mediation-outcome effects are $\balpha_m = (\underbrace{\delta, \ldots, \delta}_{10}, 0, \ldots, 0)^{\top}$. The parameter $\delta$ represents the signal strength. The simulation results are based on 200
	replications. The target FDR level is set to $q=0.1$. We consider five covariance structures for $\bSigma^*$ corresponding to Model 1 through Model 5:
	
	Model 1 (AR):  Assume the covariance matrix $\bSigma^*$ being an AR correlation structure. That is $\Sigma^*_{ij} = 0.8^{|i-j|}$.
	
	Model 2 (Factor Model): Assume $\br_i$ are generated from three factor model $\br_i = \mathbf{B}\bfv_i + \bu_i$, where factors $\bfv_i \sim N(0,\bI_3)$, the element of loading matrix $\mathbf{B}$ is generated from a uniform distribution $U[-1,1]$, and $ \bu_i \sim N(0,\bI_p)$.
	
	Model 3 (Compound Symmetric): Consider a symmetric matrix  $\bSigma^*$ with diagonal elements 1 and each off-diagonal element equals 0.8.
	
	Model 4 (Long Memory): Consider $\bSigma^*$ where each element is defıned as $\Sigma^*_{ij}=0.5(||i-j|+1|^{2H}-2|i-j|^{2H}+||i-j|-$ $1|^{2H}),1\leqslant i,j\leqslant p$, with $H=0.9.$ Model 4 is from \cite{bickel2008regularized} and has also recently been considered by \cite{fan2017estimation} for strong long memory dependence.
	
	Model 5 (Independent):  $\bSigma^* =\bI_p$.
	
	We compare the proposed FADMT (Factor-Adjusted Debiased Mediation Testing) with an
	ablation baseline DMT (Debiased Mediation Testing), which applies the same debiased inference and the same multiple-testing pipeline but without factor adjustment.
	For each of FADMT and DMT, we consider two constructions of the decorrelating matrix $\wt{\bOmega}$:
	(i) nodewise regression \citep{van2014asymptotically} and (ii) convex optimization
	\citep{javanmard2014confidence}.
	In addition, we include two  methods for high-dimensional individual mediation effect testing under sparse linear models, HIMA \citep{zhang2016estimating} and HIMA2 \citep{perera2022hima2}\footnote{Both HIMA and HIMA2 are implemented using the \texttt{HIMA} R package, with their default methods.}.

	Let $\widehat{\mathcal H}$ be the selected set of mediators at FDR level $q$.
	We report the empirical average false discovery proportion (FDP) and true positive rate (TPR), where TPR  is obtained by averaging the proportion of correctly selected mediators $|\wh{\mathcal{H}} \cap \mathcal{H}^*| / {|\mathcal{H}^*|}$ over 200 repetitions.  Table~\ref{tab:sim_fdr_tpr} reports the overall FDR and TPR for FADMT and DMT under signal $\delta=0.5$.
	Across all covariance designs, FADMT controls FDR at the nominal level $q=0.1$,
	while DMT can exhibit substantial inflation, especially under strong dependence.
	This isolates the practical benefit of factor adjustment in stabilizing debiased inference for $\balpha_m$.
	
	Table~\ref{tab:sim_fdr_tpr} further reveals difference between the two constructions
	of the decorrelating matrix, namely nodewise Lasso regression and convex optimization.
	For the DMT method, the convex optimization approach generally leads to worse FDR control than
	nodewise Lasso across most dependent designs, likely reflecting the difficulty of solving the
	optimization problem accurately when the mediator covariance structure is highly correlated.
	Only in Model~5, where mediators are independent, does convex optimization outperform nodewise
	Lasso in terms of FDR control.  Interestingly, this pattern reverses under the proposed FADMT framework.
	Both nodewise Lasso and convex optimization achieve valid FDR control after factor adjustment,
	but nodewise Lasso tends to be more conservative, often yielding type~I error rates for inference
	on $\balpha_m$ below the nominal level and correspondingly conservative FDR values.
	In contrast, convex optimization attains higher power while maintaining accurate FDR control.
	This improvement can be attributed to the factor adjustment step, which effectively removes
	cross-sectional dependence among mediators and simplifies the residual covariance structure.
	As a result, convex optimization becomes more efficient and benefits from its
	variance-minimization objective, leading to reduced estimator variance and enhanced power. From a practical perspective, these gains are particularly appealing, as convex optimization is
	computationally substantially more efficient than nodewise Lasso.
	Consequently, the proposed FADMT method not only improves statistical stability under dependence
	but also amplifies the computational advantages of convex optimization in large-scale applications. Additional decomposition results, including the Type~I error and power for testing individual
	$\alpha_{mj}$ and $\gamma_{sj}$ effects, are reported in Appendix~\ref{app:sim_typeI_power}. We also assess the convergence of the empirical $p$-value processes in the Appendix~\ref{app:empirical}, showing that the empirical distributions of the null $p_{\gamma_{sj}}$ and $p_{\max,j}$ closely follow their theoretical references. These results provide empirical support for the theoretical guarantees established in Theorem~\ref{the2}.
	
	To further investigate the robustness of FADMT, we vary the signal strength $\delta \in \{0.3, 0.4, 0.5\}$, affecting both $\gamma_{sj}$ and $\alpha_{mj}$ simultaneously. Figure~\ref{fig:sim_fdr}  compares FADMT with DMT, HIMA, and HIMA2 across all five covariance designs and signal levels. As shown in Figure~\ref{fig:sim_fdr}, FADMT consistently controls the FDR at the
	nominal level across all signal strengths and all covariance models, for both
	the nodewise Lasso and convex optimization implementations.
	In contrast, competing methods exhibit pronounced FDR inflation in the presence
	of strong dependence among mediators.
	In particular, under Models~2 and~3, DMT, HIMA, and HIMA2 frequently exceed the target FDR level. Meanwhile, the improved FDR control of FADMT does not come at the expense of
	statistical power.
	Across all signal levels, the TPR achieved by FADMT is comparable to that of existing methods.This results highlight the robustness of FADMT in simultaneously maintaining valid error control and competitive power in challenging high-dimensional settings.
	
	\begin{table}[htbp]
		\centering
		\caption{Overall FDR and TPR under five models with $(n,p)=(300,500)$.}
		\label{tab:sim_fdr_tpr}
		\small
		\setlength{\tabcolsep}{6pt}
		\renewcommand{\arraystretch}{1.0}
		\begin{tabular}{ccccccccc}
			\toprule
			& \multicolumn{2}{c}{FADMT (nodewise)} & \multicolumn{2}{c}{FADMT (convex)} 
			& \multicolumn{2}{c}{DMT (nodewise)} & \multicolumn{2}{c}{DMT (convex)} \\
			\cmidrule(lr){2-3}\cmidrule(lr){4-5}\cmidrule(lr){6-7}\cmidrule(lr){8-9}
			Model & FDR & TPR & FDR & TPR & FDR & TPR & FDR & TPR \\
			\midrule
			1 & 0.0568 & 0.9470 & 0.0725 & 0.9710 & 0.0918 & 1.0000 & 0.3098 & 1.0000 \\
			2 & 0.0708 & 0.9955 & 0.0948 & 0.9955 & 0.1109 & 0.9955 & 0.3209 & 0.9960 \\
			3 & 0.0750 & 0.9965 & 0.0973 & 1.0000 & 0.2043 & 1.0000 & 0.6420 & 1.0000 \\
			4 & 0.0713 & 1.0000 & 0.0750 & 1.0000 & 0.0919 & 1.0000 & 0.1727 & 1.0000 \\
			5 & 0.0647 & 0.9145 & 0.0848 & 0.9520 & 0.1623 & 1.0000 & 0.1497 & 1.0000 \\
			\bottomrule
		\end{tabular}
	\end{table}

	\begin{figure}[htbp]
		\centering
		\includegraphics[width=\textwidth]{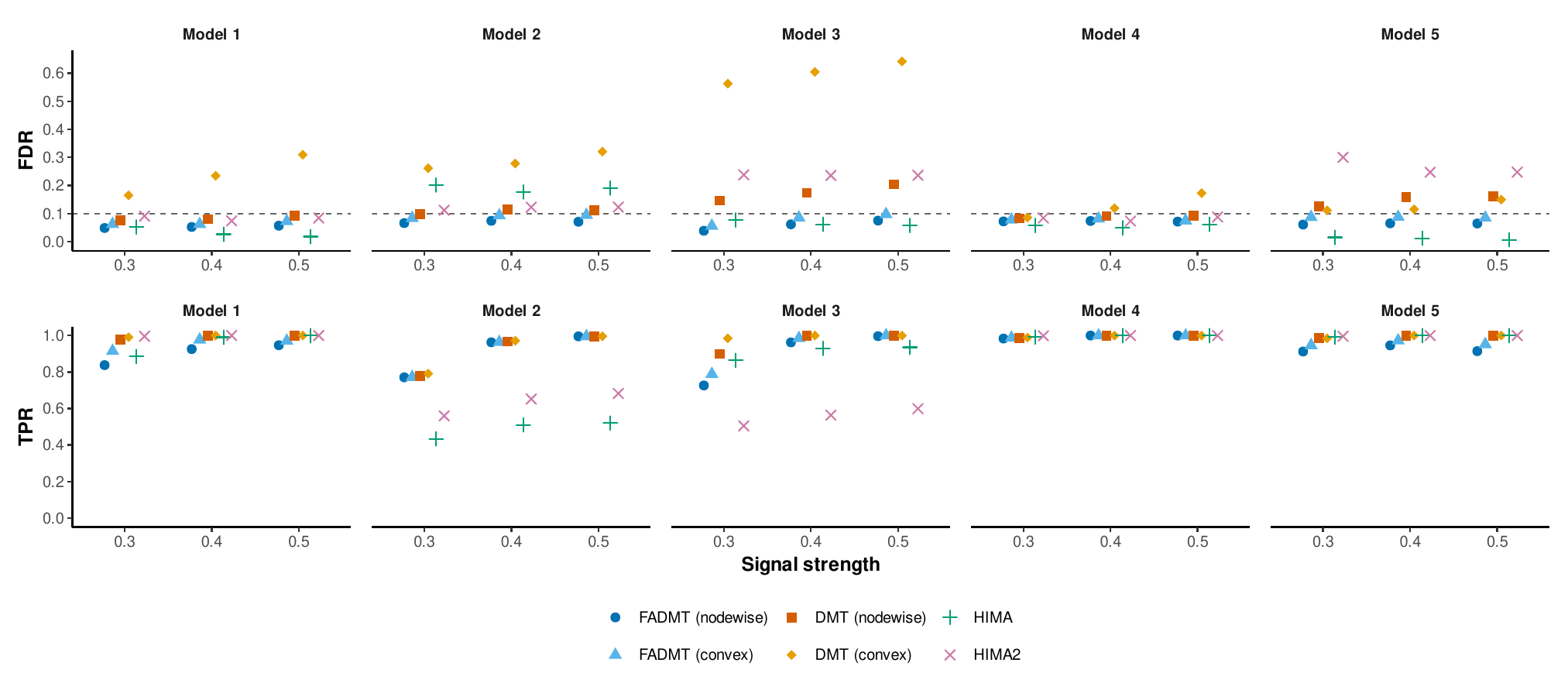}
		\caption{Comparison of FDR and TPR across different methods and signal strengths.}
		\label{fig:sim_fdr}
	\end{figure}

	\section{Data Application}
	\subsection{Biomedical Application}
	
	To demonstrate the practical utility of our proposed methodology, we apply it to multi-omics data from the TCGA-BRCA cohort\footnote{Data available at \url{https://portal.gdc.cancer.gov/projects/TCGA-BRCA.}}, focusing on the role of DNA methylation in mediating the relationship between age at diagnosis and the expression of MKI67, a well-established proliferation marker in breast cancer. Previous studies have shown that aging is accompanied by systematic changes in gene expression and epigenetic modifications that can influence cancer development and progression \citep{baylin2011decade, rakyan2011epigenome}. In particular, DNA methylation alterations have been reported to mediate transcriptional regulation in an age-dependent manner \citep{chatsirisupachai2021integrative}.
	
	Data were retrieved from the GDC portal via the R package \texttt{TCGAbiolinks}, including RNA-seq, DNA methylation (Illumina HumanMethylation450K), and clinical records. After preprocessing, we retained 811 samples with complete data for age, methylation, gene expression, and covariates (race and AJCC pathologic stage). The RNA-seq counts for MKI67 were normalized using DESeq2’s variance-stabilizing transformation (vst) to address heteroscedasticity, while DNA methylation $\beta$-values were converted to M-values via logit transformation ($M = \log_2 (\beta / (1-\beta))$) as recommended by \cite{du2010comparison} to improve normality. Age at diagnosis (range: 26-89 years; mean ± SD: 58.0 ± 13.3) served as the exposure variable, and 299,813 CpG sites with non-missing values across all samples were included as potential mediators. The transformed MKI67 expression levels constituted the outcome variable, adjusted for race and tumor stage. 
	
	Following \cite{guo2022high} we first carry out a screening step to retain the top 1000 potential mediators by ranking the absolute value of the product of two correlations—the correlation between exposure variable
	and each element of mediators, and between outcome and each element of mediators.
	This indeed is a marginal screening procedure based on Pearson correlation proposed by \cite{fan2008sure}. They show that for linear models, under some regularity conditions, the
	screening procedure possesses a sure screening property. We then employ our proposed Factor-Adjusted Debiased Mediation Testing (FADMT) method to statistically test the mediation effects of individual CpG sites. Under a prespecified false discovery rate (FDR) level of 0.1, we identify DNA methylation sites with significant mediation effects and further analyze their potential biological mechanisms.
	
	In Table \ref{tab:empirical}, we report the summary results on the three selected
	mediators by our method. Three CpG sites were identified as significant mediators of age-associated MKI67 expression (FDR-adjusted $p$-value < 0.1). The strongest mediation signal is observed at an intergenic site ch.2.71774667F (chr2:71921159), located near LINC01807. This site shows negative age-to-methylation and methylation-to-expression effects, resulting in a positive total mediation effect ($\wh{\alpha} \times \wh{\gamma} = 0.0026$).  Biologically, this finding suggests that age-related hypomethylation at this locus might lead to the downregulation of LINC01807. Given that long non-coding RNAs have been implicated in the regulation of chromatin architecture and gene expression \citep{li2016enhancers}, the observed effect is consistent with the hypothesis that reduced LINC01807 expression may relieve repression on the MKI67 promoter, thereby enhancing cellular proliferation.
	
	The CpG site cg10404601 (chr19:8468449), located in the 3' UTR of RAB11B, exhibits a negative age-to-methylation effect and a positive methylation-to-expression association, yielding a negative mediation effect ($\wh{\alpha} \times \wh{\gamma} = -0.0037$). This pattern suggests that age-associated hypomethylation may reduce RAB11B expression, ultimately leading to suppressed MKI67 levels and decreased cell proliferation. Supporting the biological relevance of this locus, RAB11B-AS1, a natural antisense transcript of RAB11B, has been shown to promote angiogenesis and metastasis in breast cancer by enhancing the expression of angiogenic factors such as VEGFA and ANGPTL4 in a hypoxia-inducible manner \citep{niu2020hif2}. This evidence highlights the potential regulatory importance of the RAB11B region in breast cancer progression.

	Lastly, cg24461063 (chr12:124971775), located in the gene body of NCOR2, shows a moderate negative effect of age on methylation and a stronger positive effect of methylation on expression, resulting in a negative mediation effect ($\wh{\alpha} \times \wh{\gamma} = -0.0046$). This is consistent with NCOR2’s known function as an estrogen receptor (ER$\alpha$) corepressor: hypomethylation may lead to upregulation of NCOR2, which in turn inhibits ER signaling and attenuates proliferative activity \citep{tsai2022screening}. 	
	Collectively, these findings provide biologically plausible insights into the epigenetic regulation of proliferation in breast cancer. The observed mediation effects are supported by previous studies on cancer epigenomics in \cite{baylin2011decade, rakyan2011epigenome} and highlight the complex interplay between aging, DNA methylation, and gene expression.

	\begin{table}[htbp]
		\centering
		\caption{Summary of selected CpGs with significant mediation effects}
		\label{tab:empirical}
		\fontsize{10}{12}\selectfont  
		\begin{threeparttable}
			\setlength{\tabcolsep}{8pt}  
			\begin{tabular}{lcccccc} 
				\hline
				CpGs & Chromosome & Neighboring gene  & $\wh{\alpha}_{mj}$ & $\wh{\gamma}_{sj}$ &$P_{\max,j}$ & adj-P\\ \hline
				ch.2.71774667F &chr2:71921159 & LINC01807	 & -0.2006 & -0.0130 &8.53$\times 10^{-7}$ &0.0001\\
				cg10404601 &chr19:8468449 &RAB11B  & -0.4973  & 0.0074  &1.26$\times 10^{-4}$ &0.0631 \\
				cg24461063 &chr12:124971775	&NCOR2	 & -0.2070 & 0.0223 &2.12$\times 10^{-4}$ &0.0705\\ \hline
			\end{tabular}
			\begin{tablenotes} 
				\footnotesize
				\item Notes: $P_{\max,j}$ represents the unadjusted $p$-value for the mediation effect; adj-P denotes FDR-adjusted $p$-value.
			\end{tablenotes}
		\end{threeparttable}
	\end{table}

	\subsection{Financial Application}
	
	We next provide a financial case study to illustrate the applicability of our mediation effect testing framework in financial scenarios. We study whether the launch of the Shanghai--Hong Kong Stock Connect affects firms' long-run idiosyncratic risk through changes in corporate fundamentals. Stock Connect is widely viewed as a quasi-natural experiment of partial equity market liberalization in China. Existing studies document that this program can affect market quality, corporate policies, and firm risk-related outcomes through through changes in investor base and information environment \citep{ma2019effect,xu2020stock,xiong2021does,li2024stock}. However, existing mechanism analyses typically consider only a small number of candidate channels. In contrast, we study the mechanism question in a high-dimensional mediation setting by jointly analyzing a broad set of corporate accounting and valuation indicators with strong dependence. Our empirical design adopts a two-period difference-in-differences framework around the policy date $t_0=$ 2014/11/17, when the first batch of 568 stocks became eligible for northbound trading.

For each stock $i$, we measure the change in firm-specific risk as $
	\Delta y_i=\sigma^{idio}_{i,post}-\sigma^{idio}_{i,pre}$, 
where $\sigma^{idio}_{i,pre}$ and $\sigma^{idio}_{i,post}$ are computed from daily Capital Asset Pricing Model (CAPM) residuals in the pre window (2014/05/22--2014/11/14) and the post window (2014/11/17--2015/05/14), each spanning 120 trading days (about six months)\footnote{We remove the common market component by estimating the CAPM model
	$r_{it}=\alpha_i+\beta_i r_{mkt,t}+\varepsilon_{it}$ and then defining
	$\sigma^{idio}_{i,\cdot}$ as the standard deviation of the fitted residuals $\hat{\varepsilon}_{it}$ within the corresponding window. This construction isolates firm-specific fluctuations after controlling for market-wide movements, following the standard market-model definition of idiosyncratic volatility \citep{sharpe1964capital}.}. The treatment indicator $s_i$ equals one if stock $i$ is included in Stock Connect at $t_0$, and zero otherwise.

To mitigate selection on observables, we implement a propensity score matching (PSM) design with 1:1 matching without replacement, using a caliper of $0.2$ times the standard deviation of the propensity score \citep{rosenbaum1985constructing,austin2011optimal}. Potential controls are Shanghai-listed stocks that were not included in Stock Connect during the post window\footnote{Restricting to the same exchange helps remove confounding exchange-level differences.}. We match on industry, market capitalization, book-to-market (all measured at 2014Q3), and pre-window idiosyncratic volatility. From 403 potential controls, the procedure yields 129 valid matched pairs (258 observations) after removing 13 pairs with fewer than 50 trading-day observations in the post window. Daily returns, industry classifications, and accounting/valuation indicators are obtained from the CSMAR database.

We first assess the total effect of stock connect on firms' long-run idiosyncratic risk by regressing $\Delta y$ on treatment indicator in the matched sample. The estimated coefficient on Stock Connect is $-0.0016$ with a $p$-value of $0.0618$, suggesting that connected stocks experienced a modest decline in idiosyncratic risk relative to otherwise comparable non-connected firms. We further examine the mediation effect to see which corporate fundamentals might affect firm-specific risk through Stock Connect. To this end, we consider 316 quarterly accounting and valuation indicators as candidate mediators and define the mediator change as $\Delta M_{ij}=M_{ij}(\text{2015Q1})-M_{ij}(\text{2014Q3})$, using percentage changes for level (currency-denominated) variables and simple differences otherwise to mitigate scale effects. These financial indicators are highly correlated, making this application a natural setting for our FADMT procedure. At an FDR level of $q=0.1$, our method identifies seven significant mediators. Table~\ref{tab:stockconnect_mediators} summarizes the results.

\begin{table}[htbp]
	\centering
	\caption{Summary of selected mediators for Stock Connect and idiosyncratic risk change}
	\label{tab:stockconnect_mediators}
	\fontsize{9}{11}\selectfont
	\setlength{\tabcolsep}{4pt}
	\begin{threeparttable}
		\begin{tabular}{l p{6cm} c c c c}
			\hline
			Mediator & Description & $\wh{\alpha}_{mj}$ & $\wh{\gamma}_{sj}$ & $P_{\max,j}$ & adj-P\\
\hline
ROE  & Return on equity (parent) 
& $-1.52{\times}10^{-3}$ & $-0.4606$ & $7.73{\times}10^{-4}$ & $0.0701$ \\
PEIC & Equity-to-invested capital ratio (parent) 
& $-8.29{\times}10^{-4}$ & $-0.3872$ & $4.31{\times}10^{-3}$ & $0.0879$ \\
PS   & Price-to-sales ratio 
& $8.22{\times}10^{-4}$ & $-0.3519$ & $4.43{\times}10^{-3}$ & $0.0879$ \\
TAPS & Tangible assets per share 
& $5.44{\times}10^{-4}$ & $0.3479$ & $4.91{\times}10^{-3}$ & $0.0879$ \\
LPS  & Liabilities per share 
& $5.16{\times}10^{-4}$ & $0.3941$ & $6.50{\times}10^{-3}$ & $0.0879$ \\
EPS  & Earnings per share (parent) 
& $-9.88{\times}10^{-4}$ & $-0.6269$ & $6.24{\times}10^{-3}$ & $0.0879$ \\
TDTA & Liability-to-tangible assets ratio 
& $4.93{\times}10^{-4}$ & $0.3356$ & $6.69{\times}10^{-3}$ & $0.0879$ \\
\hline
		\end{tabular}
		\begin{tablenotes}[flushleft]
			\footnotesize
			\item Notes: The suffix ``parent" denotes figures attributable specifically to the parent company. $\wh{\alpha}_{mj}$ denotes the effect from mediator change to $\Delta y$, and $\wh{\gamma}_{sj}$ denotes the effect from stock connect eligibility to mediator change. $P_{\max,j}$ represents the unadjusted $p$-value for the mediation effect; adj-P denotes FDR-adjusted $p$-value.
		\end{tablenotes}
	\end{threeparttable}
\end{table}

The selected mediators align with standard channels emphasized in the finance literature. First, leverage-related variables (liability-to-tangible assets ratio, equity-to-invested capital ratio, liabilities per share, and tangible assets per share) suggest a capital-structure channel: changes in financing conditions and risk-bearing can affect firm-specific risk. Second, profitability measures (ROE and EPS) capture an operating-performance channel, as improvements in firm fundamentals are typically associated with lower idiosyncratic volatility. Third, the price-to-sales ratio reflects a valuation/discount-rate channel, consistent with the idea that investor base changes can reshape pricing and firm-specific risk. This case study highlights that our method remains applicable and yields interpretable signals in dependent, high-dimensional financial environments.

	\section{Conclusion}
	In this paper, we propose a novel Factor-Adjusted Debiased Mediation Testing (FADMT)
	framework to address the long-standing challenges posed by high-dimensional
	dependence among mediators in mediation analysis.
	By integrating approximate factor modeling with debiased Lasso inference, we
	effectively decouple pervasive correlation patterns driven by unobserved common
	factors from idiosyncratic variations.
	Our theoretical results establish the asymptotic validity of the proposed tests,
	and extensive simulations demonstrate that FADMT substantially outperforms
	standard debiased mediation testing methods and existing approaches, particularly
	under strong inter-mediator correlations.
	Application to TCGA-BRCA multi-omics data and a financial stock connect study further illustrate the practical
	relevance of the method in uncovering meaningful mediation effects.

	\phantomsection\label{supplementary-material}
\bigskip

\clearpage

	\bibliography{bibliography.bib}
	
	\newpage
	\appendix
	
		\begin{center}
	{\LARGE\bf Supplementary Material}
\end{center}
	
	This supplementary document provides additional methodological and technical details
	that complement the main paper. It is organized as follows.
	Section~\ref{supp:convex} presents the convex optimization approach for constructing
	the decorrelating matrix $\wt{\bOmega}$.
	Section~\ref{sec:add_sim} reports additional simulation results supplementing those
	in the main text.
	Section~\ref{sec:proofs} collects proofs of the main theoretical results.
	Section~\ref{sec:lemmas} states and proves a set of technical lemmas used
	throughout the analysis.
	
	\section{Convex Optimization Approach for Constructing the Decorrelating Matrix \texorpdfstring{$\wt{\bOmega}$}{Omega-tilde}} \label{supp:convex}

		In this section, we introduce the idea and general procedure of constructing the decorrelating matrix $\wt{\bOmega} $ using a convex optimization approach, which was originally proposed in \cite{javanmard2014confidence}. The goal of this approach is to obtain a matrix $\wt{\bOmega}$ that effectively reduces both the bias and variance of the coordinates of the debiased estimator $\wh{\balpha}_m^d$. Specifically, the matrix $\wt{\bOmega} = (\wt{\bomega}_1, \wt{\bomega}_2, \dots, \wt{\bomega}_p)\t \in \RR^{p \times p}$ is constructed by solving a sequence of convex optimization problems, where each column is obtained as the solution to the following program:
	\begin{align*}
		\text { minimize } & \wt{\bomega} \t \widehat{\bSigma} \wt{\bomega}, \\
		\text { subject to } & \left\|\wh{\bSigma} \wt{\bomega}- \be_i\right\|_{\infty} \leq \mu,
	\end{align*}
	where $\wh{\bSigma}$ is the empirical covariance matrix, $\be_i \in \mathbb{R}^p$ denotes the $i$-th standard basis vector, and $\mu$ is a small positive tuning parameter that controls the approximation accuracy. If the optimization problem for any $i$ is not feasible, we follow the practice in \cite{javanmard2014confidence} and set $\wt{\bOmega} = \bI_{p \times p}$. In our implementation, we use the R package provided by the authors, \texttt{sslasso},\footnote{\url{https://web.stanford.edu/~montanar/sslasso/code.html}} to compute the matrix $\wt{\bOmega} $ efficiently.
	
	For comparison, in Section~\ref{node} we describe an alternative approach to estimating $\wt{\bOmega} $ based on nodewise Lasso regression, along with its theoretical properties. 
	
	\section{Additional Simulation Results} \label{sec:add_sim}
	\subsection{Type I error \& power} \label{app:sim_typeI_power}
	Table~\ref{tab:app_type1_power} additionally reports the empirical Type~I error rates and power for testing $\alpha_{mj}$ and $\gamma_{sj}$ across the five covariance models under the nominal level $0.05$.
	Across all five models, the debiased inference for the high-dimensional coefficients
	$\alpha_{mj}$ exhibits mild conservativeness in finite samples: the resulting
	$p$-values tend to be super-uniform under the null, leading to empirical Type~I
	error rates slightly below the nominal level $0.05$.
	This conservativeness is mainly driven by (i) first-stage factor estimation error,
	which propagates into the plug-in pseudo-design, and (ii) regularization bias from
	estimating the precision matrix $\wt{\bOmega}$, both of which can inflate the
	estimated standard errors and thus yield slightly larger $p$-values.
	
	Comparing implementations, the convex-optimization method typically
	outperforms the nodewise alternative, delivering higher power with Type~I error control.
	This improvement is consistent with the fact that the convex approach directly
	targets a variance-minimization criterion for the debiasing direction, thereby
	reducing estimator variance and improving finite-sample efficiency.
	
\begin{table}[t]
	\centering
	\caption{Empirical Type I error and power for testing $\alpha_{mj}$ and $\gamma_{sj}$ (nominal level $0.05$).}
	\label{tab:app_type1_power}
	\begin{threeparttable}
		\setlength{\tabcolsep}{4pt}
		\footnotesize
		
		\begin{tabular}{c *{10}{S[table-format=1.4]}}
			\toprule
			& \multicolumn{8}{c}{Testing $\alpha_{mj}$} 
			& \multicolumn{2}{c}{Testing $\gamma_{sj}$} \\
			\cmidrule(lr){2-9}\cmidrule(lr){10-11}
			Model
			& \multicolumn{2}{c}{FADMT (nodewise)}
			& \multicolumn{2}{c}{FADMT (convex)}
			& \multicolumn{2}{c}{DMT (nodewise)}
			& \multicolumn{2}{c}{DMT (convex)}
			& \multicolumn{2}{c}{Common} \\
			\cmidrule(lr){2-3}\cmidrule(lr){4-5}\cmidrule(lr){6-7}\cmidrule(lr){8-9}\cmidrule(lr){10-11}
			& \multicolumn{1}{c}{Type I} & \multicolumn{1}{c}{Power}
			& \multicolumn{1}{c}{Type I} & \multicolumn{1}{c}{Power}
			& \multicolumn{1}{c}{Type I} & \multicolumn{1}{c}{Power}
			& \multicolumn{1}{c}{Type I} & \multicolumn{1}{c}{Power}
			& \multicolumn{1}{c}{Type I} & \multicolumn{1}{c}{Power} \\
			\midrule
			1
			& 0.0336 & 0.9705
			& 0.0335 & 0.9780
			& 0.0447 & 1.0000
			& 0.0636 & 1.0000
			& 0.0516 & 1.0000 \\
			2
			& 0.0372 & 1.0000
			& 0.0401 & 1.0000
			& 0.0408 & 1.0000
			& 0.0389 & 1.0000
			& 0.0489 & 0.9996 \\
			3
			& 0.0408 & 0.9995
			& 0.0435 & 1.0000
			& 0.0456 & 1.0000
			& 0.0488 & 1.0000
			& 0.0520 & 1.0000 \\
			4
			& 0.0392 & 1.0000
			& 0.0407 & 1.0000
			& 0.0427 & 1.0000
			& 0.0464 & 1.0000
			& 0.0517 & 1.0000 \\
			5
			& 0.0338 & 0.9510
			& 0.0364 & 0.9950
			& 0.0466 & 1.0000
			& 0.0441 & 1.0000
			& 0.0490 & 1.0000 \\
			\bottomrule
		\end{tabular}
		
		\begin{tablenotes}[flushleft]
			\footnotesize
			\item \textit{Notes:} Entries are averaged over 200 replications with $(n,p)=(300,500)$. 
			Type I error is the empirical rejection probability under the null at nominal level $0.05$, and power is the empirical rejection probability under the alternative.
			For testing $\gamma_{sj}$, FADMT and DMT yield identical results (for both nodewise and convex implementations) across all five models; hence we report a single set of results under ``Common''.
		\end{tablenotes}
	\end{threeparttable}
\end{table}

\subsection{Empirical convergence of null $p$-values} \label{app:empirical}
To complement these finite-sample summaries, we provide an empirical check of Assumption~6 in the main text,
which posits that the empirical process of the null $p$-values for testing $\gamma_{sj}$ converges to its
theoretical limit.
We conduct this diagnostic under Model~2 (the same simulation setting as in Table~\ref{tab:app_type1_power}).

Panel~(\ref{fig:panel_a}) of Figure~\ref{fig:empirical_pvals} plots the empirical cumulative distribution
function (CDF) of the null $p_{\gamma_{sj}}$, defined as
\[
\widehat F_\gamma(t)
:=
\frac{\#\{j\in\mathcal H_0^\gamma:\ p_{\gamma_{sj}}\le t\}}{|\mathcal H_0^\gamma|},
\qquad t\in[0,1],
\]
and compares it with the theoretical reference line $t$.
The close agreement between $\widehat F_\gamma(t)$ and $t$ provides empirical evidence that the null
$p_{\gamma_{sj}}$ behave approximately uniformly, consistent with the claimed empirical-process convergence
under factor-structured dependence.

Panel~(\ref{fig:panel_b}) further examines the tail behavior of the combined statistic
$p_{\max,j}$ for $j\in\mathcal H_{00}$.
We plot the empirical rejection proportion
\[
\widehat F_{\max,00}(t)
:=
\frac{\#\{j\in\mathcal H_{00}:\ p_{\max,j}\le t\}}{|\mathcal H_{00}|},
\qquad t\in[0,0.1],
\]
against the theoretical benchmark $t^2$, which corresponds to the product-form tail probability under
asymptotic independence of the two component $p$-values.
Panel~(\ref{fig:panel_c}) reports the same comparison for $j\in\mathcal H_{01}$ (only the mediation--outcome
effect is null), where the appropriate reference line is $t$.

The tail diagnostics highlight the practical difficulty of inference on the high-dimensional vector
$\balpha_m$ in the presence of strong inter-mediator correlation.
In particular, when the nuisance pathway is active (e.g., $\gamma_{sj}\neq0$), DMT may exhibit visible tail
distortions near the origin, especially under the convex-optimization construction of the decorrelating
matrix, which can translate into inflated discoveries and unstable FDR control.
The nodewise implementation typically alleviates this issue but may still show mild deviations in the tail.
In contrast, the proposed FADMT tracks the theoretical benchmarks well in the critical tail region and
tends to be mildly conservative away from the tail, aligning with the Type~I error patterns in
Table~\ref{tab:app_type1_power}.
Overall, these results support Assumption~6 empirically and illustrate that factor adjustment can partially
mitigate the inferential challenge caused by strong mediator dependence.

\begin{figure}[htbp]
	\centering
	\begin{subfigure}[t]{0.6\linewidth}
		\centering
		\includegraphics[width=\linewidth]{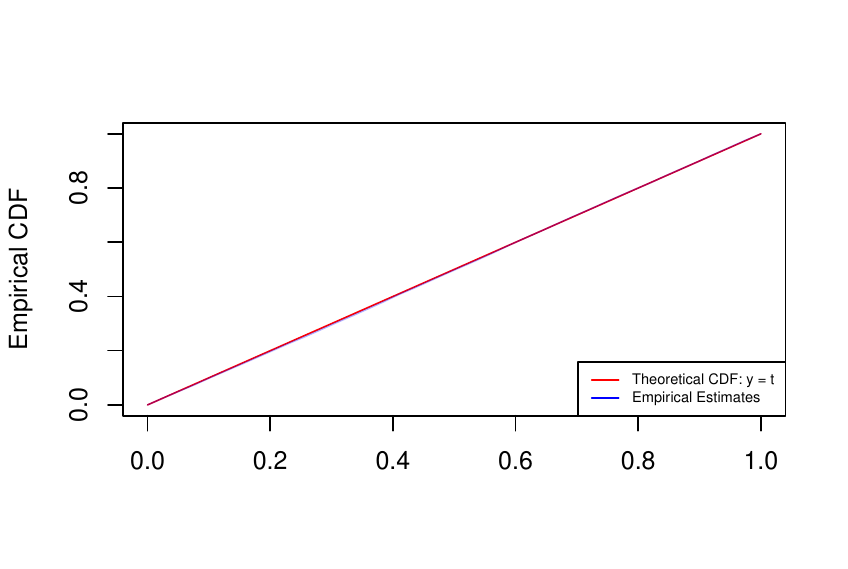}
		\caption{Empirical CDF of null $p_{\gamma_{sj}}$ with reference line $t$.}
		\label{fig:panel_a}
	\end{subfigure}

	\begin{subfigure}[t]{0.6\linewidth}
		\centering
		\includegraphics[width=\linewidth]{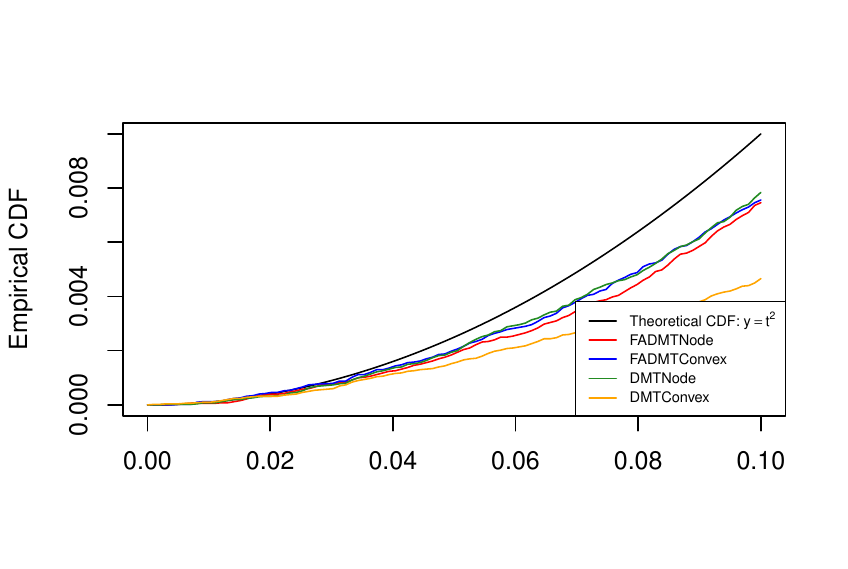}
		\caption{Empirical CDF of $p_{\max,j}$ for $j\in\mathcal H_{00}$ (tail region $t\in[0,0.1]$) with reference line $t^2$.}
		\label{fig:panel_b}
	\end{subfigure}

	\begin{subfigure}[t]{0.6\linewidth}
		\centering
		\includegraphics[width=\linewidth]{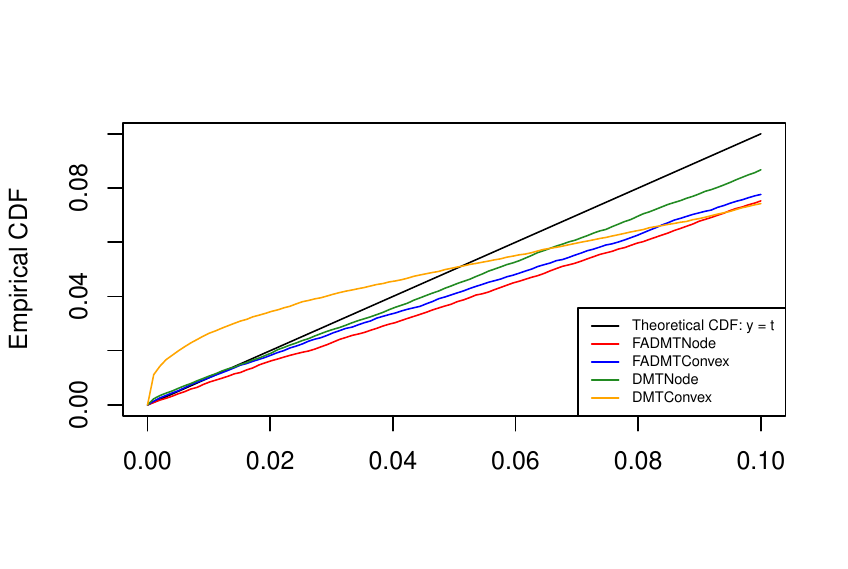}
		\caption{Empirical CDF of $p_{\max,j}$ for $j\in\mathcal H_{01}$ (tail region $t\in[0,0.1]$) with reference line $t$.}
		\label{fig:panel_c}
	\end{subfigure}
	
	\caption{Empirical validation of null $p$-value behavior under Model~2.}
	\label{fig:empirical_pvals}
\end{figure}

\newpage
\section{Proofs for Results in the Main Text}
\label{sec:proofs}

	\subsection{Proof of Proposition \ref{prop1}}
	
	\begin{proof}
		By the definition of $\wh R_{ji}$ and the OLS estimator
		\(
		\wh\gamma_{sj}=(\bs^\top\bs)^{-1}\bs^\top(\bs\gamma_{sj}+\bR_j),
		\)
		we have, for each $(j,i)$,
		\begin{align}
			|\wh R_{ji}-R_{ji}|
			&=|s_i(\gamma_{sj}-\wh\gamma_{sj})|   \notag \\
			&=\Bigl|s_i\Bigl(\gamma_{sj}-(\bs^\top\bs)^{-1}\bs^\top(\bs\gamma_{sj}+\bR_j)\Bigr)\Bigr| \notag\\
			&=\bigl|s_i(\bs^\top\bs)^{-1}\bs^\top\bR_j\bigr| \notag\\
			&\le \Bigl|\Bigl(\frac{\bs^\top\bs}{n}\Bigr)^{-1}\Bigr|
			\Bigl|\frac{\bs^\top\bR_j}{n}\Bigr|\,
			|s_i|.
			\label{eq:Rhat-R-pointwise}
		\end{align}
		Here $\bR_j=(R_{j1},\ldots,R_{jn})^\top\in\mathbb R^n$. Taking maxima in
		\eqref{eq:Rhat-R-pointwise} yields
		\begin{equation}
			\label{eq:Rhat-R-max}
			\|\wh{\bR}-\bR\|_{\max}
			:=\max_{i,j}|\wh R_{ji}-R_{ji}|
			\le
			\Bigl|\Bigl(\frac{\bs^\top\bs}{n}\Bigr)^{-1}\Bigr|
			\Bigl(\max_{j\in[p]}\Bigl|\frac{\bs^\top\bR_j}{n}\Bigr|\Bigr)
			\Bigl(\max_{i\in[n]}|s_i|\Bigr).
		\end{equation}
		
		Under Assumption~1, $\{s_i\}_{i=1}^n$ are i.i.d.\ sub-Gaussian. Standard concentration implies
		\[
		\Bigl|\Bigl(\frac{\bs^\top\bs}{n}\Bigr)^{-1}\Bigr|=O_P(1),
		\qquad
		\max_{i\in[n]}|s_i|=O_P(\sqrt{\log n}).
		\]
		It remains to bound the middle term. Note that
		\begin{align*}
			\max_{j\in[p]}\Bigl|\frac{\bs^\top\bR_j}{n}\Bigr|
			&=\max_{j\in[p]}\Bigl|\frac1n\sum_{i=1}^n s_i R_{ji}\Bigr|
			=\max_{j\in[p]}\Bigl|\frac1n\sum_{i=1}^n s_i(\bb_j^\top\bfv_i+u_{ji})\Bigr|.
		\end{align*}
		For any fixed $j$, the summands $s_i(\bb_j^\top\bfv_i+u_{ji})$ are independent sub-Gaussian,
		so by Bernstein's inequality, for any $t>0$,
		\[
		\PP\!\left(
		\Bigl|\frac1n\sum_{i=1}^n s_i(\bb_j^\top\bfv_i+u_{ji})\Bigr|\ge t
		\right)
		\le 2\exp\!\left(-c n\frac{t^2}{K^2}\right),
		\]
		for some constants $c>0$ and $K>0$ depending on sub-Gaussian norms. By the union bound over
		$j=1,\ldots,p$,
		\[
		\PP\!\left(
		\max_{j\in[p]}
		\Bigl|\frac1n\sum_{i=1}^n s_i(\bb_j^\top\bfv_i+u_{ji})\Bigr|\ge t
		\right)
		\le 2p\exp\!\left(-c n\frac{t^2}{K^2}\right).
		\]
		Taking $t=C\sqrt{\frac{\log p}{n}}$ gives
		\(
		\max_{j\in[p]}\bigl|\frac{\bs^\top\bR_j}{n}\bigr|
		=O_P\!\bigl(\sqrt{\frac{\log p}{n}}\bigr).
		\)
		
		Combining the bounds in \eqref{eq:Rhat-R-max}, we conclude
		\[
		\|\wh{\bR}-\bR\|_{\max}
		=
		O_P\!\left(\sqrt{\frac{\log n\,\log p}{n}}\right),
		\]
		which completes the proof of Proposition~\ref{prop1}.
	\end{proof}

	\subsection{Proof of Proposition \ref{lemma:factor}}
	\begin{proof}
		The proof is an adaption of that in \cite{fan2013large}.
		Note that
		\[
		\wh{\bR}=\bF\bB^\top+\wt{\bU},
		\qquad
		\wt{\bU}=\bU+(\wh{\bR}-\bR).
		\]
		Let $\delta_{ji}$ be the $(j,i)$th element of $\wh{\bR}-\bR$, so that
		$\wt u_{ji}=u_{ji}+\delta_{ji}$.
		Using expression (A.1) of \cite{bai2003inferential}, we have the identity
		\begin{equation}
			\label{identity}
			\wh{\bfv}_i-\bH\bfv_i
			=
			\Bigl(\frac{\bV}{p}\Bigr)^{-1}
			\left\{
			\frac{1}{n}\sum_{s=1}^n \wh{\bfv}_s \frac{\EE(\bu_s^\top\bu_i)}{p}
			+
			\frac{1}{n}\sum_{s=1}^n
			\Bigl(
			\wh{\bfv}_s\wt\zeta_{si}
			+\wh{\bfv}_s\wt\eta_{si}
			+\wh{\bfv}_s\wt\xi_{si}
			\Bigr)
			\right\},
		\end{equation}
		where $\bV\in\RR^{K\times K}$ is the diagonal matrix of the largest $K$
		eigenvalues of $n^{-1}\wh{\bR}\wh{\bR}^\top$.
		By Lemma~\ref{lemma:vnorm}, $\|(\bV/p)^{-1}\|=O_P(1)$.
		
		Define
		\begin{align}
			\wt\zeta_{si}
			&=\frac{\wt{\bu}_s^\top\wt{\bu}_i}{p}
			-\EE\!\left(\frac{\bu_s^\top\bu_i}{p}\right)
			=\left(\frac{\bu_s^\top\bu_i}{p}-\EE\!\left(\frac{\bu_s^\top\bu_i}{p}\right)\right)
			+\left(\frac{\bu_s^\top\bdelta_i}{p}+\frac{\bdelta_s^\top\bu_i}{p}
			+\frac{\bdelta_s^\top\bdelta_i}{p}\right)
			=\zeta_{si}+\zeta_{si}^*, \label{eq:zeta_decomp}\\
			\wt\eta_{si}
			&=\frac{\bfv_s^\top\sum_{j=1}^p \bb_j \widetilde{u}_{ji}}{p}
			=\frac{\bfv_s^\top\sum_{j=1}^p \bb_j u_{ji}}{p}
			+\frac{\bfv_s^\top\sum_{j=1}^p \bb_j \delta_{ji}}{p}
			=\eta_{si}+\eta_{si}^*, \label{eq:eta_decomp}\\
			\wt\xi_{si}
			&=\frac{\bfv_i^\top\sum_{j=1}^p \bb_j \widetilde{u}_{js}}{p}
			=\frac{\bfv_i^\top\sum_{j=1}^p \bb_j u_{js}}{p}
			+\frac{\bfv_i^\top\sum_{j=1}^p \bb_j \delta_{js}}{p}
			=\xi_{si}+\xi_{si}^*. \label{eq:xi_decomp}
		\end{align}
		
		\medskip
		\noindent\textbf{(a) Convergence rate of factors.}
		Using $(a+b+c+d)^2\le 4(a^2+b^2+c^2+d^2)$ and \eqref{identity}, we obtain
		\begin{align*}
			\max_{k\in[K]}\frac{1}{n}\sum_{i=1}^n
			\bigl|(\wh{\bfv}_i-\bH\bfv_i)_k\bigr|^2
			&\le
			C\max_{k\in[K]}\frac{1}{n}\sum_{i=1}^n
			\left[
			\frac{1}{n}\sum_{s=1}^n \wh f_{sk}\frac{\EE(\bu_s^\top\bu_i)}{p}
			\right]^2 \\
			&\quad+
			C\max_{k\in[K]}\frac{1}{n}\sum_{i=1}^n
			\left[
			\frac{1}{n}\sum_{s=1}^n \wh f_{sk}\wt\zeta_{si}
			\right]^2 \\
			&\quad+
			C\max_{k\in[K]}\frac{1}{n}\sum_{i=1}^n
			\left[
			\frac{1}{n}\sum_{s=1}^n \wh f_{sk}\wt\eta_{si}
			\right]^2 \\
			&\quad+
			C\max_{k\in[K]}\frac{1}{n}\sum_{i=1}^n
			\left[
			\frac{1}{n}\sum_{s=1}^n \wh f_{sk}\wt\xi_{si}
			\right]^2 .
		\end{align*}
		Each of the four terms on the right-hand side is bounded in Lemma~\ref{lemma:8},
		which yields
		\[
		\max_{k\in[K]}\frac{1}{n}\sum_{i=1}^n
		\bigl|(\wh{\bfv}_i-\bH\bfv_i)_k\bigr|^2
		=
		O_P\!\left(
		\frac{1}{n}+\bigl(1/\sqrt{p}+r_{n,p}+r_{n,p}^2\bigr)^2
		\right).
		\]
		
		\medskip
		\noindent\textbf{(b)}
		Part (b) follows from part (a) and the bound
		\[
		\frac{1}{n}\sum_{i=1}^n\|\wh{\bfv}_i-\bH\bfv_i\|^2
		\le
		K\max_{k\le K}\frac{1}{n}\sum_{i=1}^n
		\bigl|(\wh{\bfv}_i-\bH\bfv_i)_k\bigr|^2.
		\]
		
		\medskip
		\noindent\textbf{(c) Convergence rate of idiosyncratic components.}
		By definition,
		\begin{align*}
			u_{ji}-\wh u_{ji}
			&=R_{ji}-\wh R_{ji}+\wh{\bb}_j^\top\wh{\bfv}_i-\bb_j^\top\bfv_i \\
			&=R_{ji}-\wh R_{ji}
			+\bb_j^\top\bH^\top(\wh{\bfv}_i-\bH\bfv_i)
			+(\wh{\bb}_j^\top-\bb_j^\top\bH^\top)\wh{\bfv}_i
			+\bb_j^\top(\bH^\top\bH-\bI_K)\bfv_i .
		\end{align*}
		Using $(a+b+c+d)^2\le 4(a^2+b^2+c^2+d^2)$, we have
		\begin{align*}
			\max_{j\in[p]}\frac{1}{n}\sum_{i=1}^n |u_{ji}-\wh u_{ji}|^2
			&\le
			4\max_{j\in[p]}\frac{1}{n}\sum_{i=1}^n|R_{ji}-\wh R_{ji}|^2 \\
			&\quad+
			4\max_{j\in[p]}\|\bb_j^\top\bH^\top\|^2\cdot
			\frac{1}{n}\sum_{i=1}^n\|\wh{\bfv}_i-\bH\bfv_i\|^2 \\
			&\quad+
			4\max_{j\in[p]}\|\wh{\bb}_j^\top-\bb_j^\top\bH^\top\|^2\cdot
			\frac{1}{n}\sum_{i=1}^n\|\wh{\bfv}_i\|^2 \\
			&\quad+
			4\max_{j\in[p]}\|\bb_j\|^2\cdot
			\frac{1}{n}\sum_{i=1}^n\|\bfv_i\|^2\,
			\|\bH^\top\bH-\bI\|_F^2 .
		\end{align*}
		By Proposition~\ref{prop1} and Lemma~\ref{lemma:9},
		\[
		\max_{j\in[p]}\frac{1}{n}\sum_{i=1}^n|\wh u_{ji}-u_{ji}|^2
		=
		O_P\!\left(
		\frac{\log p}{n}+\frac{1}{p}+r_{n,p}^2+r_{n,p}^4+r_{n,p}^6
		\right).
		\]
	\end{proof}

	\subsection{Proof of Theorem \ref{the1} }
	\begin{proof}
		We have already derived the expansion of the debiased estimator $\wh{\balpha}_m^d$:
		\begin{align}
			\sqrt{n}\,(\wh{\balpha}_m^d-\balpha_m)
			&=
			\frac{1}{\sqrt{n}}\,\wt{\bOmega}\,\wh{\bU}^\top\bxi
			+
			\sqrt{n}\,(\wt{\bOmega}\wh{\bSigma}_U-\bI)\,(\balpha_m-\wh{\balpha}_m^n)
			\notag\\
			&=: Z+\Delta .
			\label{eq:debiased_expansion}
		\end{align}
		The term $Z$ is Gaussian with covariance
		$\sigma^2\,\wt{\bOmega}\wh{\bSigma}_U\wt{\bOmega}^\top$
		since it is a linear transformation of the Gaussian vector
		$\bxi\sim N(\mathbf 0,\sigma^2\bI)$.
		
		It remains to show that $\Delta$ is asymptotically negligible. Using the 
		$\ell_1$--$\ell_\infty$ product bound,
		\begin{equation}
			\label{eq:Delta_infty_bound}
			\|\Delta\|_\infty
			\le
			\sqrt{n}\,\|\wt{\bOmega}\wh{\bSigma}_U-\bI\|_\infty\,
			\|\wh{\balpha}_m^n-\balpha_m\|_1.
		\end{equation}
		By Lemmas~\ref{nodelambda} and~\ref{nodewisetau},
		\begin{equation}
			\label{eq:OmegaSigma_bound}
			\|\wt{\bOmega}\wh{\bSigma}_U-\bI\|_\infty
			=
			O_P\!\left(\sqrt{\frac{\log p}{n}}+\frac{1}{\sqrt{p}}+r_{n,p}\right).
		\end{equation}
		Furthermore, by Lemma~\ref{b1} and Lemma \ref{1norm},
		\begin{equation}
			\label{eq:alpha_l1_bound}
			\|\wh{\balpha}_m^n-\balpha_m\|_1
			=
			O_P\!\left((s^*+K+1)\Bigl(\sqrt{\frac{\log p}{n}}+\frac{1}{\sqrt{p}}+r_{n,p}\Bigr)\right).
		\end{equation}
		Recall that $r_{n,p}=\sqrt{\frac{\log n\,\log p}{n}}$. Combining
		\eqref{eq:Delta_infty_bound}--\eqref{eq:alpha_l1_bound} yields
		\[
		\|\Delta\|_\infty
		\le
		\sqrt{n}\,O_P\!\bigl((s^*+K+1)\,r_{n,p}^2\bigr).
		\]
		Under Assumption~4, this term is asymptotically negligible. This completes the proof of Theorem~\ref{the1}.
	\end{proof}

\subsection{Proof of Corollary \ref{cor1}}
\begin{proof}
	Define the standardized statistic
	\[
	W_j
	:=
	\frac{\sqrt{n}\,\wh{\alpha}_{mj}^d}
	{\wh{\sigma}\,\bigl[\wt{\bOmega}\wh{\bSigma}_U\wt{\bOmega}^\top\bigr]_{jj}^{1/2}}.
	\]
	Under the stated regularity conditions, $\wh{\sigma}^2$ is a consistent estimator
	of $\sigma^2$. Moreover, $\wh{\bU}$, $\wh{\bSigma}_U$, and $\wt{\bOmega}$ are
	functions of the observed $(\bs,\bM)$ and hence are measurable with respect to
	the filtration $\mathcal F:=\sigma(\bs,\bM)$. By Slutsky's theorem,
	\[
	W_j \mid \mathcal F \ \xrightarrow{d}\ N(0,1),
	\qquad \text{as } n,p\to\infty.
	\]
	
	Therefore, for $t\in[0,1]$,
	\begin{align}
		\Pr(p_{\alpha_{mj}}\le t\mid \mathcal F)
		&=
		\Pr\!\bigl(2\{1-\Phi(|W_j|)\}\le t\mid \mathcal F\bigr)
		\notag\\
		&=
		\Pr\!\bigl(|W_j|\ge \Phi^{-1}(1-t/2)\mid \mathcal F\bigr).
		\label{eq:pval_cond_cdf}
	\end{align}
	For any fixed $t\in[0,1]$, the conditional weak convergence implies
	\[
	\Pr(p_{\alpha_{mj}}\le t\mid \mathcal F)
	\xrightarrow{P}
	2\Bigl\{1-\Phi\bigl(\Phi^{-1}(1-t/2)\bigr)\Bigr\}=t.
	\]
	Since the limiting distribution function $F(t)=t$ is continuous on $[0,1]$, by
	P\'olya's theorem this pointwise convergence strengthens to uniform convergence
	in probability:
	\begin{equation}
		\label{eq:pval_uniform}
		\sup_{t\in[0,1]}
		\Bigl|\Pr(p_{\alpha_{mj}}\le t\mid \mathcal F)-t\Bigr|
		\xrightarrow{P}0.
	\end{equation}
	
	To show that $p_{\alpha_{mj}}$ is asymptotically independent of any
	$\mathcal F$-measurable statistic $T$, fix $t\in[0,1]$ and $y\in\RR$. By the tower
	property,
	\begin{align}
		\Pr(p_{\alpha_{mj}}\le t,\ T\le y)
		&=\E\!\left[\Pr(p_{\alpha_{mj}}\le t,\ T\le y\mid \mathcal F)\right]
		\notag\\
		&=\E\!\left[\mathbf 1(T\le y)\,\Pr(p_{\alpha_{mj}}\le t\mid \mathcal F)\right].
		\label{eq:tower}
	\end{align}
	Using \eqref{eq:pval_uniform}, we have
	$\Pr(p_{\alpha_{mj}}\le t\mid \mathcal F)=t+o_P(1)$ uniformly over $t\in[0,1]$.
	Substituting this into \eqref{eq:tower} yields
	\[
	\Pr(p_{\alpha_{mj}}\le t,\ T\le y)
	=
	\E\!\left[\mathbf 1(T\le y)\,(t+o_P(1))\right]
	=
	t\,\Pr(T\le y)+o(1),
	\]
	which proves the claimed asymptotic independence.
\end{proof}

\subsection{Proof of Theorem \ref{the2}}
\begin{proof}
	We follow the proof framework of Theorem 2 in \cite{dai2022multiple}.
	We first examine the asymptotic behavior of $\wh{\pi}_0^\gamma(\eta)$:
	\begin{align}
		\lim_{p\to\infty}\wh{\pi}_0^\gamma(\eta)
		&=
		\pi_{00}+\pi_{10}
		+\pi_{01}\frac{G_{01}(\eta)}{1-\eta}
		+\pi_{11}\frac{G_{11}(\eta)}{1-\eta}, \label{eq:pi0_limit}\\
		\lim_{p\to\infty}\bigl\{1-\wh{\pi}_0^\gamma(\eta)\bigr\}
		&=
		\pi_{01}+\pi_{11}
		-\pi_{01}\frac{G_{01}(\eta)}{1-\eta}
		-\pi_{11}\frac{G_{11}(\eta)}{1-\eta}. \label{eq:one_minus_pi0_limit}
	\end{align}
	
	Write
	\[
	V(t):=\sum_{j\in\mathcal S_0}\mathbf 1(P_{\max,j}\le t),
	\qquad
	\mathrm{FDP}(t)=\frac{V(t)}{R(t)\vee 1}.
	\]
	Decompose $V(t)=V_{00}(t)+V_{01}(t)+V_{10}(t)$, where
	\[
	V_{ab}(t):=\sum_{j\in\mathcal H_{ab}}\mathbf 1(P_{\max,j}\le t),
	\qquad (a,b)\in\{0,1\}^2.
	\]
	By Corollary~\ref{cor1}, for $j$ with $\alpha_{mj}=0$,
	$p_{\alpha_{mj}}$ is asymptotically $\mathrm{Unif}(0,1)$ conditional on
	$\mathcal F=\sigma(\bs,\bM)$ and is asymptotically independent of any
	$\mathcal F$-measurable quantity. Since $p_{\gamma_{sj}}$ is $\mathcal F$-measurable,
	this yields asymptotic independence between $p_{\alpha_{mj}}$ and $p_{\gamma_{sj}}$
	for $j\in\mathcal H_{00}\cup\mathcal H_{01}$.
	Moreover, by Assumption~\ref{ass:A6}(ii),
	\(
	V_{10}(t)/p\le |\mathcal H_{10}|/p\to 0.
	\)
	
	By the Functional Law of Large Numbers (FLLN) for weakly dependent processes
	(or a generalized Glivenko--Cantelli theorem), and given continuity of the limiting
	distributions, we have the uniform convergences
	\begin{align}
		\lim_{p\to\infty}\sup_{t\in[0,1]}
		\left|
		\frac{V_{00}(t)}{p}-\pi_{00}t^2
		\right|
		&=0, \label{eq:V00_uniform}\\
		\lim_{p\to\infty}\sup_{t\in[0,1]}
		\left|
		\frac{V_{01}(t)}{p}-\pi_{01}t\{1-G_{01}(t)\}
		\right|
		&=0, \label{eq:V01_uniform}\\
		\lim_{p\to\infty}\sup_{t\in[0,1]}
		\left|
		\frac{V_{10}(t)}{p}
		\right|
		&=0. \label{eq:V10_uniform}
	\end{align}
	
	Under the alternative, $p$-values are stochastically smaller than the uniform
	distribution, hence $G_{11}(\eta)/(1-\eta)\le 1$.
	Let
	\[
	c:=G_{01}^{-1}\!\left(\frac{G_{01}(\eta)}{1-\eta}\right),
	\]
	noting that $G_{01}$ is monotone decreasing. For any $\delta>0$ and
	$\delta\le t\le c$, we have
	\begin{align}
		\lim_{p\to\infty}\wh{\pi}_0^\gamma(\eta)\,t^2
		&\ge \pi_{00}t^2, \label{eq:pi0t2_lower}\\
		\lim_{p\to\infty}\bigl\{1-\wh{\pi}_0^\gamma(\eta)\bigr\}\,t
		&=
		\left(
		\pi_{01}+\pi_{11}
		-\pi_{01}\frac{G_{01}(\eta)}{1-\eta}
		-\pi_{11}\frac{G_{11}(\eta)}{1-\eta}
		\right)t
		\ge \pi_{01}t\{1-G_{01}(t)\}. \label{eq:one_minus_pi0t_lower}
	\end{align}
	Therefore,
	\begin{equation}
		\label{eq:fdr1}
		\lim_{p\to\infty}\inf_{\delta\le t\le c}
		\left\{
		\frac{\wh{\pi}_0^\gamma(\eta)t^2+\{1-\wh{\pi}_0^\gamma(\eta)\}t}{R(t)\vee 1}
		-
		\frac{\pi_{00}t^2+\pi_{01}t\{1-G_{01}(t)\}}{R(t)\vee 1}
		\right\}
		\ge 0.
	\end{equation}
	
	Next, observe that
	\begin{align}
		&\lim_{p\to\infty}\sup_{\delta\le t\le c}
		\left|
		\frac{V_{00}(t)+V_{01}(t)+V_{10}(t)}{R(t)\vee 1}
		-
		\frac{p\{\pi_{00}t^2+\pi_{01}t(1-G_{01}(t))\}}{R(t)\vee 1}
		\right| \notag\\
		&\le
		\lim_{p\to\infty}
		\left|
		\frac{p}{R(\delta)\vee 1}
		\right|
		\sup_{\delta\le t\le c}
		\left|
		\frac{V_{00}(t)+V_{01}(t)+V_{10}(t)}{p}
		-\{\pi_{00}t^2+\pi_{01}t(1-G_{01}(t))\}
		\right|
		=0.
		\label{eq:fdr2}
	\end{align}
	Combining \eqref{eq:fdr1} and \eqref{eq:fdr2} yields
	\begin{equation}
		\label{eq:fdr3}
		\lim_{p\to\infty}\inf_{\delta\le t\le c}
		\left\{\widehat{\mathrm{FDP}}(t)-\mathrm{FDP}(t)\right\}
		\ge 0.
	\end{equation}
	
	Let
	\[
	\wh t_{q,\eta}
	=
	\sup\{t:\ \widehat{\mathrm{FDP}}_\eta(t)\le q\}.
	\]
	For typical FDR levels of interest (e.g., $q=0.05$),
	$\wh t_{q,\eta}$ lies in the extreme left tail of the $p$-value distribution,
	hence it is close to $0$ and typically satisfies $\wh t_{q,\eta}<c$.
	Therefore, \eqref{eq:fdr3} implies
	\[
	\liminf_{p\to\infty}
	\left\{\widehat{\mathrm{FDP}}(\wh t_{q,\eta})-\mathrm{FDP}(\wh t_{q,\eta})\right\}
	\ge 0.
	\]
	Since $\widehat{\mathrm{FDP}}(\wh t_{q,\eta})\le q$, it follows that
	\[
	\limsup_{p\to\infty}\mathrm{FDP}(\wh t_{q,\eta})\le q.
	\]
	By Fatou's lemma,
	\[
	\limsup_{p\to\infty}\E\!\left[\mathrm{FDP}(\wh t_{q,\eta})\right]
	\le
	\E\!\left[\limsup_{p\to\infty}\mathrm{FDP}(\wh t_{q,\eta})\right]
	\le q,
	\]
	which proves
	\[
	\limsup_{p\to\infty}\mathrm{FDR}(\wh t_{q,\eta})\le q.
	\]
	The proof of Theorem~\ref{the2} is completed.
\end{proof}

\section{Technical Lemmas} \label{sec:lemmas}

This section collects auxiliary results used in the proof of Theorem~\ref{the1}.
Subsections~\ref{l1} and~\ref{node} provide lemmas controlling the remainder term in the
debiased expansion and establishing the required properties of the decorrelating matrix.
Subsection~\ref{factor} summarizes additional lemmas for factor-model estimation that are
invoked throughout the proofs.

\subsection{Convergence rate of \texorpdfstring{$\ell_1$}{l1} norm}
\label{l1}

In this subsection, we present the lemmas and proofs needed for Theorem~\ref{the1},
focusing on controlling the remainder term
\[
\Delta
=
\sqrt{n}\,(\wt{\bOmega}\wh{\bSigma}_U-\bI)\,(\balpha_m-\wh{\balpha}_m^n),
\]
via the standard $\ell_1$--$\ell_\infty$ inequality.

Let
\[
\wh{\bw}_i
:=
(\wh{\bu}_i^\top,\ s_i,\ \wh{\bfv}_i^\top)^\top\in\RR^{p+K+1},
\qquad
\btheta^*
:=
(\balpha_m^\top,\ \alpha_1,\ \balpha_2^\top)^\top,  
\qquad
\wh{\btheta}
:=
(\wh{\balpha}_m^\top,\ \wh{\alpha}_1, \ \wh{\balpha}_2^\top)^\top .
\]
Define the design matrix and its sample covariance by
\[
\wh{\bW}:=(\wh{\bw}_1,\ldots,\wh{\bw}_n)^\top\in\RR^{n\times(p+K+1)},
\qquad
\wh{\bSigma}_W:=n^{-1}\wh{\bW}^\top\wh{\bW}.
\]
For any index set $\mathcal S\subset\{1,\ldots,p+K+1\}$, define the cone
\[
\mathcal C(\mathcal S,3)
:=
\bigl\{\bv\in\RR^{p+K+1}:\ \|\bv_{\mathcal S^c}\|_1\le 3\|\bv_{\mathcal S}\|_1\bigr\}.
\]
Recall that
\(
S_*=\{j\in[p]:\alpha_{mj}^*\neq 0\}
\)
and
\(
S_\diamond=S_*\cup\{p+1,\ldots,p+K+1\}.
\)
Then
\(
S_\diamond^c=S_*^c=\{j\in[p]:\alpha_{mj}^*=0\}.
\)

We bound the two factors in the product
$\|\Delta\|_\infty\le \sqrt{n}\,\|\wt{\bOmega}\wh{\bSigma}_U-\bI\|_\infty\,
\|\wh{\balpha}_m^n-\balpha_m\|_1$
separately.
Lemma~\ref{b1} establishes the stochastic order of the Lasso tuning parameter $\lambda$
for the initial estimator.
Lemma~\ref{1norm} then provides the $\ell_1$-convergence rate of the initial Lasso estimator
under a Restricted Eigenvalue (RE) condition on the design.
Finally, Lemma~\ref{RSC} verifies that this RE condition holds with high probability.

\begin{lemma}[Order of $\lambda$]
	\label{b1}
	Under Assumptions~1--3, and assuming $\xi_i\sim N(0,\sigma^2)$, we have
	\begin{equation}
		\label{eq:lambda_order}
		\frac{2}{n}\Bigl\|
		\sum_{i=1}^n\bigl(y_i-\wh{\bw}_i^\top\btheta^*\bigr)\wh{\bw}_i
		\Bigr\|_\infty
		=
		O_P\!\left(\sqrt{\frac{\log p}{n}}+\frac{1}{\sqrt{p}}+r_{n,p}\right).
	\end{equation}
\end{lemma}

\begin{proof}
	Let $\bW=(\bU,\bs,\bF)\in\RR^{n\times(p+K+1)}$,
	$\wh{\bW}=(\wh{\bU},\bs,\wh{\bF})\in\RR^{n\times(p+K+1)}$, and
	$\wt{\bW}=(\bU,\bs,\bF\bH^\top)\in\RR^{n\times(p+K+1)}$.
	Write $\wh{\bw}_i^\top$ for the $i$th row of $\wh{\bW}$, and similarly for $\wt{\bw}_i^\top$.
	
	Since $y_i-\wh{\bw}_i^\top\btheta^*=\xi_i$, we can write
	\[
		\frac{2}{n}\Bigl\|
\sum_{i=1}^n\bigl(y_i-\wh{\bw}_i^\top\btheta^*\bigr)\wh{\bw}_i
\Bigr\|_\infty
	=
	\Bigl\|\frac{2}{n}\sum_{i=1}^n \wh{\bw}_i\,\xi_i\Bigr\|_\infty .
	\]
	By the triangle inequality,
	\begin{align}
		\Bigl\|\frac{2}{n}\sum_{i=1}^n \wh{\bw}_i\,\xi_i\Bigr\|_\infty
		&\le
		\Bigl\|\frac{2}{n}\sum_{i=1}^n (\wh{\bw}_i-\wt{\bw}_i)\,\xi_i\Bigr\|_\infty
		+
		\Bigl\|\frac{2}{n}\sum_{i=1}^n \wt{\bw}_i\,\xi_i\Bigr\|_\infty \notag\\
		&=
		2\max_{j\in[p+K+1]}
		\left|\frac{1}{n}\sum_{i=1}^n (\wh w_{ji}-\wt w_{ji})\xi_i\right|
		+
		2\max_{j\in[p+K+1]}
		\left|\frac{1}{n}\sum_{i=1}^n \wt w_{ji}\xi_i\right|.
		\label{eq:lambda_decomp}
	\end{align}

	By Cauchy--Schwarz inequality,
	\[
	\left|\frac{1}{n}\sum_{i=1}^n (\wh w_{ji}-\wt w_{ji})\xi_i\right|
	\le
	\left(\frac{1}{n}\sum_{i=1}^n(\wh w_{ji}-\wt w_{ji})^2\right)^{1/2}
	\left(\frac{1}{n}\sum_{i=1}^n\xi_i^2\right)^{1/2}.
	\]
	Taking maxima over $j\in[p+K+1]$ gives
	\begin{align}
		\max_{j\in[p+K+1]}
		\left|\frac{1}{n}\sum_{i=1}^n (\wh w_{ji}-\wt w_{ji})\xi_i\right|
		&\le
		\max_{j\in[p+K+1]}
		\left(\frac{1}{n}\sum_{i=1}^n(\wh w_{ji}-\wt w_{ji})^2\right)^{1/2}
		\left(\frac{1}{n}\sum_{i=1}^n\xi_i^2\right)^{1/2}.
		\label{eq:perturb_CS}
	\end{align}
	By Proposition~\ref{lemma:factor} (ignoring higher-order terms in $r_{n,p}$),
	\begin{equation}
		\label{eq:Wdiff_rate}
		\max_{j\in[p+K+1]}
		\left(\frac{1}{n}\sum_{i=1}^n(\wh w_{ji}-\wt w_{ji})^2\right)^{1/2}
		=
		O_P\!\left(\sqrt{\frac{\log p}{n}}+\frac{1}{\sqrt{p}}+r_{n,p}\right).
	\end{equation}
	Moreover,
	\(
	\left(\frac{1}{n}\sum_{i=1}^n\xi_i^2\right)^{1/2}=O_P(1).
	\)
	Combining with \eqref{eq:perturb_CS} yields
	\begin{equation}
		\label{eq:perturb_bound}
		\max_{j\in[p+K+1]}
		\left|\frac{1}{n}\sum_{i=1}^n (\wh w_{ji}-\wt w_{ji})\xi_i\right|
		=
		O_P\!\left(\sqrt{\frac{\log p}{n}}+\frac{1}{\sqrt{p}}+r_{n,p}\right).
	\end{equation}
	By Assumption~1 and Bernstein's inequality, for some constant $C>0$,
	\[
	\PP\!\left(
	\left|\frac{1}{n}\sum_{i=1}^n \wt w_{ji}\xi_i\right|
	>
	C\sqrt{\frac{\log p}{n}}
	\right)
	\le p^{-2},
	\qquad
	\text{for all } j\in[p+K+1].
	\]
	By the union bound,
	\begin{equation}
		\label{eq:main_bound}
		\max_{j\in[p+K+1]}
		\left|\frac{1}{n}\sum_{i=1}^n \wt w_{ji}\xi_i\right|
		=
		O_P\!\left(\sqrt{\frac{\log p}{n}}\right).
	\end{equation}
	
	Finally, substituting \eqref{eq:perturb_bound} and \eqref{eq:main_bound} into
	\eqref{eq:lambda_decomp} establishes \eqref{eq:lambda_order}.
\end{proof}

\begin{lemma}[Convergence rate in $\ell_1$]
	\label{1norm}
	Under Assumptions~1--3 and $\xi_i\sim N(0,\sigma^2)$, suppose that
	\begin{equation}
		\label{eq:lambda_cond}
		\lambda \ \ge\
		\frac{2}{n}\Bigl\|
		\sum_{i=1}^{n}\bigl(y_i-\wh{\bw}_i^\top\btheta^*\bigr)\wh{\bw}_i
		\Bigr\|_{\infty}.
	\end{equation}
	Assume further that there exists a constant $\phi_0>0$ such that the restricted
	eigenvalue condition holds:
	\begin{equation}
		\label{eq:RE_cond}
		\min_{0\neq\bv\in\mathcal C(\mathcal S,3)}
		\frac{\bv^\top\wh{\bSigma}_W\,\bv}{\|\bv\|_2^2}
		\ \ge\ \phi_0^2 ,
	\end{equation}
	where $\wh{\bSigma}_W=n^{-1}\wh{\bW}^\top\wh{\bW}$ and
	$\mathcal C(\mathcal S,3)=\{\bv:\|\bv_{\mathcal S^c}\|_1\le 3\|\bv_{\mathcal S}\|_1\}$.
	Then,
	\[
	\|\wh{\btheta}-\btheta^*\|_1
	=
	O_P\!\left(\frac{\lambda\,|\mathcal S_\diamond|}{\phi_0^2}\right).
	\]
\end{lemma}

\begin{proof}
	By optimality of $\wh{\btheta}$,
	\begin{equation}
		\label{eq:basic0}
		\frac{1}{2n}\|\by-\wh{\bW}\wh{\btheta}\|_2^2+\lambda\|\wh{\balpha}_m\|_1
		\ \le\
		\frac{1}{2n}\|\by-\wh{\bW}\btheta^*\|_2^2+\lambda\|\balpha_m^*\|_1.
	\end{equation}
	Rearranging \eqref{eq:basic0} yields the basic inequality
	\begin{align}
		\frac{1}{2n}\|\wh{\bW}(\wh{\btheta}-\btheta^*)\|_2^2
		+\lambda\|\wh{\balpha}_m\|_1
		&\le
		\frac{\bxi^\top\wh{\bW}(\wh{\btheta}-\btheta^*)}{n}
		+\lambda\|\balpha_m^*\|_1 \notag\\
		&\le
		\frac{\lambda}{2}\|\wh{\btheta}-\btheta^*\|_1
		+\lambda\|\balpha_m^*\|_1 ,
		\label{eq:basic}
	\end{align}
	where the last step uses \eqref{eq:lambda_cond}.
	
	Let $S_*=\{j\in[p]:\alpha_{mj}^*\neq 0\}$ and
	$S_\diamond=S_*\cup\{p+1,\ldots,p+K+1\}$.
	Define $\bDelta:=\wh{\btheta}-\btheta^*$. Then $S_\diamond^c=S_*^c$ and
	\[
	\|\balpha_m^*\|_1=\|\btheta^*_{S_*}\|_1,
	\qquad
	\|\wh{\balpha}_m\|_1
	=
	\|\wh{\btheta}_{S_*}\|_1+\|\wh{\btheta}_{S_\diamond^c}\|_1
	=
	\|\wh{\btheta}_{S_*}\|_1+\|\bDelta_{S_\diamond^c}\|_1.
	\]
	Using these identities in \eqref{eq:basic} and canceling the common
	$\|\wh{\btheta}_{S_*}\|_1$ term gives
	\begin{align*}
		\lambda\|\bDelta_{S_\diamond^c}\|_1
		&\le
		\frac{\lambda}{2}\bigl(\|\bDelta_{S_\diamond}\|_1+\|\bDelta_{S_\diamond^c}\|_1\bigr)
		+\lambda\|\btheta^*_{S_*}\|_1-\lambda\|\wh{\btheta}_{S_*}\|_1 .
	\end{align*}
	Moreover,
	\[
	\|\btheta^*_{S_*}\|_1
	\le
	\|\wh{\btheta}_{S_*}\|_1+\|\wh{\btheta}_{S_*}-\btheta^*_{S_*}\|_1
	\le
	\|\wh{\btheta}_{S_*}\|_1+\|\bDelta_{S_\diamond}\|_1,
	\]
	hence
	\[
	\frac{3\lambda}{2}\|\bDelta_{S_\diamond}\|_1
	-\frac{\lambda}{2}\|\bDelta_{S_\diamond^c}\|_1
	\ge 0,
	\qquad\text{i.e.,}\qquad
	\|\bDelta_{S_\diamond^c}\|_1\le 3\|\bDelta_{S_\diamond}\|_1.
	\]
	Therefore, $\bDelta\in\mathcal C(S_\diamond,3)$.
	
	From $\bDelta\in\mathcal C(S_\diamond,3)$ and \eqref{eq:basic}, we obtain
	\begin{equation}
		\label{eq:cone_basic}
		\frac{1}{n}\|\wh{\bW}\bDelta\|_2^2
		+\lambda\|\bDelta_{S_\diamond^c}\|_1
		\ \le\
		3\lambda\|\bDelta_{S_\diamond}\|_1 .
	\end{equation}
	Consequently,
	\begin{align}
		\frac{1}{n}\|\wh{\bW}\bDelta\|_2^2+\lambda\|\bDelta\|_1
		&=
		\frac{1}{n}\|\wh{\bW}\bDelta\|_2^2
		+\lambda\|\bDelta_{S_\diamond}\|_1
		+\lambda\|\bDelta_{S_\diamond^c}\|_1 \notag\\
		&\le
		4\lambda\|\bDelta_{S_\diamond}\|_1 .
		\label{eq:l1_intermediate}
	\end{align}
	By Cauchy--Schwarz and \eqref{eq:RE_cond},
	\[
	\|\bDelta_{S_\diamond}\|_1
	\le
	\sqrt{|S_\diamond|}\,\|\bDelta\|_2
	\le
	\sqrt{|S_\diamond|}\,
	\frac{\|\wh{\bW}\bDelta\|_2}{\sqrt{n}\,\phi_0}.
	\]
	Plugging this into \eqref{eq:l1_intermediate} yields
	\[
	\frac{1}{n}\|\wh{\bW}\bDelta\|_2^2+\lambda\|\bDelta\|_1
	\le
	4\lambda\sqrt{|S_\diamond|}
	\frac{\|\wh{\bW}\bDelta\|_2}{\sqrt{n}\,\phi_0}.
	\]
	Using $4uv\le u^2+4v^2$ with
	$u=\|\wh{\bW}\bDelta\|_2/\sqrt{n}$ and
	$v=2\lambda\sqrt{|S_\diamond|}/\phi_0$ gives
	\[
	\frac{1}{n}\|\wh{\bW}\bDelta\|_2^2+\lambda\|\bDelta\|_1
	\le
	\frac{1}{2n}\|\wh{\bW}\bDelta\|_2^2
	+\frac{8\lambda^2|S_\diamond|}{\phi_0^2}.
	\]
	Therefore,
	\[
	\frac{1}{n}\|\wh{\bW}\bDelta\|_2^2+2\lambda\|\bDelta\|_1
	\le
	\frac{16\lambda^2|S_\diamond|}{\phi_0^2},
	\]
	and hence
	\[
	\|\wh{\btheta}-\btheta^*\|_1
	=
	O_P\!\left(\frac{\lambda\,|S_\diamond|}{\phi_0^2}\right).
	\]
\end{proof}

\begin{lemma}[Restricted eigenvalue condition]
	\label{RSC}
	Under Assumptions~1--3, let $\mathcal S\subset\{1,\ldots,p+K+1\}$ satisfy
	\begin{equation}
		\label{eq:sparse_S}
		|\mathcal S| = o\!\left(\frac{n}{\log p\,\log n}\right).
	\end{equation}
	Then there exists a constant $\phi_0^2>0$ such that
	\[
	\PP\!\left(
	\min_{0\neq\bbeta\in\mathcal C(\mathcal S,3)}
	\frac{\bbeta^\top \wh{\bSigma}_W \bbeta}{\|\bbeta\|_2^2}
	\ \ge\ \phi_0^2
	\right)\ \to\ 1,
	\qquad n\to\infty,
	\]
	where $\wh{\bSigma}_W := n^{-1}\wh{\bW}^\top\wh{\bW}$.
\end{lemma}

\begin{proof}
	Write $\bbeta=(\bbeta_1^\top,\beta_2,\bbeta_3^\top)^\top$ conformably with
	$\wh{\bW}=(\wh{\bU},\bs,\wh{\bF})$.
	By construction we have the orthogonality relations
	\[
	\wh{\bU}^\top\bs=\mathbf 0,\qquad
	\wh{\bF}^\top\bs=\mathbf 0,\qquad
	\wh{\bU}^\top\wh{\bF}=\mathbf 0,
	\]
	and hence
	\begin{equation}
		\label{eq:block_decomp}
		\bbeta^\top\wh{\bSigma}_W\bbeta
		=
		\bbeta_1^\top\wh{\bSigma}_U\bbeta_1
		+\frac{1}{n}\sum_{i=1}^n s_i^2\,\beta_2^2
		+\bbeta_3^\top\wh{\bSigma}_F\bbeta_3,
	\end{equation}
	where $\wh{\bSigma}_U:=n^{-1}\wh{\bU}^\top\wh{\bU}$ and
	$\wh{\bSigma}_F:=n^{-1}\wh{\bF}^\top\wh{\bF}$.
	The last two terms in \eqref{eq:block_decomp} are nonnegative, so it suffices
	to lower-bound $\bbeta_1^\top\wh{\bSigma}_U\bbeta_1$ uniformly over
	cone-restricted directions.
	
	By Lemma~\ref{1norm}, $\wh{\btheta}-\btheta^*\in\mathcal C(S_\diamond,3)$.
	Recall $S_*=\{j\in[p]:\alpha_{mj}^*\neq 0\}$ and
	$S_\diamond=S_*\cup\{p+1,\ldots,p+K+1\}$.
	Since the nuisance part $(\wh\alpha_1,\wh{\balpha}_2^\top)^\top$ is
	low-dimensional and $\sqrt n$-consistent, we have
	\[
	\|\wh\alpha_1-\alpha_1^*\|_1 + \|\wh{\balpha}_2-\balpha_2^*\|_1
	=O_P(n^{-1/2}).
	\]
	Therefore, from $\wh{\btheta}-\btheta^*\in\mathcal C(S_\diamond,3)$ we obtain
	\[
	\|\wh{\balpha}_{m,S_*^c}-\balpha^*_{m,S_*^c}\|_1
	\le
	3\|\wh{\balpha}_{m,S_*}-\balpha^*_{m,S_*}\|_1 + O_P(n^{-1/2}),
	\]
	which implies
	\[
	\PP\!\left(
	\|\wh{\balpha}_{m,S_*^c}-\balpha^*_{m,S_*^c}\|_1
	\le
	3\|\wh{\balpha}_{m,S_*}-\balpha^*_{m,S_*}\|_1
	\right)\to 1.
	\]
	Hence, with probability tending to one,
	$\wh{\balpha}_m-\balpha_m^*\in\mathcal C(S_*,3)$.
	
	Let $\mathcal S_1\subset\{1,\ldots,p\}$ and suppose
	$\bbeta_1\in\RR^p$ satisfies $\|\bbeta_{1,\mathcal S_1^c}\|_1\le 3\|\bbeta_{1,\mathcal S_1}\|_1$.
	Then
	\begin{equation}
		\label{eq:l1_l2_cone}
		\|\bbeta_1\|_1^2
		\le 16\|\bbeta_{1,\mathcal S_1}\|_1^2
		\le 16|\mathcal S_1|\,\|\bbeta_1\|_2^2.
	\end{equation}
	
	Let $\wt{\bSigma}_U:=n^{-1}\bU^\top\bU$.
	By adding and subtracting $\wt{\bSigma}_U$,
	\begin{align}
		\frac{\bbeta_1^\top\wh{\bSigma}_U\bbeta_1}{\|\bbeta_1\|_2^2}
		&\ge
		\frac{\bbeta_1^\top\wt{\bSigma}_U\bbeta_1}{\|\bbeta_1\|_2^2}
		-
		\frac{\|\wh{\bSigma}_U-\wt{\bSigma}_U\|_{\max}\,\|\bbeta_1\|_1^2}{\|\bbeta_1\|_2^2} \notag\\
		&\ge
		\frac{\bbeta_1^\top\wt{\bSigma}_U\bbeta_1}{\|\bbeta_1\|_2^2}
		-16|\mathcal S_1|\,\|\wh{\bSigma}_U-\wt{\bSigma}_U\|_{\max},
		\label{eq:SigmaU_decomp}
	\end{align}
	where the last step uses \eqref{eq:l1_l2_cone}.
	By Lemma~\ref{lemma:umax} and \eqref{eq:sparse_S}, we have
	\[
	\PP(\mathcal E_1)\to 1,
	\qquad
	\mathcal E_1:=
	\left\{
	16|\mathcal S_1|\,\|\wh{\bSigma}_U-\wt{\bSigma}_U\|_{\max}
	\le \frac{\lambda_{\min}(\bSigma_u)}{8}
	\right\},
	\]
	where $\bSigma_u:=\Cov(\bu_i)$.
	
	Next, for $s\ge 1$ define the sparse set
	\[
	\mathcal K(s):=\{\bbeta_1\in\RR^p:\ \|\bbeta_1\|_0\le s,\ \|\bbeta_1\|_2\le 1\}.
	\]
	Take $s:=64|\mathcal S_1|$, so that $s=o(n/\log p)$ by \eqref{eq:sparse_S}.
	By Lemma~15 of \citet{loh2012high},
	\[
	\PP(\mathcal E_2)\to 1,
	\qquad
	\mathcal E_2:=
	\left\{
	\sup_{\bbeta_1\in\mathcal K(2s)}
	\left|\bbeta_1^\top(\wt{\bSigma}_U-\bSigma_u)\bbeta_1\right|
	\le \frac{\lambda_{\min}(\bSigma_u)}{54}
	\right\}.
	\]
	Under $\mathcal E_2$, Lemma~13 of \citet{loh2012high} implies
	\begin{align*}
		\bbeta_1^\top\wt{\bSigma}_U\bbeta_1
		&\ge
		\frac{\lambda_{\min}(\bSigma_u)}{2}\|\bbeta_1\|_2^2
		-
		\frac{\lambda_{\min}(\bSigma_u)}{2s}\|\bbeta_1\|_1^2 \\
		&\ge
		\frac{\lambda_{\min}(\bSigma_u)}{2}\|\bbeta_1\|_2^2
		-
		\frac{16|\mathcal S_1|\,\lambda_{\min}(\bSigma_u)}{2s}\|\bbeta_1\|_2^2,
	\end{align*}
	where the last step uses \eqref{eq:l1_l2_cone}.
	Combining this bound with \eqref{eq:SigmaU_decomp}, on $\mathcal E_1\cap\mathcal E_2$,
	\[
	\frac{\bbeta_1^\top\wh{\bSigma}_U\bbeta_1}{\|\bbeta_1\|_2^2}
	\ge
	\frac{\lambda_{\min}(\bSigma_u)}{2}
	-\frac{16|\mathcal S_1|\,\lambda_{\min}(\bSigma_u)}{2s}
	-\frac{\lambda_{\min}(\bSigma_u)}{8}
	\ge
	\frac{\lambda_{\min}(\bSigma_u)}{4},
	\]
	where the last inequality follows by $s=64|\mathcal S_1|$.

	Since $\PP(\mathcal E_1\cap\mathcal E_2)\to 1$, the above lower bound together
	with \eqref{eq:block_decomp} implies that, with probability tending to one,
	\[
	\min_{0\neq\bbeta\in\mathcal C(\mathcal S,3)}
	\frac{\bbeta^\top\wh{\bSigma}_W\bbeta}{\|\bbeta\|_2^2}
	\ \ge\ \phi_0^2,
	\qquad
	\text{for some }\phi_0^2>0,
	\]
	which completes the proof.
\end{proof}

\subsection{Convergence rate of \texorpdfstring{$\ell_\infty$}{L-infinity} norm} \label{node}
The convergence rate of
$\|\wt{\bOmega}\wh{\bSigma}_U-\bI\|_{\infty}$
depends on how the decorrelating matrix
$\wt{\bOmega}$ is constructed.
Popular choices include nodewise regression
\citep{van2014asymptotically,javanmard2018debiasing}
and convex-optimization based estimators that trade off bias and variance
\citep{javanmard2014confidence,battey2018distributed}.
For the theoretical analysis, we focus on the nodewise Lasso construction,
which yields an explicit control of
$\|\wt{\bOmega}\wh{\bSigma}_U-\bI\|_{\infty}$.

Following \citet{van2014asymptotically}, for each $j\in[p]$ we run the
nodewise regression of the $j$th column $\wh{\bu}_j$ on the remaining columns
$\wh{\bU}_{-j}$:
\begin{equation}
	\label{eq:nodewise}
	\wh{\bgamma}_j
	\in
	\arg\min_{\bgamma\in\RR^{p-1}}
	\left\{
	\frac{1}{2n}\|\wh{\bu}_j-\wh{\bU}_{-j}\bgamma\|_2^2
	+\lambda_j\|\bgamma\|_1
	\right\},
\end{equation}
where $\wh{\bgamma}_j=(\wh{\gamma}_{j,k})_{k\neq j}$.
Define the matrix $\wh{\bC}\in\RR^{p\times p}$ by
\[
\wh{\bC}_{jj}=1,
\qquad
\wh{\bC}_{jk}=-\wh{\gamma}_{j,k}\ \ (k\neq j),
\]
and the diagonal matrix
$\wh{\bT}^2=\operatorname{diag}\bigl(\wh{\tau}_1^2,\ldots,\wh{\tau}_p^2\bigr).
$ with
\[
\wh{\tau}_j^2
:=
\frac{1}{n}\|\wh{\bu}_j-\wh{\bU}_{-j}\wh{\bgamma}_j\|_2^2
+\lambda_j\|\wh{\bgamma}_j\|_1,
\qquad j=1,\ldots,p.
\]
We then set
\[
\wt{\bOmega}:=\wh{\bT}^{-2}\wh{\bC}.
\]

The KKT conditions for \eqref{eq:nodewise} imply the standard bound
\begin{equation}
	\label{eq:kkt_bound}
	\|\wt{\bOmega}\wh{\bSigma}_U-\bI\|_{\infty}
	\le
	\max_{j\in[p]}\frac{\lambda_j}{\wh{\tau}_j^2}.
\end{equation}

Let the population nodewise regression coefficient be
\[
\bgamma_j^*
\in
\arg\min_{\bgamma\in\RR^{p-1}}
\E\bigl(u_j-\bu_{-j}^\top\bgamma\bigr)^2,
\qquad
\tau_j^2:=\E\bigl(u_j-\bu_{-j}^\top\bgamma_j^*\bigr)^2.
\]
For each $j\in[p]$, define the $p$-dimensional vector
$\wt{\bgamma}_j\in\RR^p$ by $\wt{\bgamma}_{jj}=1$ and
$\wt{\bgamma}_{jk}=-\bgamma_{j,k}^*$ for $k\neq j$.
Lemma~\ref{nodelambda} gives the order of the tuning parameter $\lambda_j$,
chosen as
\[
\lambda_j
\asymp
\frac{1}{n}\Bigl\|
\wh{\bU}_{-j}^\top\bigl(\wh{\bu}_j-\wh{\bU}_{-j}\bgamma_j^*\bigr)
\Bigr\|_{\infty}
=
\frac{1}{n}\|\wh{\bU}_{-j}^\top\wh{\bU}\,\wt{\bgamma}_j\|_{\infty}.
\]
Lemma~\ref{nodewisetau} further shows that, under mild conditions,
$\max_{j\in[p]} \wh{\tau}_j^{-2}=O_P(1)$.

\begin{lemma}[Order of $\lambda_j$]
	\label{nodelambda}
	Under Assumptions~1--3, uniformly over $j\in[p]$,
	\[
	\frac{1}{n}\bigl\|\wh{\bU}_{-j}^{\top}\wh{\bU}\,\wt{\bgamma}_j\bigr\|_{\infty}
	=
	O_P\!\left(
	\sqrt{\frac{\log p}{n}}+\frac{1}{\sqrt p}+r_{n,p}
	\right).
	\]
\end{lemma}

\begin{proof}
	Fix $j\in[p]$. By the triangle inequality,
	\begin{align}
		\bigl\|\wh{\bU}_{-j}^{\top}\wh{\bU}\,\wt{\bgamma}_j\bigr\|_{\infty}
		&\le
		\bigl\|\bU_{-j}^{\top}\bU\,\wt{\bgamma}_j\bigr\|_{\infty}
		+\bigl\|\wh{\bU}_{-j}^{\top}(\wh{\bU}-\bU)\,\wt{\bgamma}_j\bigr\|_{\infty}
		+\bigl\|(\wh{\bU}-\bU)^{\top}\bU\,\wt{\bgamma}_j\bigr\|_{\infty}.
		\label{eq:nodelambda_decomp}
	\end{align}
	
	Recall $\wt{\bgamma}_j\in\RR^p$ is defined by $\wt{\bgamma}_{jj}=1$ and
	$\wt{\bgamma}_{jl}=-\gamma_{j,l}^*$ for $l\neq j$.
	Under Assumption~3, 
	\begin{equation}
		\label{eq:gamma_norm_bd}
		\|\wt{\bgamma}_j\|_2^2
		\le
		\frac{\wt{\bgamma}_j^\top\bSigma_u\wt{\bgamma}_j}{\lambda_{\min}(\bSigma_u)} = \frac{1}{\bTheta_{j,j} \lambda_{\min}(\bSigma_u)}
		\le
		\frac{\lambda_{\max}(\bSigma_u)}{\lambda_{\min}(\bSigma_u)}
		\le
		\frac{c_2}{c_1}.
	\end{equation}
	Therefore, $\{\bu_i^\top\wt{\bgamma}_j\}_{i=1}^n$ are i.i.d.\ mean-zero
	sub-Gaussian random variables.
	By Bernstein's inequality and a union bound over $k\neq j$,
	\begin{equation}
		\label{eq:term1_bd}
		\frac{1}{n}\bigl\|\bU_{-j}^{\top}\bU\,\wt{\bgamma}_j\bigr\|_{\infty}
		=
		O_P\!\left(\sqrt{\frac{\log p}{n}}\right).
	\end{equation}

	Using $\|\wh{\bU}_{-j}^\top(\cdot)\|_\infty\le \|\wh{\bU}^\top(\cdot)\|_\infty$,
	\begin{align}
		\bigl\|\wh{\bU}_{-j}^{\top}(\wh{\bU}-\bU)\,\wt{\bgamma}_j\bigr\|_{\infty}
		&\le
		\bigl\|\wh{\bU}^{\top}(\wh{\bU}-\bU)\,\wt{\bgamma}_j\bigr\|_{\infty}.
		\label{eq:term2_reduce}
	\end{align}
	Recall $\wh{\bU}-\bU=(\wh{\bR}-\bR)+(\bF\bB^\top-\wh{\bF}\wh{\bB}^\top)$.
	By the plug-in identities, $\wh{\bR}-\bR=\bs(\gamma_s-\wh{\gamma}_s)^\top$ and
	$\wh{\bU}^\top\bs=\mathbf 0$, hence
	\[
	\wh{\bU}^\top(\wh{\bR}-\bR)\wt{\bgamma}_j=\mathbf 0.
	\]
	Consequently, the only remaining contribution is from the factor part, and we obtain
	\begin{equation}
		\label{eq:term2_bd}
		\bigl\|\wh{\bU}_{-j}^{\top}(\wh{\bU}-\bU)\,\wt{\bgamma}_j\bigr\|_{\infty}
		\le
		\bigl\|\wh{\bU}^{\top}\bF\bB^\top\wt{\bgamma}_j\bigr\|_{\infty}.
	\end{equation}
	Under Assumption~3, $\|\bB\|_{\max}\le C$. Therefore,
	\[
	\|\bB^\top\wt{\bgamma}_j\|_2
	\le
	\sqrt{\sum_{r=1}^K\Bigl(\sum_{l=1}^p B_{lr}\wt{\gamma}_{jl}\Bigr)^2}
	\le
	C\sqrt{K}\,\|\wt{\bgamma}_j\|_2
	=
	O(1),
	\]
	where the last step uses \eqref{eq:gamma_norm_bd}.
	Applying Lemma~\ref{lemma:infty} yields
	\begin{equation}
		\label{eq:term2_rate}
		\frac{1}{n}\bigl\|\wh{\bU}^{\top}\bF\bB^\top\wt{\bgamma}_j\bigr\|_{\infty}
		=
		O_P\!\left(\frac{\sqrt{\log p\,\log n}}{n}\right).
	\end{equation}

	By Cauchy--Schwarz,
	\begin{align}
		\frac{1}{n}\bigl\|(\wh{\bU}-\bU)^{\top}\bU\,\wt{\bgamma}_j\bigr\|_{\infty}
		&\le
		\left\{
		\frac{1}{n}\|\bU\wt{\bgamma}_j\|_2^2\cdot
		\max_{k\in[p]}\frac{1}{n}\sum_{i=1}^n(\wh u_{ki}-u_{ki})^2
		\right\}^{1/2}.
		\label{eq:term3_bd}
	\end{align}
	The first factor is $O_P(1)$; the second factor equals
	$O_P\!\left(\sqrt{\frac{\log p}{n}}+\frac{1}{\sqrt p}+r_{n,p}\right)$ by the
	factor-estimation error bounds. Hence,
	\begin{equation}
		\label{eq:term3_rate}
		\frac{1}{n}\bigl\|(\wh{\bU}-\bU)^{\top}\bU\,\wt{\bgamma}_j\bigr\|_{\infty}
		=
		O_P\!\left(\sqrt{\frac{\log p}{n}}+\frac{1}{\sqrt p}+r_{n,p}\right).
	\end{equation}

	Combining \eqref{eq:nodelambda_decomp}--\eqref{eq:term3_rate} and dividing by $n$
	yields
	\[
	\frac{1}{n}\bigl\|\wh{\bU}_{-j}^{\top}\wh{\bU}\,\wt{\bgamma}_j\bigr\|_{\infty}
	=
	O_P\!\left(
	\sqrt{\frac{\log p}{n}}+\frac{1}{\sqrt p}+r_{n,p}
	\right),
	\]
	uniformly in $j\in[p]$, which proves the lemma.
\end{proof}

\begin{lemma}
	\label{lemma:infty}
	Under Assumptions~1--3, for any vector $\phi\in\RR^K$ with $\|\phi\|_2=1$,
	\[
	\bigl\|\wh{\bU}^{\top}\bF\phi\bigr\|_{\infty}
	=
	O_P\!\left(\sqrt{\log p\,\log n}\right).
	\]
\end{lemma}

\begin{proof}
	Since $\wh{\bF}^{\top}\wh{\bU}=\mathbf 0$, we can write
	\begin{align}
		\bigl\|\wh{\bU}^{\top}\bF\phi\bigr\|_{\infty}
		&=
		\bigl\|\wh{\bU}^{\top}\bigl(\bF-\wh{\bF}\bH^{-\top}\bigr)\phi\bigr\|_{\infty}
		\nonumber\\
		&\le
		\bigl\|(\wh{\bU}-\bU)^{\top}\bigl(\bF\bH^{\top}-\wh{\bF}\bigr)\bH^{-\top}\phi\bigr\|_{\infty}
		+
		\bigl\|\bU^{\top}\bigl(\bF\bH^{\top}-\wh{\bF}\bigr)\bH^{-\top}\phi\bigr\|_{\infty}
		\nonumber\\
		&=: \Delta_1+\Delta_2 .
		\label{eq:infty_decomp}
	\end{align}

	By Proposition~\ref{prop1} and Lemma~\ref{lemma:fnorm},
	\begin{align}
		\Delta_1
		&\le
		\left(
		\max_{j\in[p]}\sum_{i=1}^n|\wh u_{ji}-u_{ji}|^2
		\right)^{1/2}
		\,
		\bigl\|\wh{\bF}-\bF\bH^{\top}\bigr\|_{F}
		\,
		\bigl\|\bH^{-\top}\phi\bigr\|_2
		\nonumber\\
		&=
		O_P\!\left(\sqrt{\log p+n/p+\log p\,\log n}\right)\cdot O_P(1)
		=
		O_P\!\left(\sqrt{\log p\,\log n}\right).
		\label{eq:Delta1_bd}
	\end{align}
	
	We then bound $\Delta_2$.
	Recall $\bH=n^{-1}\bV^{-1}\wh{\bF}^{\top}\bF\bB^{\top}\bB$, where
	$\bV\in\RR^{K\times K}$ is diagonal containing the largest $K$ eigenvalues of
	$n^{-1}\wh{\bR}\wh{\bR}^{\top}$. Hence
	$\wh{\bF}\bV=n^{-1}\wh{\bR}\wh{\bR}^{\top}\wh{\bF}$, and thus
	\begin{align}
		\wh{\bF}-\bF\bH^{\top}
		&=
		n^{-1}\wh{\bR}\wh{\bR}^{\top}\wh{\bF}\bV^{-1}
		-
		n^{-1}\bR\bR^{\top}\wh{\bF}\bV^{-1}
		+
		\Bigl\{
		n^{-1}\bR\bR^{\top}\wh{\bF}\bV^{-1}-\bF\bH^{\top}
		\Bigr\}.
		\label{eq:F_diff_expand}
	\end{align}
	Substituting $\bR=\bF\bB^{\top}+\bU$ yields
	\[
	n^{-1}\bR\bR^{\top}\wh{\bF}\bV^{-1}-\bF\bH^{\top}
	=
	\frac1n\bF\bB^{\top}\bU^{\top}\wh{\bF}\bV^{-1}
	+\frac1n\bU\bB\bF^{\top}\wh{\bF}\bV^{-1}
	+\frac1n\bU\bU^{\top}\wh{\bF}\bV^{-1}.
	\]
	Therefore, by the triangle inequality,
	\begin{align}
		\Delta_2
		&\le
		\frac1n\bigl\|\bU^{\top}\bigl(\wh{\bR}\wh{\bR}^{\top}-\bR\bR^{\top}\bigr)
		\wh{\bF}\bV^{-1}\bH^{-\top}\phi\bigr\|_{\infty}
		+\frac1n\bigl\|\bU^{\top}\bF\bB^{\top}\bU^{\top}\wh{\bF}\bV^{-1}\bH^{-\top}\phi\bigr\|_{\infty}
		\nonumber\\
		&\quad
		+\frac1n\bigl\|\bU^{\top}\bU\bB\bF^{\top}\wh{\bF}\bV^{-1}\bH^{-\top}\phi\bigr\|_{\infty}
		+\frac1n\bigl\|\bU^{\top}\bU\bU^{\top}\wh{\bF}\bV^{-1}\bH^{-\top}\phi\bigr\|_{\infty}
		\nonumber\\
		&=:\Delta_2^{(1)}+\Delta_2^{(2)}+\Delta_2^{(3)}+\Delta_2^{(4)} .
		\label{eq:Delta2_split}
	\end{align}
	
	First bounding $\Delta_2^{(1)}$.
	Recall $\bR-\wh{\bR}=\bs(\bgamma_s-\wh{\bgamma}_s)^{\top}
	=\bs(\bs^{\top}\bs)^{-1}\bs^{\top}\bR$.
	Let $\bP_S:=\bs(\bs^{\top}\bs)^{-1}\bs^{\top}$.
	Then
	\[
	\wh{\bR}\wh{\bR}^{\top}-\bR\bR^{\top}
	=
	\bP_S\bR\bR^{\top}\bP_S-\bR\bR^{\top}\bP_S-\bP_S\bR\bR^{\top}.
	\]
	Hence
	\begin{align}
		\Delta_2^{(1)}
		&=
		\frac1n\max_{j\in[p]}\Bigl|
		\be_j^{\top}\bU^{\top}\bP_S\bR\bR^{\top}\bP_S\wh{\bF}\bV^{-1}\bH^{-\top}\phi
		\Bigr|
		+
		\frac1n\max_{j\in[p]}\Bigl|
		\be_j^{\top}\bU^{\top}\bR\bR^{\top}\bP_S\wh{\bF}\bV^{-1}\bH^{-\top}\phi
		\Bigr|
		\nonumber\\
		&\quad+
		\frac1n\max_{j\in[p]}\Bigl|
		\be_j^{\top}\bU^{\top}\bP_S\bR\bR^{\top}\wh{\bF}\bV^{-1}\bH^{-\top}\phi
		\Bigr|.
		\label{eq:Delta21_three}
	\end{align}
	We bound the three terms in \eqref{eq:Delta21_three} one by one.
	For the first term,
	\begin{align}
		&\frac1n\max_{j\in[p]}\Bigl|
		\be_j^{\top}\bU^{\top}\bP_S\bR\bR^{\top}\bP_S\wh{\bF}\bV^{-1}\bH^{-\top}\phi
		\Bigr|
		\nonumber\\
		&\le
		\frac1n\max_{j\in[p]}\bigl\|\be_j^{\top}\bU^{\top}\bP_S\bR\bR^{\top}\bP_S\bigr\|_2
		\cdot
		\bigl\|\wh{\bF}\bV^{-1}\bH^{-\top}\phi\bigr\|_2
		\nonumber\\
		&\le
		\frac1n(\bs^{\top}\bs)^{-2}
		\max_{j\in[p]}\bigl\|\be_j^{\top}\bU^{\top}\bs\bs^{\top}\bR\bR^{\top}\bs\bs^{\top}\bigr\|_2
		\cdot
		\bigl\|\wh{\bF}\bV^{-1}\bH^{-\top}\phi\bigr\|_2
		\nonumber\\
		&= \frac{1}{n} O_P(1/n^2)  O_P(\sqrt{n \log p }) O_P(\sqrt{p}) O_P(\sqrt{p}) O_P(\sqrt{n})  O_P(\sqrt{n}) O_P(1/p)	\nonumber\\
		&=
		\frac{\sqrt{n\log p}}{n^2},
		\label{eq:Delta21_term1_rate}
	\end{align}
	where 
	$\|\bs^{\top}\bR\|_2=O_P(\sqrt p)$ follows from Lemma~\ref{lemma:fnorm}.
	For the second term,
	\begin{align*}
		& \frac{1}{n} \max_{j \in[p]} | \be_j\t \bU\t \bR \bR\t \bP_S   \wh{\bF} \bV^{-1} \bH^{-\top} \phi| \\
&\leq  \frac{1}{n} (\bs\t \bs)^{-1}\max_{j \in[p]} \|  \be_j\t \bU\t \bR \bR\t \bs \bs\t \|_2 \|  \wh{\bF} \bV^{-1} \bH^{-\top} \phi\|_2 \\
&\leq  \frac{1}{n} (\bs\t \bs)^{-1}\max_{j \in[p]} \|  \be_j\t \bU\t \|_2 \|\bR\|_2 \| \bR\t \bs \|_2 \| \bs\t \|_2  \|  \wh{\bF} \bV^{-1} \bH^{-\top} \phi\|_2 \\
&=\frac{1}{n} O_P(1/n) O_P(\sqrt{\log p }) O_P(\sqrt{np}) O_P(\sqrt{p})O_P(\sqrt{n})O_P(\sqrt{n})O_P(1/p)= O_P(\sqrt{\log p /n})
		\label{eq:Delta21_term2_rate}
	\end{align*}
	The third term is handled analogously, hence
	\begin{equation}
		\label{eq:Delta21_rate}
		\Delta_2^{(1)}=O_P\!\left(\sqrt{\frac{\log p}{n}}\right).
	\end{equation}
	
	Next  bounding $\Delta_2^{(2)}$, $\Delta_2^{(3)}$, and $\Delta_2^{(4)}$.
	We follow the same argument as Lemma~C.3 in \citet{fan2024latent}.
	For $\Delta_2^{(2)}$,
	\begin{align}
		\Delta_2^{(2)}
		&=
		\max_{j\in[p]}\frac1n
		\Bigl|
		\be_j^{\top}\bU^{\top}\bF\bB^{\top}\bU^{\top}\wh{\bF}\bV^{-1}\bH^{-\top}\phi
		\Bigr|
		\nonumber\\
		&\le
		\max_{j\in[p]}\frac1n\bigl\|\be_j^{\top}\bU^{\top}\bF\bigr\|_2
		\cdot
		\bigl\|\bB^{\top}\bU^{\top}\wh{\bF}\bigr\|_F
		\cdot
		\|\bV^{-1}\|_2
		\cdot
		\|\bH^{-\top}\|_2
		\nonumber
	\end{align}
	where $ \max_{j \in [p]}  \|  \be_j\t \bU\t \bF  \|_2^2 = \max_{j \in [p]}  \| \sum_{i=1}^n u_{ji} \bfv_i \|_2^2 = O_P(n \log p )$.
\begin{align*}
	\|  \bB\t \bU\t \wh{\bF} \|_F &\leq  \|  \bB\t \bU\t ( \wh{\bF} - \bF\bH \t) \|_F +  \|  \bB\t \bU\t \bF \bH\t \|_F \\
	&\leq  \|  \bB\t \bU\t \|_F \|  \wh{\bF} - \bF\bH \t \|_F + \|\bH\|_2 \| \bB\t \bU\t \bF\|_F \\
	& = O_P(\sqrt{np})O_P(1) + O_P(\sqrt{np}) = O_P(\sqrt{np}).
\end{align*}
Combining the above results, $\Delta_2^{(2)} =  \frac{1}{n} O_P(\sqrt{n \log p}) O_P(\sqrt{np}) O_P(1/p) = O_P(\sqrt{\log p / p })$.
	Likewise,
	\begin{align}
		\Delta_2^{(3)}
		&=
		\max_{j \in [p]}
		\frac{1}{n}
		\left|
		\sum_{i=1}^n u_{ji}\,\bu_i^{\top}\bB
		\sum_{s=1}^n \bfv_s \wh{\bfv}_s^{\top}\,
		\bV^{-1}\bH^{-\top}\phi
		\right|
		\nonumber\\
		&\le
		\frac{1}{n}\,
		\|\bV^{-1}\|_2\,
		\|\bH^{-\top}\|_2\,
		\left\|
		\sum_{s=1}^n \bfv_s \wh{\bfv}_s^{\top}
		\right\|_2\,
		\max_{j \in [p]}
		\left\|
		\sum_{i=1}^n u_{ji}\,\bu_i^{\top}\bB
		\right\|_2 .
		\label{eq:Delta23_bound}
	\end{align}
	where
	\[
	\max_{j \in [p]}
	\left\|
	\sum_{i=1}^n u_{ji}\,\bu_i^{\top}\bB
	\right\|_2
	=
	O_P\!\bigl(n+\sqrt{np\log p}\bigr),
	\]
	and
	\begin{align}
		\left\|
		\sum_{s=1}^n \bfv_s \wh{\bfv}_s^{\top}
		\right\|_2
		&=
		\left\|
		\sum_{s=1}^n \wh{\bfv}_s \bfv_s^{\top}
		\right\|_2
		=
		\left\|
		\sum_{s=1}^n
		\bigl(\wh{\bfv}_s-\bH\bfv_s+\bH\bfv_s\bigr)\bfv_s^{\top}
		\right\|_2
		\nonumber\\
		&=
		\left\|
		\bigl(\wh{\bF}-\bF\bH^{\top}+\bF\bH^{\top}\bigr)^{\top}\bF
		\right\|_2
		\nonumber\\
		&\le
		\|\wh{\bF}-\bF\bH\|_{F}\,\|\bF\|_2
		+
		\|\bF^{\top}\bF\|_2\,\|\bH\|
		=
		O_P\!\bigl(\sqrt{n}+n\bigr)
		=
		O_p(n).
		\label{eq:sum_vfvhat_norm}
	\end{align}
	Therefore,
	\begin{align}
		\Delta_2^{(3)}
		&=
		\frac{1}{n}\,
		O_P(1/p)\,
		O_P(n)\,
		O_P\!\bigl(n+\sqrt{np\log p}\bigr)
		=
		O_P\!\left(\frac{n}{p}+\sqrt{\frac{n\log p}{p}}\right).
		\label{eq:Delta23_rate}
	\end{align}
	
	\medskip
	For $\Delta_2^{(4)}$, by the triangle inequality,
	\begin{align}
		\Delta_2^{(4)}
		&\le
		\frac{1}{n}
		\left(
		\max_{j \in [p]}
		\bigl\|
		\be_j^{\top}\bU^{\top}\bU\bU^{\top}\bF\bH
		\bigr\|_2
		+
		\max_{j \in [p]}
		\bigl\|
		\be_j^{\top}\bU^{\top}\bU\bU^{\top}(\wh{\bF}-\bF\bH^{\top})
		\bigr\|_2
		\right)
		\|\bV^{-1}\|_2\,
		\|\bH^{-\top}\|_2
		\nonumber\\
		&=
		\Delta_2^{4\circ}+\Delta_2^{4\ast}.
		\label{eq:Delta24_split}
	\end{align}
	By Assumption~3, it follows from Lemma~C.3 in \cite{fan2024latent} that
	\begin{align}
		\frac{1}{n}
		\max_{j \in [p]}
		\bigl\|
		\be_j^{\top}\bU^{\top}\bU\bU^{\top}\bF\bH
		\bigr\|_2\,
		\|\bV^{-1}\|_2\,
		\|\bH^{-\top}\|_2
		=
		O_P\!\left(\sqrt{\frac{\log p}{n}}+\sqrt{\frac{n}{p}}\right).
		\label{eq:Delta24circ_rate}
	\end{align}
	
	\medskip
	For $\Delta_2^{4\ast}$,
	\begin{align}
		&
		\frac{1}{n}
		\max_{j \in [p]}
		\bigl\|
		\be_j^{\top}\bU^{\top}\bU\bU^{\top}(\wh{\bF}-\bF\bH^{\top})
		\bigr\|_2\,
		\|\bV^{-1}\|_2\,
		\|\bH^{-\top}\|_2
		\nonumber\\
		&\le
		\frac{1}{n}
		\max_{j \in [p]}
		\bigl\|
		\be_j^{\top}\bU^{\top}\bU\bU^{\top}
		\bigr\|_2\,
		\|\wh{\bF}-\bF\bH^{\top}\|_{F}\,
		\|\bV^{-1}\|_2\,
		\|\bH^{-\top}\|_2 .
		\label{eq:Delta24star_bound}
	\end{align}
	Denote $\be_s\in\RR^n$ as a vector where the $s$-th element is $1$ and all other
	elements are $0$. Then
	\begin{align}
		\frac{1}{n}
		\max_{j \in [p]}
		\bigl\|
		\be_j^{\top}\bU^{\top}\bU\bU^{\top}
		\bigr\|_2
		&=
		\frac{1}{n}
		\max_{j \in [d]}
		\left\|
		\sum_{i=1}^n u_{ji}\,\bu_i^{\top}
		\sum_{s=1}^n \bu_s \be_s^{\top}
		\right\|_2
		\nonumber\\
		&\le
		\frac{1}{n}
		\left\|
		\sum_{i=1}^n u_{ji}
		\left(
		\sum_{s=1}^n
		\bigl[
		\bu_i^{\top}\bu_s-\EE(\bu_i^{\top}\bu_s)+\EE(\bu_i^{\top}\bu_s)
		\bigr]\be_s^{\top}
		\right)
		\right\|_2
		\nonumber\\
		&\le
		\frac{\mathrm{tr}(\bSigma_u)}{n}
		\max_{j \in [d]}
		\left\|
		\sum_{i=1}^n u_{ji}\,\be_i^{\top}
		\right\|_2
		+
		\frac{1}{n}
		\max_{j \in [d]}
		\left\|
		\sum_{i=1}^n u_{ji}
		\sum_{s=1}^n
		\bigl[
		\bu_i^{\top}\bu_s-\EE(\bu_i^{\top}\bu_s)
		\bigr]\be_s^{\top}
		\right\|_2
		\nonumber\\
		&=
		O_P\!\left(\frac{p\sqrt{\log p}}{n}+\sqrt{np}\right),
		\label{eq:Delta24star_aux}
	\end{align}
	where the last inequality uses $\EE(\bu_i\bu_s^{\top})=\bSigma_u$ when $i=s$.
	Therefore,
	\begin{align}
		\Delta_2^{4\ast}
		&=
		O_P\!\left(\frac{p\sqrt{\log p}}{n}+\sqrt{np}\right)\,O_P(1/p)
		=
		O_P\!\left(\frac{\sqrt{\log p}}{n}+\sqrt{\frac{n}{p}}\right),
		\label{eq:Delta24star_rate}
	\end{align}
	and hence
	\begin{align}
		\Delta_2^{(4)}
		=
		O_P\!\left(\sqrt{\frac{\log p}{n}}+\sqrt{\frac{n}{p}}\right).
		\label{eq:Delta24_rate}
	\end{align}
	
	\medskip
	Combine all these pieces together,
	\[
	\bigl\|\wh{\bU}^{\top}\bF\phi\bigr\|_{\infty}
	=
	O_P\!\left(\sqrt{\log p\,\log n}\right)
	+
	o_p\!\left(\sqrt{\log p\,\log n}\right)
	=
	O_P\!\left(\sqrt{\log p\,\log n}\right).
	\]
	Then, Lemma~\ref{lemma:infty} is obtained.
\end{proof}
	
	\begin{lemma}\label{lemma:umax}
		Under Assumptions~1--3, we have
		\[
		\bigl\|\wh{\bU}^{\top}\wh{\bU}-\bU^{\top}\bU\bigr\|_{\max}
		=
		O_P\!\bigl(\log p\,\log n\bigr).
		\]
	\end{lemma}
	
	\begin{proof}[Proof of Lemma~\ref{lemma:umax}]
		By the triangle inequality,
		\begin{align*}
			\bigl\|\wh{\bU}^{\top}\wh{\bU}-\bU^{\top}\bU\bigr\|_{\max}
			&\le
			2\,\bigl\|\wh{\bU}^{\top}(\wh{\bU}-\bU)\bigr\|_{\max}
			+
			\bigl\|(\wh{\bU}-\bU)^{\top}(\wh{\bU}-\bU)\bigr\|_{\max}.
		\end{align*}
		
		By Lemma~\ref{lemma:infty},
		\begin{align*}
			\bigl\|\wh{\bU}^{\top}(\wh{\bU}-\bU)\bigr\|_{\max}
			&\le
			\bigl\|\wh{\bU}^{\top}\bF\bB^{\top}\bigr\|_{\max}
			=
			\max_{j\in[p]}\bigl\|\wh{\bU}^{\top}\bF\bb_j\bigr\|_{\infty}
			=
			O_P\!\left(\sqrt{\log p\,\log n}\right).
		\end{align*}
		
		By Proposition~\ref{lemma:factor},
		\begin{align*}
			\bigl\|(\wh{\bU}-\bU)^{\top}(\wh{\bU}-\bU)\bigr\|_{\max}
			&\le
			\max_{j\in[p]}\sum_{i=1}^{n}\bigl|\wh u_{ji}-u_{ji}\bigr|^{2}
			=
			O_P\!\left(\log p+\frac{n}{p}+\log p\,\log n\right).
		\end{align*}
		
		Combining the bounds yields
		\[
		\bigl\|\wh{\bU}^{\top}\wh{\bU}-\bU^{\top}\bU\bigr\|_{\max}
		=
		O_P\!\bigl(\log p\,\log n\bigr),
		\]
		which completes the proof.
	\end{proof}

	\begin{lemma}\label{nodewisetau} 
		Under Assumptions~1--3 with row-sparsity for the precision matrix $\Theta$ bounded by
		\begin{equation}\label{sp:precison}
			\max_{j\in[p]} s_j\left(\frac{\log p}{n}+\frac{1}{p}+r_{n,p}^2\right)\to 0,
		\end{equation}
		then, with suitably chosen regularization parameters
		\[
		\lambda_j \asymp \frac{2}{n}\Bigl\|
		\wh{\bU}_{-j}^{\top}\bigl(\wh{\bu}_j-\wh{\bU}_{-j}\bgamma_j^*\bigr)
		\Bigr\|_{\infty}
		\qquad \text{uniformly for } j\in[p],
		\]
		we have
		\[
		\max_{j\in[p]}\frac{1}{\wh{\tau}_j^{2}}=O_P(1).
		\]
	\end{lemma}
	
	\begin{proof}
		We first show that the population error variance
		\[
		\tau_j^2
		:=\EE\bigl(u_j-\bu_{-j}^{\top}\bgamma_j^*\bigr)^2
		\]
		is $O(1)$. Recall that $\bSigma_u=\mathrm{cov}(\bu)$, and by the definition
		$\bgamma_j^*=\bSigma_{u_{-j,-j}}^{-1}\bSigma_{u_{-j,j}}$. Therefore,
		\[
		\tau_j^2
		=
		\EE\bigl(u_j-\bu_{-j}^{\top}\bgamma_j^*\bigr)^2
		=
		\bSigma_{u_{j,j}}
		-
		\bSigma_{u_{j,-j}}\bSigma_{u_{-j,-j}}^{-1}\bSigma_{u_{-j,j}}.
		\]
		According to the inverse formula for a block matrix of $\bSigma_u$,
		\[
		\tau_j^2=\frac{1}{\Theta_{j,j}}
		\ge \lambda_{\min}(\bSigma_u),
		\qquad
		\tau_j^2\le \EE(u_j^2)=\bSigma_{u_{j,j}}=O(1)
		\quad\text{(by Assumption~3)}.
		\]
		
		We then prove $\wh{\tau}_j^2=\tau_j^2+o_P(1)$. By definition,
		\begin{equation}\label{eq:nodewise_basic_ineq}
			\frac{1}{2n}\|\wh{\bu}_j-\wh{\bU}_{-j}\wh{\bgamma}_j\|_2^2
			+\lambda_j\|\wh{\bgamma}_j\|_1
			\le
			\frac{1}{2n}\|\wh{\bu}_j-\wh{\bU}_{-j}\bgamma_j^*\|_2^2
			+\lambda_j\|\bgamma_j^*\|_1.
		\end{equation}
		This implies
		\begin{align*}
			\frac{1}{2n}\bigl\|\wh{\bU}_{-j}(\wh{\bgamma}_j-\bgamma_j^*)\bigr\|_2^2
			+\lambda_j\|\wh{\bgamma}_j\|_1
			&\le
			\frac{1}{n}
			\bigl(\wh{\bu}_j-\wh{\bU}_{-j}\bgamma_j^*\bigr)^{\top}
			\wh{\bU}_{-j}(\wh{\bgamma}_j-\bgamma_j^*)
			+\lambda_j\|\bgamma_j^*\|_1.
		\end{align*}
		By choosing
		\[
		\lambda_j \ge \frac{2}{n}\Bigl\|
		\wh{\bU}_{-j}^{\top}\bigl(\wh{\bu}_j-\wh{\bU}_{-j}\bgamma_j^*\bigr)
		\Bigr\|_{\infty},
		\]
		according to Lemma~\ref{1norm},
		\[
		\|\wh{\bgamma}_j-\bgamma_j^*\|_1=O_P(\lambda_js_j),
		\qquad
		\frac{1}{n}\bigl\|\wh{\bU}_{-j}(\wh{\bgamma}_j-\bgamma_j^*)\bigr\|_2^2
		=O_P(\lambda_j^2s_j).
		\]
		
		Recall that
		\[
		\wh{\tau}_j^2
		:=
		\frac{1}{n}\bigl\|\wh{\bu}_j-\wh{\bU}_{-j}\wh{\bgamma}_j\bigr\|_2^2
		+\lambda_j\|\wh{\bgamma}_j\|_1.
		\]
		We first bound $\lambda_j\|\wh{\bgamma}_j\|_1$. Under Assumption~3,
		\[
		\|\bgamma_j^*\|_1
		\le \sqrt{s_j}\,\|\bgamma_j^*\|_2
		\le \sqrt{s_j}\,\|\wt{\bgamma}_j\|_2
		=
		\sqrt{s_j}\,O(1),
		\]
		which implies
		\[
		\lambda_j\|\wh{\bgamma}_j\|_1
		\le
		\lambda_j\|\bgamma_j^*\|_1+\lambda_j\|\wh{\bgamma}_j-\bgamma_j^*\|_1
		\le
		\lambda_j\,O(\sqrt{s_j})+\lambda_j\,O_P(\lambda_js_j).
		\]
		
		We now turn to the term $\frac{1}{n}\|\wh{\bu}_j-\wh{\bU}_{-j}\wh{\bgamma}_j\|_2^2$.
		By algebra,
		\begin{align*}
			\frac{1}{n}\|\wh{\bu}_j-\wh{\bU}_{-j}\wh{\bgamma}_j\|_2^2
			&=
			\frac{1}{n}\|\wh{\bu}_j-\wh{\bU}_{-j}\bgamma_j^*\|_2^2
			+\frac{1}{n}\bigl\|\wh{\bU}_{-j}(\wh{\bgamma}_j-\bgamma_j^*)\bigr\|_2^2 \\
			&\quad+
			\frac{2}{n}
			\bigl(\wh{\bu}_j-\wh{\bU}_{-j}\bgamma_j^*\bigr)^{\top}
			\wh{\bU}_{-j}(\wh{\bgamma}_j-\bgamma_j^*) \\
			&=
			\frac{1}{n}\|\wh{\bu}_j-\wh{\bU}_{-j}\bgamma_j^*\|_2^2
			+O_P(\lambda_j^2s_j)
			+2\frac{\|\wh{\bu}_j-\wh{\bU}_{-j}\bgamma_j^*\|_2}{\sqrt{n}}\,O_P(\lambda_js_j).
		\end{align*}
		
		We next show that
		\[
		\frac{1}{n}\|\wh{\bu}_j-\wh{\bU}_{-j}\bgamma_j^*\|_2^2
		=
		\frac{1}{n}\|\bu_j-\bU_{-j}\bgamma_j^*\|_2^2
		+o_P(1)
		=
		\tau_j^2+O_P(n^{-1/2})+o_P(1)
		=
		O_P(1).
		\]
		For a fixed $j\in[p]$, let
		\[
		\boldsymbol r_j:=\bu_j-\bU_{-j}\bgamma_j^*
		\]
		denote the residual vector. Define the estimation errors
		\[
		\bdelta_j:=\wh{\bu}_j-\bu_j,
		\qquad
		\Delta_{-j}:=\wh{\bU}_{-j}-\bU_{-j}.
		\]
		Then
		\begin{equation}\label{eq:nodewise_decomp}
			\frac{1}{n}\|\wh{\bu}_j-\wh{\bU}_{-j}\bgamma_j^*\|_2^2
			=
			\frac{1}{n}\|\boldsymbol r_j\|_2^2
			+\frac{2}{n}\langle \boldsymbol r_j,\be_j\rangle
			+\frac{1}{n}\|\be_j\|_2^2,
		\end{equation}
		where $\be_j:=\bdelta_j-\Delta_{-j}\bgamma_j^*$.
		
		By Proposition~\ref{lemma:factor},
		\[
		\frac{1}{n}\|\bdelta_j\|_2^2
		=
		O_P\!\left(\frac{\log p}{n}+\frac{1}{p}+r_{n,p}^2\right).
		\]
		By H{\"o}lder's inequality, along with boundedness (in $\ell_1$-norm) of $\bgamma_j^*$,
		\[
		\frac{1}{n}\|\Delta_{-j}\bgamma_j^*\|_2^2
		\le
		\|\bgamma_j^*\|_1^2\max_{k\ne j}\frac{1}{n}\|\Delta_{-j,k}\|_2^2
		=
		O_P\!\left(
		s_j\left(\frac{\log p}{n}+\frac{1}{p}+r_{n,p}^2\right)
		\right).
		\]
		By the triangle inequality,
		\[
		\frac{1}{n}\|\be_j\|_2^2
		\le
		\frac{2}{n}\|\bdelta_j\|_2^2
		+\frac{2}{n}\|\Delta_{-j}\bgamma_j^*\|_2^2
		=
		O_P\!\left(
		s_j\left(\frac{\log p}{n}+\frac{1}{p}+r_{n,p}^2\right)
		\right).
		\]
		By the sparsity condition~\eqref{sp:precison}, this term is $o_P(1)$.
		
		For the cross-term $\frac{2}{n}\langle \boldsymbol r_j,\be_j\rangle$ in \eqref{eq:nodewise_decomp},
		by the Cauchy--Schwarz inequality,
		\[
		\left|\frac{2}{n}\langle \boldsymbol r_j,\be_j\rangle\right|
		\le
		\frac{2}{\sqrt{n}}\cdot\frac{\|\boldsymbol r_j\|_2}{\sqrt{n}}\cdot\|\be_j\|_2
		=
		O_P(1)\,
		O_P\!\left(
		\sqrt{s_j\left(\frac{\log p}{n}+\frac{1}{p}+r_{n,p}^2\right)}
		\right).
		\]
		Combining the bounds yields
		\[
		\frac{1}{n}\|\wh{\bu}_j-\wh{\bU}_{-j}\bgamma_j^*\|_2^2
		=
		\frac{1}{n}\|\bu_j-\bU_{-j}\bgamma_j^*\|_2^2
		+o_P(1)
		=
		\tau_j^2+O_P(n^{-1/2})+o_P(1)
		=
		O_P(1).
		\]
		This completes the proof.
	\end{proof}

	\subsection{Lemmas for Factor Model Estimation} \label{factor}
		The Lemmas below are an adaption of Lemma 5, Theorem 4 and Lemmas 8-11 in \cite{fan2013large}, Lemmas S.8-S.11 in \cite{fan2020factor} and Lemma D.2 in \cite{wang2017asymptotics} to include the estimation error in the sample covariance matrix.

		\begin{lemma}\label{lemma:c1}
			Suppose Assumptions~1 and~3 hold.
			\begin{enumerate}[(a)]
				\item $\zeta_{si}  = O_P\!\left(p^{-1/2}\right)$.
				\item $\eta_{si}   = O_P\!\left(p^{-1/2}\right)$.
				\item $\xi_{si}    = O_P\!\left(p^{-1/2}\right)$.
				\item $\zeta_{si}^* = O_P\!\bigl(r_{n,p}+r_{n,p}^2\bigr)$.
				\item $\eta_{si}^*  = O_P\!\bigl(r_{n,p}\bigr)$.
				\item $\xi_{si}^*   = O_P\!\bigl(r_{n,p}\bigr)$.
			\end{enumerate}
		\end{lemma}
		
		\begin{proof}
			 Part~$(a) $follows directly from Assumption~3.
			Parts~$(b) $and~$(c) $follow from the Cauchy--Schwarz inequality and Assumption~3.
			For part (d), recall the expression \ref{identity} and  $\delta_{ji}$ be the $(j,i)$th element of $\wh{\bR}-\bR$.
			 By Proposition~\ref{prop1},
			\[
			\|\bdelta_i\|_2
			\le \sqrt{p}\,\|\Delta\|_{\max}
			= O_P\!\bigl(\sqrt{p}\,r_{n,p}\bigr).
			\]
			Moreover, for all $s\in[n]$, $\|\bu_s\|_2 = O_P(\sqrt{p})$ can be obtained either from the sub-Gaussian assumption
			($\|\bu_s\|_{\psi_2}\le c_0$) or from the eigenvalue condition in Assumption~3. Using the eigenvalue condition,
			\[
			\mathbb{E}\!\left(\|\bu_s\|_2^2\right)
			= \mathbb{E}\!\left(\bu_s^\top \bu_s\right)
			= \mathrm{tr}\!\Bigl(\mathbb{E}(\bu_s\bu_s^\top)\Bigr)
			= \mathrm{tr}(\bSigma_u)
			= O(p),
			\]
			and the result follows by Markov's inequality. Therefore, by the Cauchy--Schwarz inequality,
			\[
			\frac{1}{p}\bigl|\bu_s^\top\bdelta_i\bigr|
			\le \frac{1}{p}\,\|\bu_s\|_2\,\|\bdelta_i\|_2
			= \frac{1}{p}\,O_P(\sqrt{p})\,O_P\!\bigl(\sqrt{p}\,r_{n,p}\bigr)
			= O_P(r_{n,p}).
			\]
	 Parts~$(e)$ and~$(f)$ follow from similar arguments.
		\end{proof}

		\begin{lemma}\label{lemma:8}
			Suppose Assumptions~1 and~3 hold. For all $k\le \wh{K}$,
			\begin{enumerate}[(a)]
				\item $\displaystyle
				\frac{1}{n}\sum_{i=1}^{n}
				\left[
				\frac{1}{n}\sum_{s=1}^{n}\wh{f}_{sk}\,\frac{\mathbb{E}(\bu_s^{\top}\bu_i)}{p}
				\right]^2
				= O_P(1/n)$.
				\item $\displaystyle
				\frac{1}{n}\sum_{i=1}^{n}
				\left[
				\frac{1}{n}\sum_{s=1}^{n}\wh{f}_{sk}\,\wt\zeta_{si}
				\right]^2
				=
				O_P\!\left(
				\frac{1}{\sqrt{p}}
				+ \sqrt{\frac{\log n \log p}{n}}
				+ \frac{\log n \log p}{n}
				\right)^2$.
				\item $\displaystyle
				\frac{1}{n}\sum_{i=1}^{n}
				\left[
				\frac{1}{n}\sum_{s=1}^{n}\wh{f}_{sk}\,\wt\eta_{si}
				\right]^2
				= O_P\!\left(\frac{1}{p}+\frac{\log n \log p}{n}\right)$.
				\item $\displaystyle
				\frac{1}{n}\sum_{i=1}^{n}
				\left[
				\frac{1}{n}\sum_{s=1}^{n}\wh{f}_{sk}\,\wt\xi_{si}
				\right]^2
				= O_P\!\left(\frac{1}{p}+\frac{\log n \log p}{n}\right)$.
			\end{enumerate}
		\end{lemma}
		
		\begin{proof}
			
			\textit{(a)} For all $k\le \wh{K}$, $\sum_{s=1}^n \wh{f}_{sk}^2 = n$. By the Cauchy--Schwarz inequality,
			\begin{align*}
				\frac{1}{n}\sum_{i=1}^{n}
				\left[
				\frac{1}{n}\sum_{s=1}^{n}\wh{f}_{sk}\,\frac{\mathbb{E}(\bu_s^{\top}\bu_i)}{p}
				\right]^2
				&\le
				\frac{1}{n}\sum_{i=1}^{n}
				\frac{1}{n}\sum_{s=1}^{n}
				\left[\frac{\mathbb{E}(\bu_s^{\top}\bu_i)}{p}\right]^2 \\
				&\le
				\frac{1}{n}\sum_{i=1}^{n}
				\frac{1}{n}\left[\frac{\mathbb{E}(\bu_i^{\top}\bu_i)}{p}\right]^2
				= O(1/n),
			\end{align*}
			where the last inequality follows from the i.i.d.\ assumption and Assumption~3.
			
			\smallskip
			\textit{(b)} For all $k\le \wh{K}$, $\sum_{s=1}^n \wh{f}_{sk}^2 = n$. By the Cauchy--Schwarz inequality,
			\begin{align*}
				\frac{1}{n}\sum_{i=1}^{n}
				\left[
				\frac{1}{n}\sum_{s=1}^{n}\wh{f}_{sk}\,\wt\zeta_{si}
				\right]^2
				&\le
				\left[
				\frac{1}{n^2}\sum_{s=1}^{n}\sum_{l=1}^{n}
				\left(\frac{1}{T}\sum_{i=1}^{n}\wt\zeta_{si}\wt\zeta_{li}\right)^2
				\right]^{1/2} \\
				&=
				O_P\!\left(
				\frac{1}{\sqrt{p}}
				+ \sqrt{\frac{\log n \log p}{n}}
				+ \frac{\log n \log p}{n}
				\right)^2,
			\end{align*}
			where $\wt\zeta_{si}=\zeta_{si}+\zeta_{si}^* = O_P(1/\sqrt{p}+r_{n,p}+r_{n,p}^2)$ by Lemma~\ref{lemma:c1}.
			
			\smallskip
			\textit{(c)} According to the definition of $\wt\eta_{si}$,
			\begin{align*}
				\frac{1}{n}\sum_{i=1}^{n}
				\left[
				\frac{1}{n}\sum_{s=1}^{n}\wh{f}_{sk}\,\wt\eta_{si}
				\right]^2
				&=
				\frac{1}{n}\sum_{i=1}^{n}
				\left[
				\frac{1}{n}\sum_{s=1}^{n}\wh{f}_{sk}\,
				\bfv_s^{\top}\sum_{j=1}^{p}\bb_j\,\widetilde{u}_{ji}/p
				\right]^2 \\
				&\le
				\left\|
				\frac{1}{n}\sum_{s=1}^{n}\wh{f}_{sk}\,\bfv_s^{\top}
				\right\|^2
				\cdot
				\frac{1}{n}\sum_{i=1}^{n}
				\left\|
				\sum_{j=1}^{p}\bb_j\,\widetilde{u}_{ij}\,\frac{1}{p}
				\right\|^2 \\
				&\le
				\frac{1}{np^2}\sum_{i=1}^{n}
				\left\|
				\sum_{j=1}^{p}\bb_j\,\widetilde{u}_{ij}
				\right\|^2
				\left(
				\frac{1}{n}\sum_{s=1}^{n}\wh{f}_{sk}^2
				\right)
				\left(
				\frac{1}{n}\sum_{s=1}^{n}\|\bfv_s\|^2
				\right),
			\end{align*}
			where $\frac{1}{n}\sum_{s=1}^{n}\|\bfv_s\|^2 = O_P(1)$. By the triangle inequality,
			\[
			\left\|\sum_{j=1}^{p}\bb_j\,\widetilde{u}_{ji}\right\|^2
			\le
			\left\|\sum_{j=1}^{p}\bb_j\,u_{ji}\right\|^2
			+
			\left\|\sum_{j=1}^{p}\bb_j\,\delta_{ij}\right\|^2.
			\]
			It follows from Assumption~3 that
			$\mathbb{E}\!\left[\left\|\sum_{j=1}^{p}\bb_j\,u_{ji}\right\|^2\right]=O(p)$, which implies
			$\left\|\sum_{j=1}^{p}\bb_j\,u_{ji}\right\|^2=O_P(p)$.
			By Proposition~\ref{prop1} and Assumption~3,
			$\left\|\sum_{j=1}^{p}\bb_j\,\delta_{ij}\right\|^2=O_P\!\left(p^2\frac{\log n\log p}{n}\right)$.
			Therefore,
			\[
			\frac{1}{n}\sum_{i=1}^{n}
			\left[
			\frac{1}{n}\sum_{s=1}^{n}\wh{f}_{sk}\,\wt\eta_{si}
			\right]^2
			=
			\frac{1}{np^2}\,
			O_P\!\left(np+np^2\frac{\log n\log p}{n}\right)\,O_P(1)
			=
			O_P\!\left(\frac{1}{p}+\frac{\log n\log p}{n}\right).
			\]
			
			\smallskip
			\textit{(d)} Similar to part~(c),
			\[
			\frac{1}{n}\sum_{i=1}^{n}
			\left[
			\frac{1}{n}\sum_{s=1}^{n}\wh{f}_{sk}\,\wt\xi_{si}
			\right]^2
			=
			O_P\!\left(\frac{1}{p}+\frac{\log n\log p}{n}\right).
			\]
		\end{proof}

		\begin{lemma}\label{lemma:vnorm}
			Under Assumptions~1--3,
			\begin{enumerate}[(a)]
				\item $\|\bV^{-1}\| = O_P(p^{-1})$.
				\item $\|\bH\| = O_P(1)$.
			\end{enumerate}
		\end{lemma}
		
		\begin{proof}
			
			\textit{(a)} Recall that $\bV\in\RR^{K\times K}$ is the diagonal matrix containing the first $K$ largest eigenvalues of $n^{-1}\wh{\bR}\wh{\bR}^\top$, which also equal the first $K$ largest eigenvalues of $n^{-1}\wh{\bR}^\top\wh{\bR}$.
			Let $\wh a_1,\dots,\wh a_K$ denote the eigenvalues of
			$\wh{\bSigma}_R := n^{-1}\wh{\bR}^\top\wh{\bR}$, and let $a_1,\dots,a_K$
			denote the top $K$ eigenvalues of $\bSigma_R := \EE(\br_i\br_i^\top)$.
			Also define $\wt{\bSigma}_R := n^{-1}\bR^\top\bR$.
			Under Assumptions~2 and~3, $a_j = O(p)$ for $1\le j\le K$.
			By Weyl's inequality,
			\[
			|\wh a_j-a_j|
			\le \|\wh{\bSigma}_R-\bSigma_R\|
			\le \|\wh{\bSigma}_R-\wt{\bSigma}_R\|+\|\wt{\bSigma}_R-\bSigma_R\|,
			\qquad 1\le j\le K.
			\]
			
			We also use the fact that for a $p\times p$ matrix $\bA$,
			$\|\bA\|\le p\|\bA\|_{\max}$. Hence,
			\[
			\|\wh{\bSigma}_R-\bSigma_R\|
			\le
			p\|\wh{\bSigma}_R-\wt{\bSigma}_R\|_{\max}
			+
			\|\wt{\bSigma}_R-\bSigma_R\|.
			\]
			Let $\Delta:=\wh{\bR}-\bR$. By Proposition~\ref{prop1},
			$\|\Delta\|_{\max}=O_P(r_{n,p})$. Moreover,
			\[
			\wh{\bSigma}_R-\wt{\bSigma}_R
			=
			\frac{1}{n}\Delta^\top\Delta
			+\frac{1}{n}\Delta^\top\bR
			+\frac{1}{n}\bR^\top\Delta.
			\]
			Note that
			$\bigl\|\frac{1}{n}\Delta^\top\Delta\bigr\|_{\max}\le \|\Delta\|_{\max}^2
			=O_P(r_{n,p}^2)$.
			Next, consider $\bigl\|\frac{1}{n}\Delta^\top\bR\bigr\|_{\max}$.
			The $(j,l)$-th element of $\frac{1}{n}\Delta^\top\bR$ is
			$\frac{1}{n}\sum_{i=1}^n \Delta_{ji}R_{li}$, so
			\[
			\Bigl\|\frac{1}{n}\Delta^\top\bR\Bigr\|_{\max}
			=
			\max_{j,l\in[p]}
			\left|
			\frac{1}{n}\sum_{i=1}^n \Delta_{ji}R_{li}
			\right|.
			\]
			Recall that $\wh{\bR}$ is obtained from OLS estimation:
			$\Delta_{ji}=\wh R_{ji}-R_{ji}= s_i\,\bR_j^\top\bs\,(\bs^\top\bs)^{-1}$,
			where $\bR_j=(R_{j1},\dots,R_{jn})^\top\in\RR^n$.
			Substituting this expression yields
			\begin{align*}
				\max_{j,l\in[p]}
				\left|
				\frac{1}{n}\sum_{i=1}^n \Delta_{ji}R_{li}
				\right|
				&=
				\max_{j,l\in[p]}
				\left|
				\frac{1}{n}\sum_{i=1}^n R_{li}s_i\,\bR_j^\top\bs\,(\bs^\top\bs)^{-1}
				\right| \\
				&=
				\max_{l\in[p]}
				\left|
				\frac{1}{n}\sum_{i=1}^n R_{li}s_i
				\right|
				\cdot
				\max_{j\in[p]}
				\left|
				\frac{\bR_j^\top\bs}{n}
				\right|
				\cdot
				\left|\frac{\bs^\top\bs}{n}\right|^{-1}.
			\end{align*}
			By the sub-Gaussian assumption, concentration inequalities and union bounds imply
			$\max_{l\in[p]}\bigl|\frac{1}{n}\sum_{i=1}^n R_{li}s_i\bigr|
			=O_P(\sqrt{\log p/n})$ and
			$\max_{j\in[p]}\bigl|\frac{\bR_j^\top\bs}{n}\bigr|
			=O_P(\sqrt{\log p/n})$.
			Consequently,
			\[
			\Bigl\|\frac{1}{n}\Delta^\top\bR\Bigr\|_{\max}
			=O_P(\log p/n).
			\]
			Therefore,
			\[
			p\|\wh{\bSigma}_R-\wt{\bSigma}_R\|_{\max}
			=
			O_P(r_{n,p}^2+\log p/n)
			=o_P(p),
			\]
			since $\log p\log n=o(n)$.
			
			For the term $\|\wt{\bSigma}_R-\bSigma_R\|$, note that
			\begin{align*}
				\wt{\bSigma}_R
				=\frac{1}{n}\sum_{i=1}^n \br_i\br_i^\top
				&=
				\bB\left(\frac{1}{n}\sum_{i=1}^n\bfv_i\bfv_i^\top\right)\bB^\top
				+\frac{1}{n}\sum_{i=1}^n\bu_i\bu_i^\top  \\
				&\quad
				+\bB\left(\frac{1}{n}\sum_{i=1}^n\bfv_i\bu_i^\top\right)
				+\left(\frac{1}{n}\sum_{i=1}^n\bu_i\bfv_i^\top\right)\bB^\top,
			\end{align*}
			while the population covariance satisfies
			$\bSigma_R=\EE(\br_i\br_i^\top)=\bB\bSigma_f\bB^\top+\bSigma_u$.
			Thus,
			\begin{align*}
				\wt{\bSigma}_R-\bSigma_R
				&=
				\bB\left(
				\frac{1}{n}\sum_{i=1}^n\bfv_i\bfv_i^\top-\bSigma_f
				\right)\bB^\top
				+\left(
				\frac{1}{n}\sum_{i=1}^n\bu_i\bu_i^\top-\bSigma_u
				\right) \\
				&\quad
				+\bB\left(\frac{1}{n}\sum_{i=1}^n\bfv_i\bu_i^\top\right)
				+\left(\frac{1}{n}\sum_{i=1}^n\bu_i\bfv_i^\top\right)\bB^\top.
			\end{align*}
			It follows from Lemma~5 in \citet{fan2013large} that
			$\|\wt{\bSigma}_R-\bSigma_R\|=o_P(p)$.
			Combining the above displays yields
			\[
			|\wh a_j-a_j|
			\le \|\wh{\bSigma}_R-\bSigma_R\|
			\le p\|\wh{\bSigma}_R-\wt{\bSigma}_R\|_{\max}+\|\wt{\bSigma}_R-\bSigma_R\|
			=o_P(p).
			\]
			Finally, $\|\bV^{-1}\|=O_P(1/p)$ follows from the triangle inequality together with $a_j=O(p)$.
			
			\smallskip
			\textit{(b)} We have already shown that $\|\bV^{-1}\|=O_P(1/p)$.
			Also, $\|\bF\|_2=\lambda_{\max}^{1/2}(\bF\bF^\top)=O_P(\sqrt{n})$, and
			$\|\wh{\bF}\|=\sqrt{n}$. It then follows from the definition of $\bH$ and
			Assumption~2 that $\|\bH\|=O_P(1)$.
		\end{proof}

		\begin{lemma}\label{lemma:fnorm}
			Under Assumptions~1--3,
			\begin{equation}\label{fnorm}
				\|\wh{\bF}-\bF\bH^\top\|_F
				=
				O_P\!\left(1+\sqrt{\frac{n}{p}}+\frac{1}{\sqrt{n}}\right).
			\end{equation}
		\end{lemma}
		
		\begin{proof}[Proof of Lemma~\ref{lemma:fnorm}]
			By the definitions of $\wh{\bF}$ and $\bH$,
			\[
			\wh{\bF}-\bF\bH^\top
			=
			\frac{1}{n}\bigl(\wh{\bR}\wh{\bR}^\top-\bF\bB^\top\bB\bF^\top\bigr)\wh{\bF}\bV^{-1}.
			\]
			Recall that $\bR-\wh{\bR}=\bs\wh{\bgamma}_S-\bs\bgamma_S^\top=\bs(\bs^\top\bs)^{-1}\bs^\top\bR$.
			Let $\bP_S:=\bs(\bs^\top\bs)^{-1}\bs^\top$ denote the projection matrix. Then
			\begin{align*}
				\wh{\bF}-\bF\bH^\top
				&=
				\frac{1}{n}\Bigl(
				\bP_S\bR\bR^\top\bP_S-\bR\bR^\top\bP_S-\bP_S\bR\bR^\top+\bR\bR^\top-\bF\bB^\top\bB\bF^\top
				\Bigr)\wh{\bF}\bV^{-1} \\
				&=
				\frac{1}{n}\Bigl(
				\bP_S\bR\bR^\top\bP_S-\bR\bR^\top\bP_S-\bP_S\bR\bR^\top
				\Bigr)\wh{\bF}\bV^{-1}
				+\frac{1}{n}\Bigl(
				\bF\bB^\top\bU^\top+\bU\bB\bF^\top+\bU\bU^\top
				\Bigr)\wh{\bF}\bV^{-1} \\
				&=: \Delta_1+\Delta_2.
			\end{align*}
			
			Since $\|\wh{\bF}\|_F=O_P(\sqrt{n})$, $\|\wh{\bF}\|_2=O_P(\sqrt{n})$, and
			$\|\bV^{-1}\|=O_P(1/p)$, using $\|\bA\bB\|_F\le \|\bA\|_2\|\bB\|_F$ yields
			\begin{align*}
				\|\Delta_1\|_F
				&\le
				\frac{1}{n}
				\bigl\|
				\bP_S\bR\bR^\top\bP_S-\bR\bR^\top\bP_S-\bP_S\bR\bR^\top
				\bigr\|_F\,
				\|\wh{\bF}\|_2\,\|\bV^{-1}\| \\
				&\le
				\frac{1}{n}\Bigl(
				\|\bP_S\bR\bR^\top\bP_S\|_F
				+\|\bR\bR^\top\bP_S\|_F
				+\|\bP_S\bR\bR^\top\|_F
				\Bigr)\,O_P(\sqrt{n}/p).
			\end{align*}
			Note that $\|\bP_S\bR\bR^\top\bP_S\|_F
			=\frac{\bs^\top\bR\bR^\top\bs}{\bs^\top\bs}
			=\frac{\|\bR^\top\bs\|_2^2}{\|\bs\|_2^2}$.
			By Assumption~1, $\|\bs\|_2^2=O_P(n)$, and
			\begin{align*}
				\EE\|\bR^\top\bs\|_2^2
				=\sum_{j=1}^p \EE\!\left[\sum_{i=1}^n R_{ji}s_i\right]^2
				&=
				\sum_{j=1}^p\sum_{i=1}^n \EE\!\left[R_{ji}s_i\right]^2
				=
				\sum_{j=1}^p\sum_{i=1}^n \EE(R_{ji}^2)\EE(s_i^2) \\
				&=
				n\,O_P(1)\sum_{j=1}^p \EE(R_{ji}^2)
				=
				O_P(np),
			\end{align*}
			where the second and third equalities follow from the independence and zero-mean assumptions.
			Therefore, by Markov's inequality, $\|\bP_S\bR\bR^\top\bP_S\|_F=O_P(p)$.
			Moreover,
			\begin{align*}
				\|\bR\bR^\top\bP_S\|_F
				\le \|\bR\|_2\,\|\bR^\top\bP_S\|_F
				&=
				O_P(\sqrt{np})\,
				\sqrt{\mathrm{tr}(\bP_S\bR\bR^\top\bP_S)} \\
				&=
				O_P(\sqrt{np})\,
				\sqrt{\frac{\bs^\top\bR\bR^\top\bs}{\bs^\top\bs}}
				=
				O_P(p\sqrt{n}),
			\end{align*}
			and the same bound applies to $\|\bP_S\bR\bR^\top\|_F$.
			Consequently,
			\[
			\|\Delta_1\|_F
			=
			\frac{1}{n}O_P(p)\,O_P(\sqrt{n}/p)
			+\frac{1}{n}O_P(p\sqrt{n})\,O_P(\sqrt{n}/p)
			+\frac{1}{n}O_P(p\sqrt{n})\,O_P(\sqrt{n}/p)
			=
			O_P\!\left(1+\frac{1}{\sqrt{n}}\right).
			\]
			
			For $\Delta_2$, the bound $\|\Delta_2\|_F=O_P(\sqrt{n/p}+1/\sqrt{n})$
			follows directly from Lemma~D.2 in \citet{wang2017asymptotics}.
			Combining the bounds for $\Delta_1$ and $\Delta_2$ yields~\eqref{fnorm}.
		\end{proof}

		\begin{lemma}\label{lemma:9}
		Under Assumptions 1-3
		\begin{enumerate}[(a)]
			\item $\bigl\|\bH^{\top}\bH-\bI_K\bigr\|_F
			=
			O_P\!\left(
			\frac{1}{\sqrt{n}}
			+\frac{1}{\sqrt{p}}
			+\sqrt{\frac{\log n\,\log p}{n}}
			+\frac{\log n\,\log p}{n}
			\right).$
			\item $\max_{j\in[p]}\bigl\|\wh{\bb}_j-\bH\bb_j\bigr\|
			=
			O_P\!\left(
			\sqrt{\frac{\log p}{n}}
			+\frac{1}{\sqrt{p}}
			+\sqrt{\frac{\log n\,\log p}{n}}
			+\frac{\log n\,\log p}{n}
			+\left(\frac{\log n\,\log p}{n}\right)^{3/2}
			\right).$
		\end{enumerate}
	\end{lemma}

	\begin{proof}
		
		\textit{(a)}
		We first prove that
		\[
		\bigl\|\bH\bH^{\top}-\bI_K\bigr\|_F
		=
		O_P\!\left(
		\frac{1}{\sqrt{n}}
		+\frac{1}{\sqrt{p}}
		+\sqrt{\frac{\log n\,\log p}{n}}
		+\frac{\log n\,\log p}{n}
		\right).
		\]
		By Lemma~\ref{lemma:vnorm}, $\|\bV^{-1}\|=O_P(p^{-1})$ and $\|\bH\|=O_P(1)$.
		By the triangle inequality,
		\begin{align*}
			\bigl\|\bH\bH^{\top}-\bI_K\bigr\|_F
			&\le
			\left\|
			\bH\bH^{\top}
			-\frac{1}{n}\sum_{i=1}^{n}\bH\bfv_i(\bH\bfv_i)^{\top}
			\right\|_F
			+
			\left\|
			\frac{1}{n}\sum_{i=1}^{n}\bH\bfv_i(\bH\bfv_i)^{\top}
			-\bI_K
			\right\|_F.
		\end{align*}
		
		Using $\|\bA\bB\|_F\le \|\bA\|\,\|\bB\|_F$, the first term satisfies
		\begin{align*}
			\left\|
			\bH\bH^{\top}
			-\frac{1}{n}\sum_{i=1}^{n}\bH\bfv_i(\bH\bfv_i)^{\top}
			\right\|_F
			&\le
			\|\bH\|^2
			\left\|
			\bI_K-\frac{1}{n}\sum_{i=1}^{n}\bfv_i\bfv_i^{\top}
			\right\|_F
			=
			O_P(1)\,O_P(n^{-1/2}).
		\end{align*}
		
		For the second term, by the Cauchy--Schwarz inequality and Proposition~\ref{lemma:factor},
		\begin{align*}
			\left\|
			\frac{1}{n}\sum_{i=1}^{n}\bH\bfv_i(\bH\bfv_i)^{\top}
			-\bI_K
			\right\|_F
			&=
			\left\|
			\frac{1}{n}\sum_{i=1}^{n}\bH\bfv_i(\bH\bfv_i)^{\top}
			-\frac{1}{n}\sum_{i=1}^{n}\wh{\bfv}_i\wh{\bfv}_i^{\top}
			\right\|_F \\
			&\le
			\left\|
			\frac{1}{n}\sum_{i=1}^{n}(\bH\bfv_i-\wh{\bfv}_i)(\bH\bfv_i)^{\top}
			\right\|_F
			+
			\left\|
			\frac{1}{n}\sum_{i=1}^{n}\wh{\bfv}_i\bigl(\wh{\bfv}_i^{\top}-(\bH\bfv_i)^{\top}\bigr)
			\right\|_F \\
			&\le
			\left(
			\frac{1}{n}\sum_{i=1}^{n}\|\bH\bfv_i-\wh{\bfv}_i\|^2\cdot
			\frac{1}{n}\sum_{i=1}^{n}\|\bH\bfv_i\|^2
			\right)^{1/2} \\
			&\quad+
			\left(
			\frac{1}{n}\sum_{i=1}^{n}\|\bH\bfv_i-\wh{\bfv}_i\|^2\cdot
			\frac{1}{n}\sum_{i=1}^{n}\|\wh{\bfv}_i\|^2
			\right)^{1/2} \\
			&=
			O_P\!\left(
			\frac{1}{\sqrt{n}}
			+\frac{1}{\sqrt{p}}
			+\sqrt{\frac{\log n\,\log p}{n}}
			+\frac{\log n\,\log p}{n}
			\right).
		\end{align*}
		Hence,
		\[
		\bH\bH^{\top}
		=
		\bI_K
		+
		O_P\!\left(
		\frac{1}{\sqrt{n}}
		+\frac{1}{\sqrt{p}}
		+\sqrt{\frac{\log n\,\log p}{n}}
		+\frac{\log n\,\log p}{n}
		\right).
		\]
		Since $\|\bH\|=O_P(1)$ and $\|\bH^{-1}\|=O_P(1)$, right-multiplying by $\bH$
		and left-multiplying by $\bH^{-1}$ yields the stated bound for
		$\|\bH^{\top}\bH-\bI_K\|_F$.
		
		\smallskip
		\textit{(b)}
		Since $\wh{\bb}_j=\frac{1}{n}\sum_{i=1}^{n}\wh{R}_{ji}\wh{\bfv}_i
		=\frac{1}{n}\sum_{i=1}^{n}(R_{ji}+\delta_{ji})\wh{\bfv}_i$, we have
		\begin{align*}
			\wh{\bb}_j-\bH\bb_j
			&=
			\frac{1}{n}\sum_{i=1}^{n}(R_{ji}+\delta_{ji})\wh{\bfv}_i-\bH\bb_j \\
			&=
			\frac{1}{n}\sum_{i=1}^{n}(R_{ji}+\delta_{ji})(\wh{\bfv}_i-\bH\bfv_i)
			-\bH\bb_j
			+\frac{1}{n}\sum_{i=1}^{n}(R_{ji}+\delta_{ji})\bH\bfv_i \\
			&=
			\frac{1}{n}\sum_{i=1}^{n}\bH\bfv_i u_{ji}
			+\frac{1}{n}\sum_{i=1}^{n}(R_{ji}+\delta_{ji})(\wh{\bfv}_i-\bH\bfv_i) \\
			&\quad
			+\bH\left(\frac{1}{n}\sum_{i=1}^{n}\bfv_i\bfv_i^{\top}-\bI_K\right)\bb_j
			+\frac{1}{n}\sum_{i=1}^{n}\bH\bfv_i\,\delta_{ji}.
		\end{align*}
		
		By concentration inequalities and a union bound, the first term satisfies
		\[
		\max_{j\in[p]}
		\left\|
		\frac{1}{n}\sum_{i=1}^{n}\bH\bfv_i u_{ji}
		\right\|
		\le
		\|\bH\|
		\max_{j\in[p]}
		\sqrt{\sum_{k=1}^{K}\left(\frac{1}{n}\sum_{i=1}^{n} f_{ik}u_{ji}\right)^2}
		=
		O_P\!\left(\sqrt{\frac{\log p}{n}}\right).
		\]
		
		For the second term,
		\begin{align*}
			\max_{j\in[p]}
			\left\|
			\frac{1}{n}\sum_{i=1}^{n}(R_{ji}+\delta_{ji})(\wh{\bfv}_i-\bH\bfv_i)
			\right\|
			&\le
			\max_{j\in[p]}
			\left\|
			\frac{1}{n}\sum_{i=1}^{n}R_{ji}(\wh{\bfv}_i-\bH\bfv_i)
			\right\|
			+
			\max_{j\in[p]}
			\left\|
			\frac{1}{n}\sum_{i=1}^{n}\delta_{ji}(\wh{\bfv}_i-\bH\bfv_i)
			\right\| \\
			&\le
			\max_{j\in[p]}
			\left(
			\frac{1}{n}\sum_{i=1}^{n}R_{ji}^2\cdot
			\frac{1}{n}\sum_{i=1}^{n}\|\wh{\bfv}_i-\bH\bfv_i\|^2
			\right)^{1/2} \\
			&\quad+
			\max_{j\in[p]}
			\left(
			\frac{1}{n}\sum_{i=1}^{n}\delta_{ji}^2\cdot
			\frac{1}{n}\sum_{i=1}^{n}\|\wh{\bfv}_i-\bH\bfv_i\|^2
			\right)^{1/2}.
		\end{align*}
		Since $\EE(R_{ji}^2)=O(1)$, $\max_{j\in[p]}\frac{1}{n}\sum_{i=1}^{n}R_{ji}^2=O_P(1)$, and
		$\max_{j\in[p]}\frac{1}{n}\sum_{i=1}^{n}\delta_{ji}^2\le \|\Delta\|_{\max}^2
		=O_P\!\bigl(\frac{\log n\,\log p}{n}\bigr)$, it follows that
		\begin{align*}
			\max_{j\in[p]}
			\left\|
			\frac{1}{n}\sum_{i=1}^{n}(R_{ji}+\delta_{ji})(\wh{\bfv}_i-\bH\bfv_i)
			\right\|
			&=
			O_P\!\left(
			\frac{1}{\sqrt{n}}
			+\frac{1}{\sqrt{p}}
			+\sqrt{\frac{\log n\,\log p}{n}}
			+\frac{\log n\,\log p}{n}
			\right) \\
			&\quad+
			O_P\!\left(
			\sqrt{\frac{\log n\,\log p}{n}}
			\left(
			\frac{1}{\sqrt{n}}
			+\frac{1}{\sqrt{p}}
			+\sqrt{\frac{\log n\,\log p}{n}}
			+\frac{\log n\,\log p}{n}
			\right)
			\right).
		\end{align*}
		
		For the third term,
		\[
		\max_{j\in[p]}
		\left\|
		\bH\left(\frac{1}{n}\sum_{i=1}^{n}\bfv_i\bfv_i^{\top}-\bI_K\right)\bb_j
		\right\|
		\le
		\|\bH\|
		\left\|\frac{1}{n}\sum_{i=1}^{n}\bfv_i\bfv_i^{\top}-\bI_K\right\|
		\max_{j\in[p]}\|\bb_j\|
		=
		O_P(1)\,O_P(n^{-1/2})\,O_P(1).
		\]
		
		For the last term,
		\begin{align*}
			\max_{j\in[p]}
			\left\|
			\frac{1}{n}\sum_{i=1}^{n}\bH\bfv_i\,\delta_{ji}
			\right\|
			&\le
			\|\bH\|
			\max_{j\in[p]}
			\sqrt{\sum_{k=1}^{K}\left(\frac{1}{n}\sum_{i=1}^{n} f_{ik}\delta_{ji}\right)^2} \\
			&\le
			\|\bH\|
			\sqrt{\frac{1}{n}\sum_{i=1}^{n}f_{ik}^2}\,
			\max_{j\in[p]}\sqrt{\frac{1}{n}\sum_{i=1}^{n}\delta_{ji}^2}
			=
			O_P\!\left(\sqrt{\frac{\log n\,\log p}{n}}\right).
		\end{align*}
		
		Combining the bounds above gives
		\[
		\max_{j\in[p]}\bigl\|\wh{\bb}_j-\bH\bb_j\bigr\|
		=
		O_P\!\left(
		\sqrt{\frac{\log p}{n}}
		+\frac{1}{\sqrt{p}}
		+\sqrt{\frac{\log n\,\log p}{n}}
		+\frac{\log n\,\log p}{n}
		+\left(\frac{\log n\,\log p}{n}\right)^{3/2}
		\right).
		\]
	\end{proof}

\end{document}